\documentclass[11pt,a4paper]{article}

\usepackage{hyperref}
\usepackage{color}
\usepackage{amsthm}
\usepackage{mathtools}

\usepackage[normalem]{ulem} 
\usepackage{graphicx}
\usepackage{amssymb,latexsym,cite}
\usepackage{amsmath}
\usepackage{amsfonts}
\usepackage{mathrsfs}
\usepackage{bbm}
\usepackage{bm}
\usepackage[T1]{fontenc}

\usepackage[matrix,arrow,color]{xy}

\usepackage{enumitem}

\newenvironment{myitemize}{\begin{itemize}[itemsep=-0.05cm, leftmargin=*, topsep=0.1cm]}{\end{itemize}}
\newenvironment{myenumerate}{\begin{enumerate}[itemsep=-0.05cm, leftmargin=*, topsep=0.1cm, label={(\arabic*)}]}{\end{enumerate}}

\def\slasha#1{\setbox0=\hbox{$#1$}#1\hskip-\wd0\hbox to\wd0{\hss\sl/\/\hss}}

\def\periodb#1{\setbox0=\hbox{$#1$}#1\hskip-\wd0\hbox to\wd0{-}}

\newcommand{\unit}{\mathbbm{1}}   			
\newcommand{\ii}{\mathrm{i}}   			

\newcommand{\CA}{\mathcal{A}}    			

\newcommand{\CCA}{\mathscr{A}}

\newcommand{\CB}{\mathcal{B}}

\newcommand{\CC}{\mathcal{C}}
\newcommand{\CCC}{\mathscr{C}}

\newcommand{\CF}{\mathcal{F}}
\newcommand{\CCF}{\mathscr{F}}

\newcommand{\CL}{\mathcal{L}}

\newcommand{\CCM}{\mathscr{M}}

\newcommand{\CP}{\mathcal{P}}

\newcommand{\CR}{\mathcal{R}}

\newcommand{\CT}{\mathcal{T}}

\newcommand{\CV}{\mathcal{V}}

\newcommand{\CE}{\mathcal{E}}

\newcommand{\frg}{\mathfrak{g}}				

\newcommand{\frF}{\mathfrak{F}}

\newcommand{\fru}{\mathfrak{u}}

\newcommand{\frv}{\mathfrak{v}}
	
\newcommand{\frk}{\mathfrak{k}}	
\newcommand{\frt}{\mathfrak{t}}	
\newcommand{\frsu}{\mathfrak{su}}

\def\swone{{\textrm{\tiny (1)}}}

\def\hodge{{\textrm{\tiny H}}}

\def\dual{{\textrm{\tiny $\vee$}}}

\def\lact{{\textrm{\tiny L}}}
\def\ract{{\textrm{\tiny R}}}
\def\BV{{\textrm{\tiny BV}}}
\def\MC{{\textrm{\tiny MC}}}
\def\ce{{\textrm{\tiny CE}}}
\def\br{{\textrm{\tiny $\RR$}}}
\def\tCF{{\textrm{\tiny $\CF$}}}
\def\tCR{{\textrm{\tiny $\RR$}}}
\def\tW{{\textrm{\tiny $W$}}}
\def\tv{{\textrm{\tiny $V$}}}
\def\tv1{{\textrm{\tiny $V[1]$}}}
\def\tA{A}
\def\tsA{{\textrm{\tiny$\sA$}}}

\newcommand{\Sym}{\mathrm{Sym}}

\newcommand{\starcom}{\stackrel{\scriptstyle\star}{\scriptstyle,}}

\newcommand{\FR}{\mathbbm{R}}     			
\newcommand{\FC}{\mathbbm{C}}     			
\newcommand{\RZ}{\mathbbm{Z}}     			
\newcommand{\FS}{\mathbbm{S}}     			

\newcommand{\dd}{\mathrm{d}}     			

\newcommand{\sB}{\mathsf{B}}

\newcommand{\LL}{\mathrm{L}}     			

\newcommand{\aso}{\mathfrak{so}}

\newcommand{\sU}{\mathsf{U}}     			

\newcommand{\sA}{\mathsf{A}}
\newcommand{\sG}{\mathsf{G}}

\newcommand{\sSU}{\mathsf{SU}}

\newcommand{\sM}{\mathsf{M}}
\newcommand{\sD}{\mathsf{D}}

\newcommand{\sSO}{\mathsf{SO}}

\newcommand{\sSp}{\mathsf{Sp}}

\newcommand{\comment}[1]{}     				
\def\tyng(#1){\hbox{\tiny$\yng(#1)$}}			
\def\tyoung(#1){\hbox{\tiny$\young(#1)$}}			

\newcommand{\beq}{\begin{eqnarray}}
\newcommand{\eeq}{\end{eqnarray}}

\newcommand{\sfd}{{\sf d}}

\newcommand{\sff}{{\sf f}}

\newcommand{\sfGamma}{{\sf \Gamma}}
\newcommand{\Set}{{\sf Set}}
\newcommand{\dAlg}{{\sf dAlg}}

\newcommand{\sfR}{\mathsf{R}}
\newcommand{\sfP}{\mathsf{P}}
\newcommand{\sfM}{\mathsf{M}}

\definecolor{outrageousorange}{rgb}{1.0, 0.43, 0.29}

\newcommand{\vol}{{\rm vol}}

\newcommand{\Tr}{\mathrm{Tr}}

\theoremstyle{plain}
\newtheorem{theorem}[equation]{Theorem}
\newtheorem{corollary}[equation]{Corollary}

\newtheorem{proposition}[equation]{Proposition}

\theoremstyle{definition}
\newtheorem{definition}[equation]{Definition}

\newtheorem{example}[equation]{Example}

\newcommand{\midwedge}{\text{\Large$\wedge$}}

\newcommand{\dsf}{{\mathsf{d}}}

\def\RR{{\mathcal R}}

\def\beq{\begin{equation}}
\def\bee{\begin{equation}}
\def\eeq{\end{equation}}
\def\bea{\begin{eqnarray}}
\def\eea{\end{eqnarray}}
\def\ba{\begin{align}}
\def\ea{\end{align}}

\numberwithin{equation}{section}
\catcode`@=12

\begin{document}

\renewcommand{\thefootnote}{\fnsymbol{footnote}}

\begin{titlepage}
	
	\renewcommand{\thefootnote}{\fnsymbol{footnote}}
	
	\begin{flushright}
		\small
		{\sf EMPG--21--14}
	\end{flushright}
	
	\begin{center}
		
		\vspace{1cm}
		
		\baselineskip=24pt
		
		{\Large\bf Braided Symmetries in Noncommutative Field Theory}
		
		\baselineskip=14pt
		
		\vspace{1cm}
		
{\bf Grigorios Giotopoulos}${}^{\,(a),(b),}$\footnote{Email: \ {\tt
				grigorios.giotopoulos@nyu.edu}} \ \ and \ \  {\bf Richard
			J. Szabo}${}^{\,(b),(c),}$\footnote{Email: \ {\tt R.J.Szabo@hw.ac.uk}}
		\\[6mm]
		
		\noindent{${}^{(a)}$ {\it Mathematics, Division of Science, New York University Abu Dhabi, U.A.E.} }
\\[3mm]		
		
\noindent  {${}^{(b)}$ {\it Department of Mathematics, Heriot-Watt University\\ Colin Maclaurin Building,
			Riccarton, Edinburgh EH14 4AS, U.K.}}\\ and {\it Maxwell Institute for
			Mathematical Sciences, Edinburgh, U.K.} \\[3mm]
\noindent{${}^{(c)}$ {\it Higgs Centre
			for Theoretical Physics, Edinburgh, U.K.}}
		\\[30mm]
		
	\end{center}
	
	\begin{abstract}
		\noindent
We give a pedagogical introduction to $L_\infty$-algebras and their uses in organising the symmetries and dynamics of classical field theories, as well as of the conventional noncommutative gauge theories that arise as low-energy effective field theories in string theory. We review recent developments which formulate field theories with braided gauge symmetries as a new means of overcoming several obstacles in the standard noncommutative theories, such as the restrictions on gauge algebras and matter fields. These theories can be constructed by using techniques from Drinfel'd twist deformation theory, which we review in some detail, and their symmetries and dynamics are controlled by a new homotopy algebraic structure called a `braided $L_\infty$-algebra'. We expand and elaborate on several novel theoretical issues surrounding these constructions, and present three new explicit examples: the standard noncommutative scalar field theory (regarded as a braided field theory), a braided version of $BF$ theory in arbitrary dimensions (regarded as a higher gauge theory), and a new braided version of noncommutative Yang--Mills theory for arbitrary gauge algebras.
	\end{abstract}
	
\vspace{1cm}

\begin{center}
{\sl\small Contribution to the Special Issue of Journal of Physics A on `Noncommutative Geometry in Physics'}
\end{center}
	
\end{titlepage}

{\baselineskip=12pt
	\tableofcontents
}

\setcounter{footnote}{0}
\renewcommand{\thefootnote}{\arabic{footnote}}

\newpage

\section{Introduction and summary}
\label{sec:Intro}

In this paper we review some recent developments which offer a new perspective on symmetries in noncommutative field theory; in this paper, by a `noncommutative field theory' we will mean a field theory that is a \emph{deformation} of a classical field theory (and in almost all instances via a star-product on the relevant algebras of fields). To motivate and introduce these investigations, we begin with a flash review of one of the best understood natural occurrence of noncommutative geometry in a candidate theory of quantum gravity, namely open string theory on spacetimes with non-trivial background NS--NS ($\text{NS}=\text{Neveu--Schwarz}$) fields.

\subsubsection*{D-branes and noncommutative gauge theory}

Consider open string theory on flat spacetime in the background of a constant Kalb--Ramond $B$-field. With suitable boundary conditions the open strings are constrained to move in directions transverse to a submanifold $\FR^{1,p}\subseteq\FR^{1,9}$ which is the worldvolume wrapped by a stack of coincident flat D$p$-branes. It is a classic and remarkable result that the effects of the $B$-field on open string interactions in worldsheet conformal field theory correlation functions on the disk are completely captured by the noncommutative Moyal--Weyl star-product\footnote{Notation: $x=(x^\mu)$ are (local) coordinates on the worldvolume $\FR^{1,p}$, and $\partial_\mu=\frac\partial{\partial x^\mu}$ are vector fields of the corresponding (local) holonomic frame.}
\begin{align}\label{eq:MoyalWeylstarprod}
\begin{split}
f\star g &= \mu\circ \exp\big(\mbox{$\frac{\ii\,\hbar}2$}\,
\theta^{\mu\nu}\, \partial_\mu\otimes\partial_\nu\big)(f\otimes g) \\[4pt]
& = f\cdot g + \sum_{n=1}^\infty\, \Big(\frac{\ii\,\hbar}2\Big)^n \, \frac1{n!} \, \theta^{\mu_1\nu_1}\cdots\theta^{\mu_n\nu_n} \, \partial_{\mu_1}\cdots\partial_{\mu_n}f \cdot \partial_{\nu_1}\cdots\partial_{\nu_n}g \ ,
\end{split}
\end{align}
where $\mu(f\otimes g)=f\cdot g$ is the pointwise product of fields $f$ and $g$, $\hbar$ is treated as a formal deformation parameter\footnote{Recall that there are also convergent integral expressions for the Moyal--Weyl star-product, defined on the space of Schwartz functions on $\FR^{1,p}$, whose asymptotic series expansions in $\hbar$ coincide with \eqref{eq:MoyalWeylstarprod}, see e.g.~\cite{Szabo:2001kg}.} and the constant Poisson bivector $\theta=\frac12\,\theta^{\mu\nu}\,\partial_\mu\wedge\partial_\nu$ is related to the pullback of the $2$-form $B$ on $\FR^{1,9}$ to $\FR^{1,p}$~\cite{Douglas:1997fm, Ardalan:1998ce, Chu:1998qz, Schomerus:1999ug, Seiberg:1999vs}. In a suitable low-energy scaling limit, the boundary operator product expansion of tachyon vertex operators becomes exactly the Moyal--Weyl star-product.

The massless bosonic modes of the open strings consist of gauge fields $A=A_\mu\,\dd x^\mu$ on the worldvolume and scalar fields $X^{a}$ describing the transverse fluctuations of the D-branes. Their low-energy dynamics is described by a \emph{noncommutative gauge theory} on $\FR^{1,p}$ with gauge group $\sU(n)$, where $n$ is the number of coincident D$p$-branes. This noncommutative field theory is defined by deformation of the standard (dimensionally reduced) Yang--Mills action on $\FR^{1,p}\subseteq\FR^{1,9}$ via replacements of ordinary pointwise matrix products of fields by the Moyal--Weyl star-product \eqref{eq:MoyalWeylstarprod} composed with matrix multiplication. See~\cite{Douglas:2001ba,Szabo:2001kg} for extensive reviews of this subject.

This story extends to curved D-branes in flat space, where the Poisson structure $\theta$ is no longer constant and the Moyal--Weyl star-product \eqref{eq:MoyalWeylstarprod} is replaced by the more general Kontsevich star-product~\cite{Kontsevich:1997vb}, and also to curved backgrounds with a non-zero pullback of the NS--NS $3$-form flux $H=\dd B$ to the D-brane worldvolume~\cite{Cornalba:2001sm,Herbst:2001ai}. A standard class of examples comes from D-branes in Wess--Zumino--Witten (WZW) models~\cite{Alekseev:1999bs,Alekseev:2000fd}, and more recently as holographic dual gauge theories to backgrounds obtained via integrable deformations of string sigma-models whose target space is $\mathbbm{A}\mathbbm{d}\mathbbm{S}_5{\times}\FS^5$~\cite{vanTongeren:2015uha,Araujo:2017jkb}.

Despite their intensive investigation for over 20 years by various communities in physics and mathematics, there are still many open general problems surrounding the construction and description of these gauge theories on curved D-branes. Foremost amongst these problems is the structure of the algebra of star-gauge transformations. Consider for definiteness $\sU(1)$ noncommutative gauge theory on a manifold $M$ with the star-commutator
\begin{align}
[\lambda\stackrel{\scriptstyle\star}{\scriptstyle,} A] := \lambda\star A-A\star\lambda
\end{align}
among a gauge parameter $\lambda\in\Omega^0(M)=C^\infty(M,\FR)$ and a gauge field $A\in\Omega^1(M)$. Then the naive definition of a star-gauge transformation
\begin{align}
\delta_\lambda^\star A = \dd\lambda+[\lambda\stackrel{\scriptstyle\star}{\scriptstyle,} A] 
\end{align}
will generally obstruct the closure of the gauge algebra:
\begin{align}
\big[\delta^\star_{\lambda_1},\delta^\star_{\lambda_2}\big]^{\phantom{\dag}}_\circ A := \big(\delta^\star_{\lambda_1}\circ\delta^\star_{\lambda_2}-\delta^\star_{\lambda_2} \circ\delta^\star_{\lambda_1}\big)A \ \neq \ \delta^\star_{[\lambda_1\stackrel{\scriptstyle\star}{\scriptstyle,} \lambda_2]}A \ .
\end{align}

This problem is related to the failure of the Leibniz rule of the de~Rham differential $\dd$ for a general non-constant Poisson bivector $\theta$. For instance
\begin{align}
\dd(f\star g)\  \neq\  \dd f\star g + f \star \dd g
\end{align}
for functions $f,g\in C^\infty(M)$. This means that in general there is no well-defined noncommutative differential calculus on $M$ from which to obtain the standard ingredients of a field theory.

Both of these issues are related to the problem of firstly extending the corresponding Poisson algebra on $C^\infty(M,\FR)$ to the exterior algebra $\Omega^\bullet(M)$ of differential forms on $M$ in such a way as to extend the de~Rham complex $\big(\Omega^\bullet(M),\dd\big)$ to a differential graded Poisson algebra, and subsequently quantizing this to a noncommutative differential graded algebra. There is a fairly extensive literature of investigations into this problem through different approaches during the last two decades, most of which involve a choice of an auxiliary (symplectic) connection and work only at the semi-classical (Poisson) level without achieving the second step of quantization; see e.g.~\cite{Ho:2001fi}. Here we focus instead on recent approaches to this problem which are tied to more modern perspectives on classical and quantum field theory.

\subsubsection*{Homotopy algebras}

Higher structures in string theory and quantum field theory have become more and more prominent in the last few years. In this paper we study occurrences of `higher algebras', and in particular homotopy Lie algebras. Their role in physics makes nice contact with recent mathematical developments in deformation theory and derived geometry.
Homotopy algebraic techniques and concepts have also recently been shown to offer a natural arena and physically motivated tools for formulating the closure of the noncommutative gauge algebra, as well as constructing a differential graded algebra structure on the fields and their dynamics. As these developments are central to the discussion of this paper, let us offer a brief and somewhat sketchy history of the appearance of $L_\infty$-algebras in physics and mathematics, emphasising those occurrences which fit contextually into our treatment. More precise definitions and explanations will be given later on in this article.

The origins of $L_\infty$-algebras in physics can be traced back to early work from the 1980s on higher spin gauge theories~\cite{Berends:1984rq}, where it was observed that the closure of the gauge algebra takes place in an extended space of field dependent gauge parameters:
\begin{align}
\big[\delta_{\lambda_1},\delta_{\lambda_2}\big]^{\phantom{\dag}}_\circ \Phi = \delta_{C(\lambda_1,\lambda_2;\Phi)}\Phi \ ,
\end{align}
where here and below $\Phi$ symbolically denotes the collection of fields of the theory, and $\lambda$ the corresponding gauge parameters. The first full structural occurrence of $L_\infty$-algebras appeared in the early 1990s in the work of Zwiebach on closed string field theory~\cite{Zwiebach:1992ie}, where the generalized\footnote{Here and in the following by `generalized' we mean local symmetries of a field theory which do not necessarily correspond to infinitesimal automorphisms of a principal bundle formulation.} gauge symmetries were shown to involve infinitely-many higher bracket operations $\ell_n$ for $n\geq1$:
\begin{align}
\delta_\lambda\Phi=\ell_1(\lambda)+\sum_{n=1}^\infty\,\frac{(-1)^{n\choose 2}}{n!} \, \ell_{n+1}(\lambda,\Phi^{\otimes n}) \ .
\end{align}
Inspired by these algebraic structures from string field theory, it has been recently emphasised by~\cite{Hohm:2017pnh,BVChristian} that the dynamics of any perturbative\footnote{By `perturbative' here we mean theories with only polynomial interactions among the fields.} classical field theory with generalized gauge symmetries is naturally organised by an $L_\infty$-algebra.\footnote{Our conventions for the brackets $\ell_n$ agree with~\cite{Hohm:2017pnh}, whereas the operations $\mu_n$ of~\cite{BVChristian} are related to ours by a sign: \smash{$\mu_n=(-1)^{n\choose 2} \, \ell_n$}.} This observation is useful because it enables one to develop new techniques for studying perturbative quantum field theories in a unifed setting using the powerful tools of homotopical algebra, see e.g.~\cite{Macrelli:2019afx,Lopez-Arcos:2019hvg,Arvanitakis:2020rrk,Borsten:2021hua} for some recent applications.

The connection between homotopy algebras and field theories arises because $L_\infty$-algebras are the natural algebraic structure underlying the Batalin--Vilkovisky (BV) formalism~\cite{BVChristian}. Mathematically, this is due to the duality between $L_\infty$-algebras and differential graded commutative algebras~\cite{Lada:1992wc}. These are also useful tools in deformation theory. For example, Kontsevich's formality theorem~\cite{Kontsevich:1997vb}, which lies at the heart of deformation quantization of Poisson manifolds, is based on $L_\infty$-quasi-isomorphisms between differential graded Lie algebras, namely the Schouten--Nijenhuis algebra of multivector fields and the Gerstenhaber algebra of differential operators. 

These developments have recently inspired systematic descriptions of noncommutative gauge theories on curved D-branes based on $L_\infty$-algebras~\cite{Blumenhagen:2018kwq,Kupriyanov:2019ezf} (homotopy coherent weakenings of Lie algebras) and also on $P_\infty$-algebras~\cite{Kupriyanov:2021cws} (homotopy coherent weakenings of Poisson algebras). However, while these approaches are superior to earlier ones in that they do not involve auxiliary choices such as a symplectic connection, they have also so far not been understood beyond the semi-classical level. The purpose of this paper is to review more recent developments which use homotopy algebras to go beyond the Poisson approximation, and to use these techniques to present new examples of noncommutative gauge theories  that are based on fully fledged noncommutative differential graded algebras for the first time.

\subsubsection*{Non-geometric backgrounds and noncommutative gravity}

The new homotopy algebraic techniques which we review are in part inspired by much less concrete and understood appearances of noncommutative geometry in \emph{closed} string theory. On general grounds, it is expected that noncommutative theories of gravity should serve as low-energy limits of any theory of quantum gravity, retaining some features of a quantized spacetime. In particular, in certain non-geometric flux backgrounds, it has been conjectured that the low-energy dynamics of closed strings may be described by a theory of noncommutative gravity~\cite{Blumenhagen:2010hj,Lust:2010iy,Mylonas:2012pg} (see~\cite{Szabo:2018hhh} for a review). One of the hopes here is that noncommutative geometry may serve as a more precise description of the notion of `non-geometry'; see e.g.~\cite{Hull:2019iuy} for how this works in the \emph{open} string sector, and~\cite{Aschieri:2020uqp} for a recent rigorous perspective, where the target space geometry probed by D-branes is a noncommutative torus bundle. This is difficult to check explicitly, physically because the low-energy limits are much harder to control as compared to the open string sector, and mathematically because the metric aspects of noncommutative differential geometry are only partially developed so far~\cite{TwistApproach,Aschieri:2005zs,Blumenhagen:2016vpb,NAGravity,Aschieri:2020ifa}. In particular, no general version of the Einstein--Hilbert action functional is currently known.

This has led to attempts to treat noncommutative gravity as a deformation of a `gauge theory', that is, in the first order formalism for general relativity: One uses the Einstein--Cartan principal bundle formulation, with star-gauge symmetry, where the corresponding variational principle is based on the Palatini action functional~\cite{Chamseddine:2000si,Cardella:2002pb,AschCast,Barnes:2016cjm,Ciric:2016isg}. However, diffeomorphisms can never be implemented as star-gauge symmetries, in any formalism, and instead the noncommutative theory of gravity is invariant under a `twisted' action of infinitesimal diffeomorphisms~\cite{TwistApproach,Aschieri:2005zs}. This dichotomy between the local Lorentz and diffeomorphism symmetries presents a serious obstruction to the formulation of noncommutative gravity for general backgrounds in the setting of homotopy algebras, because the twisted diffeomorphism symmetry does not appear to fit naturally into the standard $L_\infty$-algebra picture of star-gauge symmetries. 

These issues prompt another approach. Firstly, we resolve the dichotomy between the different types of symmetries by letting the local Lorentz rotations instead act in a `twisted' manner, like the diffeomorphisms, which replaces the notion of star-gauge symmetry with \emph{braided gauge symmetry}. Secondly, we do not attempt to package this symmetry into a standard $L_\infty$-algebra, which indeed does not seem possible to do in a closed or algorithmic manner, but rather deform the underlying $L_\infty$-algebra itself to make it compatible with the new type of symmetry. This results in a systematic construction of new theories of noncommutative gravity, which were formulated in detail by~\cite{Ciric:2020eab,Ciric:2021rhi}.

\subsubsection*{Summary and outlook}

The procedure discussed above works without any reference to gravity (even though this is what motivated the construction), and leads to both a new symmetry principle for noncommutative field theories from a physics perspective, as well as a new homotopy algebraic structure from a mathematics perspective, that was called a \emph{braided $L_\infty$-algebra} in~\cite{Ciric:2020eab,Ciric:2021rhi}. Noncommutative field theories whose symmetries and dynamics are organised by braided $L_\infty$-algebras are correspondingly called \emph{braided field theories}. The purpose of this contribution is to review the definition of noncommutative field theories with braided gauge symmetries, together with their formulation through a general notion of braided ${L_\infty}$-algebras, and to explain how the mathematical machinery constructs novel examples with exciting new features. We further expand and elucidate various details of the general theory which were not addressed in~\cite{Ciric:2021rhi}. In addition to reviewing the basic example of braided Chern--Simons theory from~\cite{Ciric:2021rhi}, we present three new examples: the standard noncommutative scalar field theory (regarded as a braided field theory), braided $BF$ theory in any dimension (which serves as an elementary example of a higher gauge theory), and a new braided version of noncommutative Yang--Mills theory for an arbitrary gauge algebra. The reader interested in the applications to noncommutative gravity is invited to consult~\cite{Ciric:2021rhi} where a detailed exposition can be found.

We should stress that the notion of `braided gauge symmetry' that we review here is not really new: kinematical aspects of this idea have appeared long before (see e.g.~\cite{BraidedGauge}), while some concepts and techniques are borrowed from twisted noncommutative gravity, for which there is a fairly extensive literature. What is new however is the application of these techniques to ordinary (non-gravitational) gauge theories, the full treatment of braided noncommutative dynamics through the organising principle of a braided $L_\infty$-algebra, and the new examples of noncommutative field theories with braided symmetries that this produces. We further emphasise that the braided noncommutative gauge and gravity theories considered here and in~\cite{Ciric:2021rhi} are \emph{not} the same as the `twisted' noncommutative  gauge and gravity theories considered in the mid-2000s (see e.g.~\cite{TwistApproach,Aschieri:2006ye}). In particular, braided symmetries are based on braided derivations closing a braided Lie algebra, whereas twisted symmetries are based on a classical Lie algebra and a deformed Leibniz rule following from twist deformation of the coproduct of the associated enveloping Hopf algebra.

There are several obvious open questions which we do not address in this contribution, as these are still under development in this very new subject. For instance, what do the corresponding quantum field theories look like? The quantization of our noncommutative field theories are related to Oeckl's symmetric braided quantum field theory~\cite{Oeckl:1999zu,Oeckl:2000eg}, which however does not consider theories with gauge symmetries. Preliminary investigations in this direction within the context of fuzzy field theories are found in~\cite{Nguyen:2021rsa}, which should be related to the lattice regularization of braided quantum field theory discussed in~\cite{DAdda:2009vrc}. A simple example of a continuum field theory with gauge symmetry is under development in~\cite{Serbiainprep}.

Another immediate question concerns the explicit realizations of our braided noncommutative field theories in string theory (or in other theories of quantum gravity). For this, one could consider the Hopf algebraic symmetries of string worldsheet conformal field theory correlation functions discussed by~\cite{Asakawa:2008cc,Asakawa:2008nu}, and then mimick the route used derive the standard noncommutative field theories. This approach should be made feasible by comparing with the Ward--Takahashi identities of braided quantum field theory developed in~\cite{Sasai:2007me}. Such a development then further has the exciting potential of providing the very first precise derivation of a (braided) noncommutative theory of gravity from string theory, bypassing previous no-go results concerning the realisation of twisted diffeomorphisms as symmetries of string theory in the low-energy limit~\cite{Alvarez-Gaume:2006qrw,Blumenhagen:2016vpb}. This will rely on a deeper understanding of the Hopf algebra symmetry of string theory with a $B$-field background, which can be understood by viewing quantisation of the string using Drinfel'd twist deformations~\cite{Asakawa:2008nu}. Braided diffeomorphisms which are consistent with the closed string sector are necessarily different from their actions in the low-energy effective theory on D-branes, as the latter is known to be insufficient for describing symmetries in brane-induced gravity~\cite{Alvarez-Gaume:2006qrw}. The twisted Hopf algebras of~\cite{Asakawa:2008nu} enable extensions to couplings of D-branes to the closed string sector, or even to closed string amplitudes. Braided noncommutative deformations of string theory are also argued by~\cite{vanTongeren:2015uha} to underlie the AdS/CFT dual gauge theories to Yang--Baxter deformations of $\mathbbm{AdS}_5{\times}\FS^5$ string sigma-models. Understanding these approaches better may realise how the noncommutative geometry probed by D-branes carries over naturally to the bulk closed string sector.

\subsubsection*{Outline of the paper}

Throughout this paper we have endeavoured to give a pedagogical and relatively self-contained presentation of both the relevant physical and mathematical concepts. We provide a fairly in-depth review of $L_\infty$-algebras and their uses in formulating classical field theory, and how the formalism immediately adapts to some standard examples of conventional noncommutative gauge theories. We review only the salient mathematical aspects of homotopy algebras and Drinfel'd twist deformation techniques that we need, but formulated in a language that we hope is palatable to both novices and experts alike. We elaborate and expand on various aspects of the notion of braided $L_\infty$-algebra that was introduced by~\cite{Ciric:2021rhi}, as well as of their applications to braided field theories. In particular, we expand the repertoire of examples presented in~\cite{Ciric:2021rhi} beyond the realm of diffeomorphism invariant theories.

We begin in Section~\ref{sec:Linfty} by giving a fairly detailed review of the $L_\infty$-algebra formulation of classical field theories, starting with a pedagogical review of gauge symmetries through the prototypical example of Chern--Simons gauge theory. We then explain the notion of an $L_\infty$-algebra and its connection with the BV formalism, in a way that is immediately adaptable to the braided generalizations that we consider later on in the paper. We further elucidate on the role of higher gauge transformations in field theories with reducible gauge symmetries, and in particular their dual notion of `higher  Noether identities' which are not usually stressed in the literature, but which are of paramount importance in braided field theory, as we discuss. We also discuss the basic example of Yang--Mills theory whose homotopy algebraic structure goes beyond the differential graded Lie algebra organising Chern--Simons theory.

In Section~\ref{sec:NCLinfty} we revisit some standard examples of noncommutative gauge theories which arise as low-energy effective field theories on D-branes in string theory, and how the classical $L_\infty$-algebra formulation easily adapts to them using the examples considered in Section~\ref{sec:Linfty}. These include the conventional noncommutative field theories on flat Minkowski spacetime with Moyal--Weyl star-gauge symmetry, as well as the fuzzy field theories supported by curved D-branes in WZW models, treating the examples of gauge theories on the fuzzy sphere in detail. We further discuss the issues involved in adapting the standard $L_\infty$-algebra picture to generic examples of noncommutative gauge theories on D-branes in curved backgrounds, which serves as a (partial) motivation for moving outside of the framework of classical homotopy algebras.

With the classical framework and its limitations in noncommutative field theory understood, we proceed in Section~\ref{sec:braidedLinfty} to our main topic of braided gauge symmetry and the accompanying notion of a braided $L_\infty$-algebra. We start by reviewing some well-known issues with the closure of standard noncommutative gauge transformations and diffeomorphisms, which motivates the use of Drinfel'd twist deformation techniques that we explain in detail. These methods naturally lead to the notions of braided gauge symmetry and associated noncommutative kinematics. Like we did in the classical case, we illustrate the novel ensuing properties through the simplest example of a braided version of Chern--Simons gauge theory, which  inspires a natural braided deformation of the definition of an $L_\infty$-algebra. We expand on several technical points surrounding the mathematical presentation of braided $L_\infty$-algebras in~\cite{Ciric:2021rhi}, and their applications to the unambiguous construction of noncommutative field theories, such as their reality properties. Following~\cite{Nguyen:2021rsa} we explain the relation to a braided version of the BV formalism, and we further discuss some preliminary ideas on how to make sense of moduli spaces of classical solutions in braided field theories, which seems to resolve some paradoxical observations in a way that is common to the usual thinking in noncommutative geometry, but which requires a great deal of further technical development to make precise.

In the remaining three sections we proceed to analyse new explicit examples in the braided $L_\infty$-algebra formalism. We revisit the standard noncommutative scalar field theory with polynomial interactions in Section~\ref{sec:NCscalar} and show that it is naturally encoded in a braided $L_\infty$-algebra. As a \emph{braided} field theory, this example exhibits a number of interesting features compared to its incarnation in the usual noncommutative field theory framework. In particular, we address the well-known problem of implementation of global symmetries in noncommutative field theory and show that they too have a natural realisation in the language of braided $L_\infty$-algebras, treating the traditional example of Lorentz transformations in detail.

In Section~\ref{sec:braidedBFtheory} we turn to the natural extensions of Chern--Simons gauge theory to arbitrary dimensions provided by $BF$ theories. After reviewing their classical $L_\infty$-algebra formulation, we describe their braided noncommutative deformations following the braided $L_\infty$-algebra prescription. These theories provide simple illustrative examples of gauge symmetries that are reducible, which are naturally described in the $L_\infty$-algebra framework. We discuss the conceptual issues surrounding the braided versions of these higher gauge symmetries, and how they may be resolved along the lines of our sketch of the description of the classical moduli spaces of braided field theories.

Finally, in Section~\ref{sec:braidedNCYM} we proceed to our main example. By twist deforming the $L_\infty$-algebra structure of classical Yang--Mills theory, we obtain a new braided version of noncommutative Yang--Mills theory for arbitrary gauge algebras. Unlike the examples of scalar field theory, or the topological gauge theories which are based on differential graded Lie algebras, the braided noncommutative Yang--Mills theory does not follow the obvious deformation of its classical counterpart, like it does in conventional noncommutative field theories. We derive the modified field equations, Noether identities and action functional in detail, and we further discuss the implementation of global symmetries of the theory in the braided $L_\infty$-algebra framework. We provide all necessary details that we hope will inspire further investigations of this somewhat complicated but interesting gauge theory, and in particular a detailed comparison with its (unbraided) standard noncommutative counterpart, but we leave this for future work. Of particular immediate interest are basic questions concerning the corresponding braided quantum field theory, such as unitarity and the behaviour of perturbation theory.

\subsubsection*{Acknowledgments}

We thank Marija Dimitrijevi\'c \'Ciri\'c, Larisa Jonke, Voja Radovanovi\'c, Alexander Schenkel and Francesco Toppan for helpful discussions and correspondence. {\sc R.J.S.} thanks the editors Paolo Aschieri, Edwin Beggs, Francesco D'Andrea, Emil Prodan and Andrzej Sitarz for the invitation to contribute to this special issue.
 The work of {\sc G.G.} was supported by a
Doctoral Training Grant from the UK Engineering and Physical Sciences
Research Council. The work of {\sc R.J.S.} was supported by
the Consolidated Grant ST/P000363/1 
from the UK Science and Technology Facilities Council.

\section{Gauge symmetry and $L_\infty$-algebras}
\label{sec:Linfty}

In this section we will give a basic and pedagogical introduction to the $L_\infty$-algebra formulation of classical field theories with generalized gauge symmetries, reviewing the necessary mathematical concepts which we will later generalize to the case of noncommutative field theories with braided gauge symmetries.

\subsection{What is a gauge symmetry?}
\label{sec:whatisgauge}

To motivate the uses of $L_\infty$-algebras in classical field theory, as well as to highlight some salient features of the notion of gauge symmetry that we will need later on, we begin with a simple prototypical model which illustrates  most essential features of more general field theories, like the ones discussed in Section~\ref{sec:Intro}. We consider the well-known example of Chern--Simons theory on a closed oriented three-dimensional manifold $M$. Let $\frg$ be a quadratic Lie algebra, that is, $\frg$ is a real vector space equipped with a Lie bracket $[-,-]_\frg:\frg\wedge\frg\to\frg$ and a non-degenerate bilinear pairing $\Tr_\frg:\frg\otimes\frg\to\FR$ which is invariant under the natural adjoint action of $\frg$ on itself:
\begin{align}\label{eq:adginv}
\Tr_\frg\big([X,Z]_\frg\otimes Y + Z\otimes [X,Y]_\frg\big) = 0 \ ,
\end{align}
for all $X,Y,Z\in\frg$. We denote by $\Omega^\bullet(M,\frg):=\Omega^\bullet(M)\otimes\frg$ the graded Lie algebra of $\frg$-valued differential forms on $M$ with Lie bracket given by the tensor product of the Lie bracket of $\frg$ with exterior multiplication; by a slight abuse of notation we also denote this extended Lie bracket by~$[-,-]_\frg$. The de~Rham differential acts non-trivially only on the form part of $\Omega^\bullet(M,\frg)$, and we denote $\dd\otimes\unit$ simply by $\dd$ here and in the following.

The Chern--Simons action functional $S:\Omega^1(M,\frg)\to\FR$ for a gauge field $A\in\Omega^1(M,\frg)$ is defined by
\begin{align}\label{eq:CSaction}
S(A)=\int_M\, \Tr_\frg\Big(\frac12\,A\wedge \dd A +\frac1{3!}\,A\wedge[A,A]_\frg \Big) \ .
\end{align}
Despite its seemingly non-covariant dependence on $A$, this action functional is invariant under the standard infinitesimal gauge transformations
\begin{align}\label{eq:CSgaugevar}
\delta_\lambda A=\dd\lambda+[A,\lambda]_\frg
\end{align}
for $\lambda \in \Omega^{0}(M,\frg)$, that is, $\delta_\lambda S(A)=0$. Using the graded derivation property of the de~Rham differential $\dd$ with respect to the Lie bracket $[-,-]_\frg$ and the Jacobi identity for $[-,-]_\frg$, one easily verifies that the gauge variations close off-shell to the Lie algebra $\Omega^0(M,\frg)$:
\begin{align}
\big[\delta_{\lambda_1},\delta_{\lambda_2}\big]_\circ A := \big(\delta_{\lambda_1}\circ\delta_{\lambda_2} - \delta_{\lambda_2}\circ\delta_{\lambda_1}\big)A = \delta_{[\lambda_1,\lambda_2]_\frg}A \ .
\end{align}

Varying \eqref{eq:CSaction} with respect to arbitrary variations $\delta A$ of the gauge fields, one easily derives the corresponding Euler--Lagrange equations $\frF_{\tA}=0$, which are given in terms of the curvature of the connection $A$ as
\begin{align}\label{eq:CSeom}
\frF_{\tA} = F_{\tA} := \dd A + \tfrac12\,[A,A]_\frg = 0 \ \in \ \Omega^2(M,\frg) \ .
\end{align}
The critical locus of the Chern--Simons functional thus consists of flat $\frg$-connections on the three-manifold $M$.
Unlike the gauge field $A$, the equations of motion transform covariantly under gauge transformations by virtue of the covariance of the curvature $2$-form $F_{\tA}$:
\begin{align}
\delta_\lambda \frF_{\tA} = -[\lambda,\frF_{\tA}]_\frg \ .
\end{align}

This implies that the gauge symmetry acts on the classical solutions of Chern--Simons theory because
\begin{align}
\frF_{A+\delta_\lambda A}=\frF_{\tA}+\delta_\lambda \frF_{\tA}+O(\lambda^2) \ ,
\end{align}
and so the first order gauge variation of a solution $A$ to $\frF_{\tA}=0$ produces another solution $A+\delta_\lambda A$. Thus there are gauge redundancies in the description of the classical degrees of freedom, and the physical states are gauge orbits of classical solutions. The space of physical states of the classical field theory is therefore the moduli space of classical solutions modulo gauge transformations, which in this case is the moduli space $\CCM_{\rm flat}(M,\frg)$ of flat $\frg$-connections on $M$.\footnote{For the time being we do not go into any technical details on how one should make rigorous sense of such a moduli space, which is possible but we are solely interested in a motivational discussion here. We will return to a more precise description of this and other related moduli spaces in Section~\ref{sec:braidedMC}.}

What is important for us is an equivalent but perhaps less widely appreciated perspective on the gauge redundancies. The key observation, that has already been implicitly used, is that conventional gauge transformations $\delta_\lambda A$ are just special instances of general field variations $\delta A$ along specified directions in the tangent space to the space of fields $\Omega^1(M,\frg)$. The gauge variation of the action functional \eqref{eq:CSaction} for an arbitrary gauge parameter $\lambda\in\Omega^0(M,\frg)$ can be computed as
\begin{align}\label{eq:deltaCS}
\delta_\lambda S(A)=\int_M\,\Tr_\frg\big(\delta_\lambda A\wedge \frF_{\tA}\big) = - \int_M\,\Tr_\frg\big(\lambda\wedge \dd^{\tA}\frF_{\tA}\big) \ ,
\end{align}
where
\begin{align}
\dd^{\tA} \frF_{\tA}=\dd \frF_{\tA}+[A,\frF_{\tA}]_\frg \ \in \ \Omega^3(M,\frg)
\end{align}
is the gauge covariant derivative of the Euler--Lagrange equation $\frF_{\tA}=F_{\tA}$. The first equality follows from an integration by parts as well as ${\rm ad}(\frg)$-invariance \eqref{eq:adginv} using the Leibniz rule for the gauge variations $\delta_\lambda$, and is the same calculation one does when deriving the equations of motion \eqref{eq:CSeom}. The second equality similarly follows using \eqref{eq:CSgaugevar}. 

Of course, we know that
\begin{align}
\dd^{\tA}\frF_{\tA}=\dd^{\tA}F_{\tA}=0
\end{align}
by the standard Bianchi identity for the curvature $F_{\tA}$, and this is what guarantees the gauge invariance of the action functional. On the other hand, if we didn't know about the Bianchi identity, we would have deduced it from the requirement of gauge invariance of the Chern--Simons functional. Thus the Bianchi identity $\dd^{\tA}\frF_{\tA}=0$ is \emph{equivalent} to gauge invariance of the action functional $\delta_\lambda S(A)=0$ for all $\lambda\in\Omega^0(M,\frg)$.

What we have just explained is a simple example of the more general statement of
\begin{theorem}[{\bf Noether's second theorem}]
Gauge symmetries of a classical field theory are in a one-to-one correspondence with differential {identities} among its Euler--Lagrange derivatives.
\end{theorem}
The proof is a generalization of the computation in \eqref{eq:deltaCS}. We stress that Noether's second theorem asserts the existence of \emph{identities}, which hold off-shell, i.e. when $\frF_{\tA}\neq0$.\footnote{This is not to be confused with the more commonly known Noether's \emph{first} theorem, which is an on-shell statement asserting the existence of weakly conserved currents (and corresponding weakly conserved charges) for every continuous global symmetry of a classical field theory.} 
These identities are called \emph{Noether identities} and they exhibit the interdependence of the classical degrees of freedom due to the gauge symmetries, offering a different perspective on the gauge redundancies of the field theory. In particular, the converse of Noether's second theorem is a means of recovering gauge symmetries of an action functional $S$ which may be unknown \emph{a priori}. Notice that Chern--Simons theory is a special case where the geometric Bianchi identity coincides with the Noether identity, but in general Noether identities do not coincide with Bianchi identities.

Let us now recap what mathematically went into the description of classical Chern--Simons theory. To describe the classical moduli space we relied on two ingredients:
\begin{myitemize}
\item The graded vector space
\begin{align}
V=\Omega^\bullet(M,\frg) = V^0\oplus V^1\oplus V^2\oplus V^3 \ ,
\end{align}
where $V^p=\Omega^p(M,\frg)$ is the space of $\frg$-valued $p$-forms on the three-manifold $M$. The homogeneous subspace for $p=0$ contains the gauge parameters $\lambda$, for $p=1$ the fields $A$, for $p=2$ the field equations $\frF_{\tA}$, and for $p=3$ the Noether identities $\dd^{\tA}\frF_{\tA}=0$.
\item The `brackets'
\begin{align}\label{eq:CSbrackets}
\ell_1 = \dd \qquad \mbox{and} \qquad \ell_2=-[-,-]_\frg \ .
\end{align}
The linear map $\ell_1:V\to V$ of degree~$1$ is a differential making $V$ into a cochain complex, which in this case is just the de~Rham complex $\big(\Omega^\bullet(M,\frg),\dd\big)$ of $\frg$-valued differential forms on $M$. The bilinear map $\ell_2:V\otimes V\to V$ of degree~$0$ is a graded Lie bracket on $V$, i.e. it is graded antisymmetric and satisfies the graded Jacobi identity, and $\ell_1$ is a derivation of $\ell_2$. With these brackets, the gauge transformations, field equations and Noether identites can be written in the form
\begin{align}
\begin{split}
\delta_\lambda A = \ell_1(\lambda) + \ell_2(\lambda,A) & \qquad , \qquad \frF_{\tA} = \ell_1(A) - \tfrac12\,\ell_2(A,A) \ , \\[4pt] \dd^{\tA} \frF_{\tA} =& \ \ell_1(\frF_{\tA})+\ell_2(\frF_{\tA},A) \ ,
\end{split}
\end{align}
with gauge closure and covariance represented by
\begin{align}
\big[\delta_{\lambda_1},\delta_{\lambda_2}\big]_\circ A = \delta_{-\ell_2(\lambda_1,\lambda_2)}A \qquad \mbox{and} \qquad \delta_\lambda \frF_{\tA} = \ell_2(\lambda,\frF_{\tA}) \ .
\end{align}
\end{myitemize}

In order to derive the dynamics and Noether identities from a variational principle, we further required:
\begin{myitemize}
\item The `cyclic inner product'
\begin{align}
\langle\alpha,\beta\rangle = \int_M\,\Tr_\frg(\alpha\wedge\beta)
\end{align}
on $V$, which pairs differential forms $\alpha$ and $\beta$ on $M$ in complementary degrees; in other words, $\langle-,-\rangle:V\otimes V\to\FR$ is a bilinear map of degree $-3$, and so is non-vanishing only on the homogeneous subspaces $V^0\otimes V^3$ and $V^1\otimes V^2$. With it, the action functional is constructed by pairing a gauge field $A\in V^1$ with $\ell_1(A)\in V^2$ and $\ell_2(A,A)\in V^2$ through
\begin{align}\label{eq:CSactionell}
S(A) = \tfrac12\,\langle A,\ell_1(A)\rangle - \tfrac1{3!}\,\langle A,\ell_2(A,A)\rangle \ .
\end{align}
The cyclic property in this case is just the statement of $\ell_1$-invariance of the pairing, which follows from integration by parts, and of ${\rm ad}(\frg)$-invariance  with respect to the brackets \eqref{eq:CSbrackets}, both of which were used when varying \eqref{eq:CSactionell} to derive the equations of motion as well as the Noether identities.
\end{myitemize}

The first two items simply say that the triple $\big(\Omega^\bullet(M,\frg),\ell_1,\ell_2)$ is a differential graded Lie algebra, while the third item equips this with a cyclic structure. In other words, the symmetries and dynamics of Chern--Simons gauge theory are governed in an algebraic way by the objects
and relations of a \emph{cyclic differential graded Lie algebra}. This is the simplest prototypical example of a more general statement: Any perturbative classical field theory with generalized gauge symmetries is organised by a \emph{cyclic $L_\infty$-algebra}~\cite{Hohm:2017pnh,BVChristian}. In the remainder of this section we will explain this general statement in detail, beginning with an introduction to the mathematical concepts involved.

\subsection{What is an $L_\infty$-algebra?}
\label{sec:whatisLinfty}

We begin with a simple and concise version of the notion of an `$L_\infty$-algebra', following~\cite{Lada:1992wc,LadaMarkl94}, which also nicely ties in with its field theory realisations. Let $V=\bigoplus_{k\in\RZ}\,V^k$ be a $\RZ$-graded real vector space with (reduced) graded symmetric algebra $\Sym(V[1]):=\bigoplus_{n\geq1}\,V[1]^{\odot n} $, where $V[1]$ is the same underlying vector space as $V$ but with the degrees of its homogeneous subspaces shifted by~$1$, i.e.~$V[1]^k:=V^{k+1}$ for all $k\in\RZ$. Then $\Sym(V[1])$ can be regarded as a free cocommutative coalgebra with coproduct $\Delta_{\textrm{\tiny$V$}}:\Sym(V[1])\to\Sym(V[1])\otimes\Sym(V[1])$ defined on homogeneous subspaces by
\begin{align}\label{eq:DeltaV}
\begin{split}
\Delta_{\textrm{\tiny$V$}}\big|_{V[1]^{\odot n}} := \sum_{i=1}^{n-1} \ \sum_{\sigma\in{\rm Sh}(i;n-i)} \, \big(\unit^{\odot i}\otimes\unit^{\odot n-i}\big)\circ\tau^\sigma \ ,
\end{split}
\end{align}
where ${\rm Sh}(i;n-i)\subset S_n$ is the set of $(i;n-i)$-shuffled permutations of degree $n$, which are ordered as $\sigma(1)<\cdots<\sigma(i)$ and $\sigma(i+1)<\cdots<\sigma(n)$, and $\tau^\sigma:V[1]^{\otimes n}\to V[1]^{\otimes n}$ denotes the action of the permutation $\sigma$ via the trivial transposition isomorphism $\tau$ times the Koszul sign
multiplication which interchanges factors in a tensor product of graded vector spaces. 

\begin{definition}\label{def:coalgebra}
An \emph{$L_\infty$-algebra} is a graded vector space $V$ together with a coderivation
\begin{align}
D:\Sym(V[1])\longrightarrow\Sym(V[1])
\end{align}
of degree~$1$ which is a differential, that is, $D^2=0$. 
\end{definition}

The coderivation property
\begin{align}
\Delta_{\textrm{\tiny$V$}}\circ D = (D\otimes\unit + \unit\otimes D)\circ\Delta_{\textrm{\tiny$V$}}
\end{align}
implies that the differential $D:\Sym(V[1])\to \Sym(V[1])$ is completely determined by its `Taylor components' defined by projecting its image to the generating subspace to give a map ${}_{\textrm{\tiny$V[1]$}} D:\Sym(V[1])\to V[1]$ with the expansion
\begin{align}
{}_{\textrm{\tiny$V[1]$}} D=\sum_{n\geq1} \, b_n \ ,
\end{align}
where each $b_n:V[1]^{\odot n}\to V[1]$ is a graded symmetric multilinear map of degree~$1$. We can restore the original grading by introducing the `suspension' isomorphism $s:V\to V[1]$ which is simply the identity map but decreases the grading by~$1$, and defining multilinear graded antisymmetric maps $\ell_n:\midwedge^nV\to V$ of degree~$2-n$ by
\begin{align}
\ell_n := s^{-1}\circ b_n \circ s^{\otimes n} \ ,
\end{align}
for each $n\geq1$. The maps $\ell_n$ are called the \emph{$n$-brackets} of the $L_\infty$-algebra.

Graded antisymmetry is encoded by
\begin{align}
\ell_n(v_1,\dots,v_n) = -(-1)^{|v_i|\,|v_{i+1}|} \ \ell_n(v_1,\dots,v_{i-1},v_{i+1},v_i,v_{i+2},\dots,v_n)
\end{align}
for $i=1,\dots,n-1$, where $|v|$ denotes the degree of a homogenous element $v\in V$. Nilpotency $D^2=0$ translates into an infinite set of identities\footnote{The peculiar extra signs here and in the following are due to shifting degree and using \smash{$(s^{-1})^{\otimes n}\circ s^{\otimes n}=(-1)^{n\choose 2} \, \unit$} (see e.g.~\cite{BVChristian}).}
\begin{align}\label{eq:homotopyJacobi}
\sum_{i=1}^n \, (-1)^{i\,(n-i)} \, \ell_{n-i+1}\circ\big(\ell_{i}\otimes\unit^{\otimes n-i}\big) \ \circ \ \sum_{\sigma\in{\rm Sh}(i;n-i)} \, {\rm sgn}(\sigma) \, \tau^\sigma = 0
\end{align}
for each $n\geq1$. These identities are called the \emph{homotopy Jacobi identities} of the $L_\infty$-algebra.

Let us unravel the first few homotopy Jacobi identities to get an idea of their meaning. For $n=1$ the identity \eqref{eq:homotopyJacobi} reads
\begin{align}\label{eq:homotopyn=1}
(\ell_1)^2 = 0 \ ,
\end{align}
which states that the linear map $\ell_1:V\to V$ is a differential. Thus every $L_\infty$-algebra has an underlying cochain complex $(V,\ell_1)$. The identity \eqref{eq:homotopyJacobi} for $n=2$ states that $\ell_2:V\otimes V\to V$ is a cochain map. This is equivalent to saying that the differential $\ell_1$ is a graded derivation of the $2$-bracket $\ell_2$; explicitly, when evaluated on elements $v_1,v_2\in V$:
\begin{align}\label{eq:homotopyn=2}
\ell_1\big(\ell_2(v_1,v_2)\big) = \ell_2\big(\ell_1(v_1),v_2\big) + (-1)^{|v_1|} \, \ell_2\big(v_1,\ell_1(v_2)\big) \ .
\end{align}
For $n=3$ the identity \eqref{eq:homotopyJacobi} states that the $2$-bracket $\ell_2$ obeys the Jacobi identity up to the cochain homotopy $\ell_3:V\otimes V\otimes V\to V$; explicitly, when evaluated on elements $v_1,v_2,v_3\in V$:
\begin{align}\label{eq:homotopyn=3}
\begin{split}
& \ell_2\big(\ell_2(v_1,v_2),v_3\big) - (-1)^{|v_2|\,|v_3|}\,
  \ell_2\big(\ell_2(v_1,v_3),v_2\big) +  (-1)^{(|v_2|+|v_3|)\,|v_1|}\,
  \ell_2\big(\ell_2(v_2,v_3),v_1\big) \\[4pt]
& \hspace{1cm} = -\ell_3\big(\ell_1(v_1),v_2,v_3\big) - (-1)^{|v_1|}\,
  \ell_3\big(v_1, \ell_1(v_2), v_3\big) - (-1)^{|v_1|+|v_2|}\,
  \ell_3\big(v_1,v_2, \ell_1(v_3)\big) \\
& \hspace{4cm} -\ell_1\big(\ell_3(v_1,v_2,v_3)\big) \ .
\end{split}
\end{align}

This continues in general for all $n\geq4$ to higher coherence conditions for the $n$-brackets. In the special instance when $\ell_n=0$ for all $n\geq3$, an $L_\infty$-algebra is simply a differential graded Lie algebra. Thus $L_\infty$-algebras are homotopy coherent generalizations of Lie algebras; for this reason, $L_\infty$-algebras are also sometimes called (\emph{strong}) \emph{homotopy Lie algebras}. In particular, the cohomology $H^\bullet(V,\ell_1)$ of an $L_\infty$-algebra is a $\RZ$-graded Lie algebra.

The natural notion of an invariant inner product on an $L_\infty$-algebra is called a \emph{cyclic pairing}. 
\begin{definition}
A \emph{cyclic $L_\infty$-algebra} is an $L_\infty$-algebra $(V,\{\ell_n\})$ together with a non-degenerate graded symmetric cochain map $\langle-,-\rangle:V\otimes V\to \FR$ that satisfies the \emph{cyclicity condition}
\begin{align}
\langle-,-\rangle \, \circ \, (\unit\otimes\ell_n) = {\rm sgn}(\sigma) \, \langle-,-\rangle \, \circ \, (\unit\otimes\ell_n) \, \circ \, \tau^\sigma \ ,
\end{align}
for all $n\geq1$ and for all cyclic permutations $\sigma\in C_{n+1}\subset S_{n+1}$.
\end{definition}

Explicitly, when evaluated on elements $v_0,v_1,\dots,v_n\in V$, cyclicity translates to
\begin{align}\label{eq:cyclicity}
\langle v_0,\ell_n(v_1,v_2,\dots,v_n)\rangle = \pm \,  \langle v_n,\ell_n(v_0,v_1,\dots,v_{n-1})\rangle \ ,
\end{align}
where the $\pm$ sign is determined by the grading of the elements $v_i$.\footnote{See~\cite{BVChristian,Ciric:2021rhi} for the explicit cumbersome sign factors.} Thus cyclic $L_\infty$-algebras are homotopy coherent generalizations of quadratic Lie algebras. 

The dualization of the differential graded cocommutative coalgebra $(V,D)$ leads to an equivalent formulation of an $L_\infty$-algebra in terms of a differential graded commutative algebra with respect to the symmetric tensor product of $\Sym(V[1])$.
\begin{definition}\label{def:CEalgebra}
The \emph{Chevalley--Eilenberg algebra} of an $L_\infty$-algebra $(V,D)$ is the differential graded commutative algebra $\big(\Sym(V[1])^*,Q\big)$, where 
\begin{align}
Q=D^*:\Sym(V[1])^*\longrightarrow \Sym(V[1])^*
\end{align}
is a graded derivation of degree~$1$ such that $Q^2=0$.
\end{definition}
This is a homotopy coherent generalization of the Chevalley--Eilenberg algebra of a Lie algebra. In this dual language, a cyclic pairing translates to a graded symplectic $2$-form $\omega\in\Omega^2(V[1])$ which is $Q$-invariant. Evidently, one could equivalently use Definition~\ref{def:CEalgebra} as the starting point for the definition of an $L_\infty$-algebra, and this is the point of view which naturally makes contact between the objects and relations of $L_\infty$-algebras and the symmetries and dynamics of field theories, as we now briefly sketch; see e.g.~\cite{BVChristian} for a more detailed and precise review of the correspondence.

\subsection{Batalin--Vilkovisky formalism}
\label{sec:BV}

The Batalin--Vilkovisky (BV) formalism constructs a differential graded commutative algebra\footnote{The field observables actually form a much bigger space of functionals on $V[1]$ which are smooth in a suitable sense, but for our illustrative purposes here it suffices to consider only polynomial field observables in the dual symmetric algebra $\Sym(V[1])^*$.} $\big(\Sym(V[1])^*,Q_\BV\big)$ for a graded vector space $V=\bigoplus_{k\in\RZ}\,V^k$ which encodes the BV fields, namely ghosts and fields associated to a generalized gauge field theory, together with their antifields. More precisely, the space of BRST fields is $V^{\leq1}[1]$ where $V^{\leq1}:=\bigoplus_{k\leq1}\,V^k$ with $V^1$ encoding the fields and \smash{$V^0$} the ghosts, while the negatively graded homogeneous subspaces are non-zero only if there are higher gauge redundancies so that they encode ghosts-for-ghosts, etc. The BV field space extends these by their corresponding antifields to the $-1$-shifted cotangent bundle  $V[1]\simeq T^*[-1] V^{\leq1}[1]$. This data is supplemented by the BV antibracket $\{-,-\}_\BV$, which is the canonical graded Poisson bracket compatible with the differential $Q_\BV$ and whose inverse is the BV symplectic form $\omega_\BV$ on $V$ of degree~$-1$; using this non-degenerate symplectic pairing we identify the dual $V^*\simeq V[1]$ and\footnote{We do not delve into precise technical details of defining duals and tensor products of infinite-dimensional vector spaces, such as those which typically arise in field theories, as these involve subtle topological considerations.}
\begin{align}
\Sym(V[1])^*\simeq \Sym(V[2]) \ .
\end{align}

For illustrative purposes, let us first sketch how to build the corresponding $L_\infty$-algebra, and then make a more precise statement afterwards. The action of the BV differential is given by taking the BV antibracket with the BV action functional $S_\BV$, $Q_\BV=\{S_{\BV},-\}_\BV$. The action functional can be expanded as
\begin{align}
S_\BV = \sum_{m\geq2} \, S_\BV^{(m)} \ ,
\end{align}
where $S_\BV^{(m)}$ is the part of $S_\BV$ which is a polynomial of degree~$m$ in the BV fields. Then the brackets $\ell_n$ of the $L_\infty$-algebra are given by \smash{$\big\{S_\BV^{(n+1)},-\big\}_\BV$} for $n\geq1$. 
Nilpotency $(Q_\BV)^2=0$, or equivalently the classical master equation $\{S_\BV,S_\BV\}_\BV=0$, then translates to the homotopy Jacobi identities \eqref{eq:homotopyJacobi} for the brackets $\ell_n$, and $(V,\{\ell_n\})$ is an $L_\infty$-algebra. The $(-1)$-shifted symplectic structure $\omega_\BV$ induces a cyclic pairing of degree $-3$ on $V$, making it into a cyclic $L_\infty$-algebra. 

The action of the BV differential $Q_\BV$ on fields $A\in V^1$ and ghosts $c_{-k}\in V^{-k}$ for $k\geq0$ encodes the kinematical gauge symmetry of the field theory, that is, the (higher) gauge transformations and the closure of the gauge algebra. Its action on the antifields $A^+\in V^2$ incorporates the dynamical brackets of the $L_\infty$-algebra, while the BV transformations of the antifields $c_k^+\in V^{k+3}$ correspond to the (higher) Noether identities and the corresponding actions of the gauge parameters. The cohomology of $Q_\BV$ in degree~$0$ thus simultaneously encodes the quotients of the space of fields by the ideal of Euler--Lagrange derivatives and by the Lie algebra action of gauge transformations, that is, the classical observables of the field theory. 

Altogether, the natural algebraic structure underlying the BV formalism organises the gauge symmetries and dynamics of a field theory in a cochain complex that can be interpreted in the form
\begin{align}
\cdots \ \xrightarrow{ \quad \ \ } \ \begin{matrix}{\scriptstyle V^0} \\ \Big\{\begin{matrix} \textrm{\footnotesize gauge} \\[-1ex] \textrm{\footnotesize parameters} \end{matrix} \Big\} \end{matrix}  \ \xrightarrow{\quad \ \ } \  \begin{matrix} {\scriptstyle V^1} \\ \begin{matrix} \big\{ \textrm{\footnotesize fields} \big\} \\[-1ex] \phantom{\textrm{\footnotesize fields}}  \end{matrix} \end{matrix} \ \xrightarrow{\quad \ \ } \  \begin{matrix} {\scriptstyle V^2} \\ \Big\{ \begin{matrix} \textrm{\footnotesize equations} \\[-1ex] \textrm{\footnotesize of motion} \end{matrix}  \Big\} \end{matrix} \ \xrightarrow{\quad \ \ } \ \begin{matrix} {\scriptstyle V^3} \\ \Big\{ \begin{matrix} \textrm{\footnotesize Noether} \\[-1ex] \textrm{\footnotesize identities} \end{matrix} \Big\} \end{matrix}  \ \xrightarrow{\quad \ \ } \ \cdots 
\end{align}
The homogeneous subspaces $V^{-k}$ for $k\geq1$ to the left encode `higher' gauge transformations for reducible symmetries, while $V^{k+3}$ to the right dually encode `higher' Noether identities (differential relations among the Noether identities, and so on).

Conversely, starting from an  $L_\infty$-algebra $(V,\{\ell_n\})$ with a cyclic structure $\langle-,-\rangle$ of degree~$-3$, choose a basis $\{\tau_\alpha\}\subset V$ of the $L_\infty$-algebra with corresponding dual basis $\{\tau^\alpha\}\subset V^*\simeq V[3]$ relative to the cyclic pairing, that is, $\langle\tau^\alpha,\tau_\beta\rangle = \delta^\alpha_\beta$ for all $\alpha,\beta$. Following~\cite{BVChristian}, we introduce the `contracted coordinate functions' as the elements\footnote{Throughout this paper we adhere to the Einstein summation convention over repeated upper and lower indices.}
\begin{align}\label{eq:contractedcoord}
\xi := \tau^\alpha\otimes \tau_\alpha \ \in  \ \Sym(V[2])\otimes V
\end{align}
of degree~$1$. The $L_\infty$-algebra structure on $V$ naturally extends to the tensor product $\Sym(V[2])\otimes V$ through the extended brackets\footnote{More generally, the tensor product of any differential graded commutative algebra with an $L_\infty$-algebra admits an $L_\infty$-algebra structure with grading defined by the total degree, see Section~\ref{sec:braidedMC}; this generalizes the corresponding statement for Lie algebras, a prominent example being the Lie algebra $\Omega^\bullet(M,\frg)=\Omega^\bullet(M)\otimes\frg$ of differential forms valued in a Lie algebra $\frg$ that we encountered before. Here we view $\Sym(V[2])$ as equipped with the zero differential.}
\begin{align}
\ell_n^{\,\rm ext}(\zeta_1\otimes v_1,\dots,\zeta_n\otimes v_n) := \pm\,(\zeta_1\odot\cdots\odot \zeta_n)\otimes\ell_n(v_1,\dots,v_n) \ ,
\end{align}
for all $n\geq1$, $\zeta_1,\dots,\zeta_n\in\Sym(V[2])$ and $v_1,\dots,v_n\in V$; the explicit Koszul sign factors $\pm$ depend on the gradings and can be found in~\cite[Section~2.3]{BVChristian}. Similarly, the cyclic structure naturally extends to a symmetric non-degenerate pairing $\langle-,-\rangle^{\rm ext}:\big(\Sym(V[2])\otimes V\big)\otimes \big(\Sym(V[2])\otimes V\big)\to \Sym(V[2])$ of degree~$-3$ given by
\begin{align}
\langle \zeta_1\otimes v_1,\zeta_2\otimes v_1\rangle^{\rm ext} := \pm\,(\zeta_1\odot\zeta_2)\,\langle v_1,v_2\rangle \ .
\end{align}

With the contracted coordinate functions \eqref{eq:contractedcoord}, we now set
\begin{align}
Q_\BV\xi = - \sum_{n\geq1} \, \frac{(-1)^{n\choose 2}}{n!} \, \ell_n^{\,\rm ext}(\xi^{\otimes n})  \ ,
\end{align}
which we interpret as acting on elements of $\Sym(V[1])$. The homotopy Jacobi identities \eqref{eq:homotopyJacobi} imply that $(Q_\BV)^2=0$. The
extended pairing $\langle-,-\rangle^{\rm ext}$ induces a $(-1)$-shifted symplectic structure 
\begin{align}
\omega_\BV := -\tfrac12\,\langle\delta\xi,\delta\xi\rangle^{\rm ext} \ \in \ \Omega^2(V[1]) \ . 
\end{align}
In this way we recover the differential $Q_\BV$ and antibracket $\{-,-\}_\BV$ of the BV formalism directly from the cyclic $L_\infty$-algebra of a field theory. Moreover, the BV action functional is obtained from the element~\cite[Section~4.3]{BVChristian}
\begin{align}
S_\BV = \sum_{n\geq1} \, \frac{(-1)^{n\choose 2}}{(n+1)!} \, \langle\xi,\ell_n^{\,\rm ext}(\xi^{\otimes n})\rangle^{\rm ext} \ \in \ \Sym(V[2])
\end{align}
of degree~$0$. Then
\begin{align}
Q_\BV = \{S_\BV,-\}_\BV \ ,
\end{align}
and nilpotency $(Q_\BV)^2=0$ is equivalent to the {classical master equation}
\begin{align}
\{S_\BV,S_\BV\}_\BV = 0 \ .
\end{align}

\subsection{$L_\infty$-algebras of classical field theories}
\label{sec:LinftyCFT}

We can turn the formalism sketched in Section~\ref{sec:BV} into a concrete prescription for explicitly developing the classical dynamics of any generalized gauge field theory entirely from its underlying $L_\infty$-algebra $(V,\{\ell_n\})$; in this framework the equations of motion of the theory are the Maurer--Cartan equations of the $L_\infty$-algebra. For a field theory on a manifold $M$ we always assume that its $L_\infty$-algebra is \emph{local}, that is, all bracket operations are given by polydifferential operators. Focusing momentarily on field theories with only irreducible gauge symmetries, i.e. with independent gauge transformations, the $4$-term cochain complex
\begin{align}\label{eq:cochaincomplex}
V^0\xrightarrow{ \ \ \ell_1 \ \ }V^1\xrightarrow{ \ \ \ell_1 \ \ 
  }V^2\xrightarrow{ \ \ \ell_1 \ \ }V^3
\end{align}
encodes the free field dynamics, that is, the linearized gauge transformations, equations of motion and Noether identities. In particular, the free fields (solutions to the linearized equations of motion modulo linearized gauge transformations) live in the degree~$1$ cohomology $H^1(V,\ell_1)$ of this complex. The higher brackets $\ell_n$ with $n\geq2$ encode the interactions, and recover the full symmetries and dynamics of the generalized gauge theory. The homotopy Jacobi identities ensure covariance of the field equations, (on-shell) closure of gauge transformations, and the Noether identities, as we now explain. 

The vector spaces $V^{0}$ and $V^{1}$ are respectively the  spaces of
gauge parameters and dynamical fields. Given $\lambda \in V^{0}$ and $A\in V^{1}$, a gauge
transformation is the map $A\mapsto A+\delta_\lambda A$ where the infinitesimal gauge variation is given by\footnote{Here and in the following we use the convention ${k\choose 2}:=0$ for $k<2$.}
\begin{align} \label{gaugetransfA}
\delta_{\lambda}A=\ell_1(\lambda) + \sum_{n \geq1} \, \frac{(-1)^{n\choose 2}}{n!} \, \ell_{n+1}(\lambda,A^{\otimes n}) \ .
\end{align}
The gauge variations extend to maps $\delta_\lambda:V\to V$ of degree~$0$ by acting trivially on $V^0$, by commuting with arbitrary maps $V\to V$, and as derivations of operations $\mu$ defined on $V\otimes V$, that is, they satisfy the Leibniz rule
\begin{align}
\delta_\lambda\mu(v_1\otimes v_2) = \mu(\delta_\lambda v_1\otimes v_2) + \mu(v_1\otimes\delta_\lambda v_2)
\end{align}
for all $v_1,v_2\in V$.

The vector space $V^{2}$ is
the  space that the field equations of the  theory
belong to. They are encoded as the \emph{Maurer--Cartan equations} $\frF_{\tA}=0$, where
\begin{align}\label{EOM} 
\frF_{\tA}:=\ell_1(A) + \sum_{n \geq2} \, \frac{ (-1)^{n\choose 2}}{n!} \, \ell_{n}(A^{\otimes n})  \ ,
\end{align}
with the covariant gauge variations
\begin{align} \label{gaugetransfF}
\delta_{\lambda} \frF_{\tA} =\ell_2(\lambda,\frF_{\tA}) + \sum_{n \geq1} \, \frac{(-1)^{n\choose 2}}{n!} \, \ell_{n+2}(\lambda,\frF_{\tA},A^{\otimes n})  \ ,
\end{align}
which follow from homotopy Jacobi identity \eqref{eq:homotopyJacobi} evaluated on $(\lambda,A^{\otimes n})$ and summed over $n\geq1$.
Using symmetry of the brackets under exchange of degree~$1$ elements along with the Leibniz rule, it is easy to see that $\frF_{A+\delta_\lambda A}=\frF_{\tA} + \delta_\lambda\frF_{\tA} + O(\lambda^2)$, and so the gauge symmetry acts on the subspace of classical solutions $A\in V^1$.
The quotient of the space of solutions to the equations of motion $\frF_{\tA}=0$ by the action of gauge transformations is the \emph{Maurer--Cartan moduli space} $\CCM\CCC(V,\{\ell_n\})$, which is the space of classical physical states of the field theory. The BV complex $\big(\Sym(V[1])^*,Q_\BV\big)$ of Section~\ref{sec:BV} provides a rigorous definition of this moduli space as a homological resolution of the classical observables. On the other hand, the form of the dynamical equations \eqref{EOM} suggest that $L_\infty$-algebras are a dual formulation of the free differential algebras which play a central role in the geometric formulation of supergravity theories~\cite{Castellani:1991et}.

The closure relation for the gauge algebra then has the form
\begin{align}\label{eq:closure}
[\delta_{\lambda_1},\delta_{\lambda_2}]_\circ A
  = \delta_{C(\lambda_1,\lambda_2;A)}A  + \triangle_{\lambda_1,\lambda_2}A \ ,
\end{align}
where $[\delta_{\lambda_1},\delta_{\lambda_2}]_\circ := \delta_{\lambda_1}\circ\delta_{\lambda_2} - \delta_{\lambda_2}\circ\delta_{\lambda_1}$ is the commutator, while
\begin{align}\label{eq:ClambdaA}
C(\lambda_1,\lambda_2;A) = -\ell_2(\lambda_1,\lambda_2) -\sum_{n\geq1} \, \frac{(-1)^{n\choose 2}}{n!} \,
   \ell_{n+2}(\lambda_1,\lambda_2,A^{\otimes n})
\end{align}
and
\begin{align}\label{eq:trianglelambda}
\triangle_{\lambda_1,\lambda_2}A = \sum_{n\geq0} \, \frac{(-1)^{n-2\choose 2}}{n!} \, 
  \ell_{n+3}(\lambda_1,\lambda_2,\frF_{\tA},A^{\otimes n})  \ .
\end{align}
This follows from the homotopy Jacobi identity \eqref{eq:homotopyJacobi} evaluated on $(\lambda_1,\lambda_2,A^{\otimes n})$ and summed over $n\geq1$.
The algebra of gauge transformations thus closes on-shell, that is, when $\frF_{\tA}=0$, and on this locus the gauge algebra generally depends on the fields $A$; in other words, the gauge symmetries generally form an \emph{open} gauge algebra. The homotopy relations of the $L_\infty$-algebra guarantee that the Jacobi
identity is generally satisfied for any triple of maps $\delta_{\lambda_1}$,
$\delta_{\lambda_2}$ and~$\delta_{\lambda_3}$. The field theories considered in this paper will always have \emph{closed} field-independent gauge algebras generated by gauge transformations that are linear in the fields; that is, \eqref{eq:trianglelambda} vanishes identically at each order in $A^{\otimes n}$, while the sums in \eqref{gaugetransfA} and \eqref{eq:ClambdaA} terminate at the $2$-bracket $\ell_2$.

Finally, the vector space $V^3$ accomodates Noether's second theorem: It contains the image of a local differential operator $\sfd_{\tA}$ on $V^2$, which may depend on the fields. The Noether identities are then encoded by
\begin{align}\label{eq:Noether}
\dsf_{\tA}\frF_{\tA} := \ell_1(\frF_{\tA}) + \sum_{n\geq1} \, \frac{(-1)^{n\choose 2}}{n!} \,  \ell_{n+1}(\frF_{\tA},A^{\otimes n}) \ = \ 0 \ ,
\end{align}
which vanishes identically as a consequence of the homotopy relations of the $L_\infty$-algebra: evaluating \eqref{eq:homotopyJacobi} on $A^{\otimes n}$, using symmetry of the brackets under interchange of degree~$1$ elements, shows that
\begin{align}\label{eq:homotopyJacobidegree1}
\sum_{k=0}^n \, \frac{(-1)^{n\choose 2}}{k!\,(n-k)!} \, \ell_{k+1}\big(\ell_{n-k}(A^{\otimes n-k}),A^{\otimes k}\big) = 0 \ ,
\end{align}
and summing over $n\geq1$ yields \eqref{eq:Noether}. This expresses local differential relations among the Euler--Lagrange derivatives $\frF_{\tA}$ which hold \emph{off-shell}.

The action functional of the  field theory is encoded
via a symmetric non-degenerate bilinear pairing $\langle -,-\rangle :
V \otimes V\to\FR$ of degree~$-3$, which makes $V$ into a cyclic $L_\infty$-algebra. The only non-trivial pairings are
\begin{align}
\langle -,-\rangle : V^1 \otimes V^2 \longrightarrow\FR \qquad \mbox{and} \qquad \langle -,-\rangle : V^0 \otimes V^3 \longrightarrow\FR \ .
\end{align}
It is easy to see that the equations of motion $\frF_{\tA} = 0$ follow from varying the \emph{Maurer--Cartan action functional}
\begin{align} \label{action}
S(A) := \frac12\,\langle A,\ell_1(A)\rangle + \sum_{n\geq2} \, \frac{(-1)^{n\choose 2}}{(n+1)!} \,  \langle A, \ell_{n}(A^{\otimes n})\rangle \ ,
\end{align}
since then cyclicity and the Leibniz rule imply $\delta S(A)=\langle \delta A,\frF_{\tA}\rangle$. Note that $\ell_1(A)$ is associated with the free field equations of motion, while $\ell_n(A^{\otimes n})$ for $n\geq2$ contribute interaction vertices in the action functional.

Cyclicity and the Leibniz rule also imply
\begin{align}\label{eq:gtNoether}
\begin{split}
\delta_\lambda S(A)=\langle \delta_\lambda A,\frF_{\tA}\rangle &= \sum_{n \geq0} \, \frac{(-1)^{n\choose 2}}{n!} \, \langle \ell_{n+1}(\lambda,A^{\otimes n}),\frF_A\rangle \\[4pt]
&= -\sum_{n \geq0} \, \frac{(-1)^{n\choose 2}}{n!} \, \langle\lambda,\ell_{n+1}(A^{\otimes n},\frF_A)\rangle
=-\langle\lambda,\dsf_{\tA} \frF_{\tA} \rangle \ .
\end{split}
\end{align}
Thus gauge invariance of the action functional $\delta_\lambda
S(A)=0$, for all $\lambda\in V^0$,  is  equivalent to the Noether identities
$\dsf_{\tA}\frF_{\tA}=0$ by non-degeneracy of the pairing. This is simply a reformulation of Noether's second theorem. From this perspective, the Noether operator $\sfd_{\tA}:V^2\to V^3$ is the adjoint, with respect to the cyclic inner product, of $-\delta_\lambda$ viewed as a differential operator $-\delta_{(-)}A:V^0\to V^1$ acting on a gauge parameter $\lambda\in V^0$ with image in $V^1$. 

\subsubsection*{Higher gauge symmetries}

If the symmetries themselves
  have non-trivial symmetries, that is, there are further gauge
  redundancies in the description and the gauge symmetries are
  reducible, then the cochain complex \eqref{eq:cochaincomplex} should be extended into negative degrees $V^{-k}$ for $k\geq1$, which are the spaces of `higher' gauge parameters, together with their duals $V^{k+3}$ which contain the corresponding `higher' Noether identities. The vector space $V^{-1}$ parametrizes gauge transformations of the gauge parameters, $V^{-2}$ parametrizes gauge variations of the gauge transformations of gauge parameters, and so on. These higher gauge symmetries are encoded as
\begin{align}\label{eq:highergauge}
\delta_{(\lambda_{-k-1},A)}\lambda_{-k}=\ell_1(\lambda_{-k-1}) + \sum_{n\geq1} \, \frac{(-1)^{{n+1 \choose 2} + n\,k}}{n!} \,  \ell_{n+1}(\lambda_{-k-1},A^{\otimes n})  \ ,
\end{align}
where $\lambda_{-k}\in V^{-k}$ for $k\geq0$. 

Like the commutator of gauge transformations, these close only on-shell in general, as
\begin{align}\label{eq:highergaugecov}
\delta_{(\lambda_{-k-2},A)}\big(\delta_{(\lambda_{-k-1},A)}\lambda_{-k}\big) = \ell_2(\lambda_{-k-2},\frF_{\tA}) + \sum_{n\geq1} \, \frac{(-1)^{{n+2\choose 2} + n\,k}}{n!} \, \ell_{n+2}(\lambda_{-k-2},\frF_{\tA},A^{\otimes n}) \ ,
\end{align}
for all $k\geq-1$ with the convention $\lambda_{+1}:=A$, which follow from the homotopy Jacobi identity \eqref{eq:homotopyJacobi} evaluated on $(\lambda_{-k-2},A^{\otimes n})$ and summed over $n\geq1$. Thus gauge transformations of level~$k$ gauge parameters leave level~$k$ gauge transformations unchanged up to the equations of motion $\frF_{\tA}=0$. 
At the classical level, their
inclusion is purely algebraic and only serves to alter the cohomology
$H^\bullet(V,\ell_1)$ of the underlying cochain complex at its
extremities, leaving the moduli space $\CCM\CCC(V,\{\ell_n\})$ of classical states
unchanged. The inclusion of higher gauge parameters is necessary so that the resulting BV complex $\big(\Sym(V[1])^*,Q_\BV\big)$ of Section~\ref{sec:BV} properly forms a resolution of the space of classical observables.

Higher gauge redundancies are equivalent to higher Noether identities. Taking into account level 1 higher gauge parameters extends the underlying graded vector space $V$ with a non-trivial homogeneous subspace $V^{-1}$ and its dual $V^{4}$, and hence the duality pairing  
\begin{align}	
\langle -,-\rangle : V^{-1} \otimes V^{4} \longrightarrow\FR \ .
\end{align}
Using cyclicity, we define the level 1 Noether operator $\dsf_{A}^{\swone}:V^{3}\to V^{4}$ as the adjoint of $\delta_{(\lambda_{-1},A)}$ from \eqref{eq:highergauge} viewed as an operator $\delta_{(-,A)}\lambda: V^{-1}\to V^0$ on $V^{-1}$ with values in $V_{0}$, that is, 
\begin{align}
	\begin{split}
	\langle \delta_{(\lambda_{-1},A)} \lambda, \varLambda \rangle =: \langle \lambda_{-1} , \sfd_{A}^{\swone} \varLambda \rangle \ , 
	\end{split}
\end{align}
for all $\varLambda \in V^3$. Explicitly
\begin{align}\label{eq:level1NoetherOp}
\sfd_A^\swone\varLambda = \ell_1(\varLambda) + \sum_{n\geq1} \, \frac{(-1)^{{n+1 \choose 2}+n}}{n!} \,  \ell_{n+1}(\varLambda,A^{\otimes n})  \ .
\end{align}

Applying $\delta_{(\lambda_{-1},A)}$ on both sides of $\langle \delta_\lambda A,\CA\rangle=-\langle\lambda,\dsf_A \CA \rangle$ for arbitrary $\CA\in V^2$, we get 
\begin{align}
	\langle \delta_{(\lambda_{-1},A)}(\delta_\lambda A),\CA\rangle=-\langle\delta_{(\lambda_{-1},A)}\lambda,\dsf_A \CA \rangle \ .
\end{align}
Using the closure property \eqref{eq:highergaugecov}, the definition \eqref{eq:level1NoetherOp} of the level 1 Noether operator, and finally cyclicity and non-degeneracy of the pairing, we obtain the level 1 Noether identities
\begin{align}\label{eq:level1Noether}
\dsf_{A}^{\swone}\circ \dsf_A = -\ell_2(\frF_A,-)-\sum_{n\geq1} \, \frac{(-1)^{{n\choose 2}}}{n!} \, \ell_{n+2}(\frF_A,A^{\otimes n}, -) \ ,
\end{align}
which state that, up to the equations of motion $\frF_A=0$, not all components of the level 0 Noether operator $\dsf_A$ are independent.\footnote{Of course, this can also be derived directly from the homotopy Jacobi identities of the extended $L_\infty$-algebra.} This is naturally `dual' to the statement that the level 0 gauge parameters generate reducible gauge symmetries, i.e. that not all gauge directions are independent in the tangent space to the space of fields $V^1$. This is to be compared with the level 0 Noether identities $\dsf_A \circ \frF_{-}= 0$, which state that not all components of the Maurer--Cartan operator $\frF_{-}:V^1\to V^2$ are independent due to the existence of level 0 gauge symmetries. 

Proceeding inductively,  the existence of level $k$ gauge transformations extends the underlying graded vector space $V$ with a non-trivial homogeneous subspace $V^{-k}$ and its dual $V^{k+3}$, with the duality pairing $\langle -,-\rangle : V^{-k} \otimes V^{k+3} \to \FR $.  The adjoint of $\delta_{(\lambda_{-k},A)}$ defines the level $k$ Noether operator \smash{$\dsf_{A}^{\textrm{\tiny$(k)$}}: V^{k+2} \to V^{k+3}$} which encodes the interdependence of the level $k{-}1$ Noether operator through the level $k$ Noether identities
\begin{align}\label{eq:levelkNoether}
	\dsf_{A}^{\textrm{\tiny$(k)$}}\circ \dsf^{\textrm{\tiny$(k{-}1)$}}_A = -\ell_2(\frF_A,-)-\sum_{n\geq1} \, \frac{(-1)^{{n\choose 2}}}{n!} \, \ell_{n+2}(\frF_A,A^{\otimes n}, -) \ .
\end{align}

\subsection{Example: Yang--Mills theory in the $L_\infty$-algebra formalism}
\label{sec:Yang-Millstheory}

Yang--Mills theory is the basic example of a field theory with gauge symmetries where the underlying $L_{\infty}$-algebra is not a  differential graded Lie algebra. The higher bracket arises from the quartic interaction vertex in the Yang--Mills action functional. The Yang--Mills $L_{\infty}$-algebra has appeared many times in the literature. It was first explicitly noted in~\cite{Zeitlin:2009zza} and has more recently reappeared in~\cite{Hohm:2017pnh,BVChristian,Elliott:2021ffl}. We close this section with a review of the $L_\infty$-algebra formulation of Yang--Mills theory, which along with the Chern--Simons gauge theory of Section~\ref{sec:whatisgauge} will play a prominent role in some of our later examples.

Let $(M,g)$ be a $d$-dimensional oriented Lorentzian manifold,\footnote{The discussion below is presented for Lorentzian signature in light of our later considerations, but it is easily adapted to more general pseudo-Riemannian manifolds.} and $\ast_\hodge : \Omega^{k}(M) \to \Omega^{d-k}(M)$ the corresponding Hodge duality operator. Let $(\frg,[-,-]_\frg,\Tr_\frg)$ be a quadratic Lie algebra. As in Section~\ref{sec:whatisgauge}, we extend the Lie bracket on $\frg$ to a Lie bracket on $\frg$-valued differential forms $\Omega^\bullet(M,\frg):=\Omega^\bullet(M)\otimes\frg$ via the tensor product with exterior multiplication, and we denote this extended Lie bracket with the same symbol $[-,-]_\frg$.

The Yang--Mills action functional for a gauge field $A \in \Omega^{1}(M,\frg)$ is defined by
\begin{align}\label{eq:YMaction}
	S(A)= \int_{M}\,\text{Tr}_\frg (F_{\tA} \wedge \ast_\hodge  F_{\tA}) \ ,
\end{align}
where as before $F_{\tA}:= \dd A +\frac{1}{2}\,[A,A]_\frg\in\Omega^2(M,\frg)$ is the curvature 2-form. It is invariant under the usual gauge transformations
\begin{align}
\delta_\lambda A=\dd\lambda+[A,\lambda]_\frg
\end{align}
for $\lambda \in \Omega^{0}(M,\frg)$. The Yang--Mills field equations state that the Hodge dual of the curvature is covariantly constant:
\begin{align}
	\frF_{\tA}=\dd^{\tA} \ast_\hodge  F_{\tA} =\dd \ast_\hodge  F_{\tA} + [A,\ast_\hodge  F_{\tA}]_\frg=0 \  \in \ \Omega^{d-1}(M,\frg) \ ,
\end{align}
whereas the corresponding Noether identities are
\begin{align}\label{eq:YMNoether}
	\dd^{\tA}\frF_{\tA} = \big(\dd^{\tA}\big)^{2}\ast_\hodge  F_{\tA} = [F_{\tA},\ast_\hodge  F_{\tA}]_\frg=0 \  \in \ \Omega^{d}(M,\frg) \ . 
\end{align}
These may be also verified directly by symmetry and antisymmetry of the $d$-form $[F_{\tA},\ast_\hodge  F_{\tA}]_\frg$ valued in the Lie algebra $\frg$. From these expressions the $L_\infty$-algebra organising Yang--Mills theory may be determined.

The underlying graded vector space is $V=V^{0}\oplus V^{1}\oplus V^{2}\oplus V^{3}$, where
\begin{align}
	\begin{split}
		V^{0}:= \Omega^{0}(M,\frg)  \ , \quad V^{1}:= \Omega^{1}(M,\frg) \ , \quad
		V^{2}&:= \Omega^{d-1}(M,\frg) \qquad \mbox{and} \qquad 
		V^{3}:= \Omega^{d}(M,\frg) \ .
	\end{split}
\end{align}
Denoting the corresponding elements by $\lambda$, $A$, $\CA$ and $\varLambda$ respectively, the $L_{\infty}$-algebra structure is given by the differential 
\begin{align}
	\ell_1(\lambda) = \dd\lambda \ , \quad \ell_1(A) = \dd \ast_\hodge  \dd 
	A \qquad  
	\mbox{and} \qquad \ell_1(\CA) = \dd \CA   \ ,
\end{align} 
along with the 2-brackets 
\begin{align}
	\begin{split} 
		\ell_{2}(\lambda_1,\lambda_2)= -[\lambda_1,\lambda_2]_\frg \qquad & , \qquad
		\ell_2 (\lambda, A)=-[\lambda,A]_\frg \ , \\[4pt]
		\ell_{2}(\lambda, \CA) = -[\lambda,\CA]_\frg \qquad , \qquad
		\ell_2(\lambda,\varLambda)&=-[\lambda,\varLambda]_\frg \qquad , \qquad
		\ell_{2}(A, \CA)= -[A,\CA]_\frg \ , \\[4pt]
		\ell_{2}(A_1, A_2)= - \dd \ast_\hodge  [A_1,A_2]_\frg -&[A_1,\ast_\hodge  \dd A_2]_\frg + (-1)^{d} \, [\ast_\hodge  \dd A_1, A_2]_\frg \ , 
	\end{split}
\end{align}
and finally the non-zero 3-bracket
\begin{align}
	\begin{split}
		\ell_3(A_1,&A_2,A_3) = -\big[A_1, \ast_\hodge [A_2,A_3]_\frg\big]_\frg-\big[A_2,\ast_\hodge [A_1,A_3]_\frg\big]_\frg+ (-1)^{d}\, \big[\ast_\hodge [A_1,A_2]_\frg,A_3\big]_\frg \ ,
	\end{split}
\end{align}
owing to the quartic interaction vertex. Note that the underlying cochain complex differs from that of Chern--Simons theory (see Section~\ref{sec:whatisgauge}), and that now the $L_\infty$-algebra is no longer a differential graded Lie algebra because $\ell_3\neq0$.

The resulting cyclic inner product is given by
\begin{align} \label{eq:YMpairing}
	\langle A ,\CA \rangle := \int_M\,
	\Tr_\frg(A\wedge  \CA) \qquad \mbox{and} \qquad
	\langle \lambda,\varLambda \rangle := \int_M\,
	\Tr_\frg(\lambda\wedge \varLambda) \ ,
\end{align}
with cyclicity following from ${\rm ad}(\frg)$-invariance of the pairing $\Tr_{\frg}:\frg\otimes\frg\to\FR$, integration by parts (for fields with suitable asymptotic decay) and symmetry properties of the Hodge duality operator. By design, this cyclic $L_{\infty}$-algebra encodes the full data of classical Yang--Mills theory as described in Section~\ref{sec:LinftyCFT}. 

\section{Noncommutative gauge symmetry and $L_\infty$-algebras}
\label{sec:NCLinfty}

As a first step towards understanding noncommutative field theories in the modern context of this paper, in this section we look at some basic examples, mentioned in Section~\ref{sec:Intro}, of noncommutative gauge theories arising in string theory which fit into the standard $L_\infty$-algebra framework discussed in Section~\ref{sec:Linfty}. These are the examples of flat D-branes in a constant $B$-field, which lead to the standard Moyal--Weyl deformations of gauge theories, and also curved D-branes in WZW models based on compact Lie groups, which lead to fuzzy field theories (i.e. matrix models). We will then discuss the limitations of this perspective when confronted with the noncommutative gauge theories on D-branes wrapping curved submanifolds with non-trivial worldvolume fluxes in a generic background.

\subsection{Star-gauge theories in the $L_\infty$-algebra formalism}
\label{sec:usualNCLinfty}

Let us start with the conventional noncommutative field theories on flat $d$-dimensional Minkowski spacetime $\FR^{1,d-1}$ with a constant Poisson bivector $\theta=\frac12\,\theta^{\mu\nu}\,\partial_\mu\wedge\partial_\nu$. We consider for illustration a noncommutative gauge theory with gauge algebra $\fru(1)$ and fields multiplied with the Moyal--Weyl star-product \eqref{eq:MoyalWeylstarprod}.\footnote{The extension to nonabelian gauge algebras $\fru(n)$ with $n>1$ is a straightforward modification (with no essential novelties) obtained by composing the star-product \eqref{eq:MoyalWeylstarprod} with matrix multiplication and all integrations with the trace in the fundamental representation of $\fru(n)$.} 

Let us start by introducing some notation: Given functions $\lambda,\rho\in\Omega^0(\FR^{1,d-1})$, we denote their star-commutator by
\begin{align}\label{eq:starcomm}
[\lambda\starcom\rho] := \lambda\star\rho - \rho\star\lambda \ ,
\end{align}
which makes $\Omega^0(\FR^{1,d-1})$ into a Lie algebra due to associativity of the Moyal--Weyl star-product. Given differential forms $\alpha,\beta\in\Omega^\bullet(\FR^{1,d-1})$, we denote their star-exterior product by $\alpha\wedge_\star\beta$, which is defined by regarding the graded commutative exterior algebra $\Omega^\bullet(\FR^{1,d-1})$ as a module over the noncommutative algebra of functions $\big(\Omega^0(\FR^{1,d-1}),\star\big)$; in particular, the holonomic coframe on $\FR^{1,d-1}$ obeys $\dd x^\mu\wedge_\star\dd x^\nu = \dd x^\mu\wedge \dd x^\nu = -\dd x^\nu\wedge_\star\dd x^\mu$ and $\lambda\star\dd x^\mu = \lambda\cdot\dd x^\mu = \dd x^\mu\star \lambda$ for $\lambda\in\Omega^0(\FR^{1,d-1})$. Similarly to \eqref{eq:starcomm}, we define the graded star-commutator of arbitrary homogeneous differential forms by
\begin{align}
[\alpha\starcom\beta] := \alpha\wedge_\star\beta - (-1)^{|\alpha|\,|\beta|} \,  \beta\wedge_\star\alpha \ .
\end{align}
Then $\big(\Omega^\bullet(\FR^{1,d-1}),[-\starcom-],\dd\big)$ is a differential graded Lie algebra, and noncommutative field theory with the Moyal--Weyl star-product has a simple noncommutative differential calculus.

With these preliminaries, it is clear that the conventional noncommutative $\sU(1)$ gauge theories are formally identical to their classical counterparts with nonabelian gauge symmetry, as is well-known. In particular, the \emph{star-gauge transformation} of a gauge field $A\in\Omega^1(\FR^{1,d-1})$ by a gauge parameter $\lambda\in\Omega^0(\FR^{1,d-1})$ is given by the usual expression
\begin{align}\label{eq:stargauge}
\delta_\lambda^\star A = \dd\lambda - [\lambda\starcom A] \ .
\end{align}
Using the Leibniz rule for the star-commutator $[-\starcom-]$ with respect to the de~Rham differential $\dd$, together with its Jacobi identity, one easily checks the off-shell gauge closure
\begin{align}\label{eq:stargaugeclosure}
\big[\delta^\star_{\lambda_1},\delta^\star_{\lambda_2}\big]^{\phantom{\dag}}_\circ A = \delta^\star_{[\lambda_1\stackrel{\scriptstyle\star}{\scriptstyle,} \lambda_2]}A
\end{align}
to the Lie algebra $\big(\Omega^0(\FR^{1,d-1}),[-\starcom-]\big)$. 

Under these star-gauge transformations, 
the noncommutative field strength
\begin{align}\label{eq:fieldstrengthstargauge}
F_{\tA}^\star := \dd A + \tfrac12\,[A\starcom A] = \dd A+ A\wedge_\star A \ \in \ \Omega^2(\FR^{1,d-1})
\end{align}
transforms covariantly:
\begin{align}\label{eq:starFcov}
\delta_\lambda^\star F_{\tA}^\star = -[\lambda\starcom F_{\tA}^\star] \ .
\end{align}
This similarity to the analysis of ordinary nonabelian gauge theories was used in~\cite{Blumenhagen:2018kwq} to cast noncommutative $\sU(1)$ Chern--Simons and Yang--Mills theories into the framework of standard $L_\infty$-algebras, which we now describe. 

\subsubsection*{Noncommutative Chern--Simons theory}

The $L_\infty$-algebra formulation of noncommutative Chern--Simons theory in three dimensions follows the analysis of the usual nonabelian Chern--Simons gauge theory presented in Section~\ref{sec:whatisgauge}. In this case the organising $L_\infty$-algebra is simply the differential graded Lie algebra with underlying graded vector space 
\begin{align}
V=\Omega^\bullet(\FR^{1,2})
\end{align}
and the non-zero brackets
\begin{align}
\ell_1(\alpha) = \dd \alpha \qquad \mbox{and} \qquad \ell_2(\alpha,\beta) = -[\alpha\starcom\beta]
\end{align}
for $\alpha,\beta\in V$. Note that the underlying cochain complex \eqref{eq:cochaincomplex} (the de~Rham complex in this case) is unchanged by the noncommutative deformation; that is, the free fields are the same and only the interactions are modified by the noncommutativity, as is well-known for Moyal--Weyl type noncommutative field theories. 

Applying the standard machinery from Section~\ref{sec:LinftyCFT}, one easily verifies that the noncommutative gauge symmetry \eqref{eq:stargauge} is reproduced in the expected way as
\begin{align}\label{eq:NCelldelta}
\delta_\lambda^\star A = \ell_1(\lambda) + \ell_2(\lambda,A) \ ,
\end{align}
with the closure \eqref{eq:stargaugeclosure} expressed through
\begin{align}\label{eq:NCellclosure}
\big[\delta^\star_{\lambda_1},\delta^\star_{\lambda_2}\big]_\circ A = \delta^\star_{-\ell_2(\lambda_1,\lambda_2)}A \ .
\end{align}
The Maurer--Cartan equations $\frF_{\tA}=0$ reproduce the noncommutative Chern--Simons equation of motion
\begin{align}
\frF_{\tA} = \ell_1(A) - \tfrac12\,\ell_2(A,A) = F_{\tA}^\star \ ,
\end{align}
with the covariance \eqref{eq:starFcov} expressed as
\begin{align}\label{eq:NCellcov}
\delta_\lambda^\star \frF_{\tA} = \ell_2(\lambda,\frF_{\tA}) \ .
\end{align}
The Noether identities
\begin{align}
0 \ = \ \dsf_{\tA}\frF_{\tA} = \ell_1(\frF_{\tA})+\ell_2(\frF_{\tA},A) = \dd F_{\tA}^\star + [A\starcom F_{\tA}^\star] =: \dd_\star^{\tA} F_{\tA}^\star
\end{align}
are simply the star-gauge Bianchi identities $\dd_\star^{\tA} F_{\tA}^\star=0$ satisfied by the noncommutative field strength~$F_{\tA}^\star$, where $\dd_\star^{\tA}$ is star-gauge covariant derivative. 

The inner product
\begin{align}
\langle\alpha,\beta\rangle = \int_{\FR^{1,2}} \, \alpha\wedge_\star\beta = \int_{\FR^{1,2}} \, \alpha\wedge\beta
\end{align}
on $V$ pairs differential forms $\alpha$ and $\beta$ on $\FR^{1,2}$ in complementary degrees, where the second equality follows from an integration by parts and is sometimes expressed as saying that the Moyal--Weyl star-product is \emph{cyclic}. This makes $V$ into a cyclic $L_\infty$-algebra, and the corresponding Maurer--Cartan action functional then reproduces the standard noncommutative $\sU(1)$ Chern--Simons functional:
\begin{align}
S(A) = \frac12\,\langle A,\ell_1(A)\rangle - \frac1{3!}\,\langle A,\ell_2(A,A)\rangle = \frac12\,\int_{\FR^{1,2}} \, A\wedge_\star \dd A +\frac2{3}\,A\wedge_\star A\wedge_\star A \ ,
\end{align}
where we used $[A\starcom A]=2\,A\wedge_\star A$. Thus, like the classical theory, noncommutative $\sU(1)$ Chern--Simons theory fits nicely into the algebraic framework of a differential graded Lie algebra.

\subsubsection*{Noncommutative Yang--Mills theory}

A similar analysis can be done for noncommutative Yang--Mills theory in any spacetime dimension $d$, using the $L_\infty$-algebra of usual nonabelian Yang--Mills theory from Section~\ref{sec:Yang-Millstheory}. Let $\ast_\hodge : \Omega^{k}(\FR^{1,d-1}) \to \Omega^{d-k}(\FR^{1,d-1})$ be the Hodge duality operator induced by the standard flat Minkowski metric of $\FR^{1,d-1}$. 

The underlying graded vector space is $V=V^{0}\oplus V^{1}\oplus V^{2}\oplus V^{3}$, where
\begin{align}
	\begin{split}
		V^{0}= \Omega^{0}(\FR^{1,d-1}) \ , \quad V^{1}= \Omega^{1}(\FR^{1,d-1}) \ , \quad
		V^{2}= \Omega^{d-1}(\FR^{1,d-1}) \ , \quad 
		V^{3}= \Omega^{d}(\FR^{1,d-1}) \ .
	\end{split}
\end{align}
Denoting the corresponding elements by $\lambda\in V^0$, $A\in V^1$, $\CA\in V^2$ and $\varLambda\in V^3$, the differential is given by
\begin{align}
	\ell_1(\lambda) = \dd\lambda \ , \quad \ell_1(A) = \dd \ast_\hodge  \dd 
	A \qquad  
	\mbox{and} \qquad \ell_1(\CA) = \dd \CA   \ ,
\end{align} 
the 2-brackets are
\begin{align}
	\begin{split} 
		\ell_{2}(\lambda_1,\lambda_2)= -[\lambda_1\starcom\lambda_2] \qquad & , \qquad
		\ell_2 (\lambda, A)=-[\lambda\starcom A] \ , \\[4pt]
		\ell_{2}(\lambda, \CA) = -[\lambda\starcom\CA] \qquad , \qquad
		\ell_2(\lambda,\varLambda)&=-[\lambda\starcom\varLambda] \qquad , \qquad
		\ell_{2}(A, \CA)= -[A\starcom\CA] \ , \\[4pt]
		\ell_{2}(A_1, A_2)= - \dd \ast_\hodge  [A_1\starcom A_2] - & [A_1\starcom\ast_\hodge  \dd A_2] + (-1)^{d} \, [\ast_\hodge  \dd A_1\starcom A_2] \ , 
	\end{split}
\end{align}
and finally there is a non-zero 3-bracket
\begin{align}
	\begin{split}
		\ell_3(A_1,&A_2,A_3) = -\big[A_1\starcom \ast_\hodge [A_2\starcom A_3]\big]-\big[A_2\starcom\ast_\hodge [A_1\starcom A_3]\big] + (-1)^{d}\, \big[\ast_\hodge [A_1\starcom A_2]\starcom A_3\big] \ .
	\end{split}
\end{align}

The noncommutative gauge transformations and their closure again follow as in \eqref{eq:NCelldelta} and \eqref{eq:NCellclosure}, while the Maurer--Cartan equations $\frF_{\tA}=0$ reproduce the vacuum equation of motion for noncommutative $\sU (1)$ Yang--Mills theory:
\begin{align}
\frF_{\tA} = \ell_{1}(A) -\tfrac{1}{2}\,  \ell_{2}(A,A) - \tfrac{1}{3!} \, \ell_{3} (A,A,A) = \dd_\star^{\tA} \ast_\hodge  F_{\tA}^\star \ ,
\end{align}
with covariance again captured by \eqref{eq:NCellcov}. The Noether identities
\begin{align}
0 \ = \ \dsf_{\tA}\frF_{\tA} = \ell_1(\frF_{\tA}) + \ell_2(\frF_{\tA},A) = \dd_\star^{\tA}\frF_{\tA} = 
	\big(\dd_\star^{\tA}\big)^{2}\ast_\hodge  F_{\tA}^\star = [F_{\tA}^\star\starcom\ast_\hodge  F_{\tA}^\star]
\end{align}
may be checked directly using symmetry properties of the $d$-form $[F_{\tA}^\star\starcom\ast_\hodge  F_{\tA}^\star]\in\Omega^d(\FR^{1,d-1})$. 

The inner product on $V$ given by
\begin{align} \label{eq:NCYMpairing}
	\langle A ,\CA \rangle = \int_{\FR^{1,d-1}} \,
	A\wedge_\star \CA \qquad \mbox{and} \qquad
	\langle \lambda,\varLambda \rangle = \int_{\FR^{1,d-1}} \,
	\lambda \wedge_\star \varLambda 
\end{align}
defines a cyclic structure on the noncommutative Yang--Mills $L_\infty$-algebra, and the corresponding Maurer--Cartan action functional reproduces (after integration by parts) the standard noncommutative $\sU(1)$ Yang--Mills functional:
\begin{align}
S(A) = \frac{1}{2}\, \langle A, \ell_{1}(A) \rangle - \frac{1}{3!}\, \langle A, \ell_2 (A,A)\rangle -\frac{1}{4!}\,  \langle A,\ell_3 (A,A,A) \rangle = \frac{1}{2}\, \int_{\FR^{1,d-1}}\, F_{\tA}^\star \wedge_\star  \ast_\hodge  F_{\tA}^\star \ .
\end{align}

\subsection{Fuzzy field theories in the $L_\infty$-algebra formalism}
\label{sec:fuzzyNCLinfty}

A particularly well understood and tractable class of noncommutative gauge theories on D-branes wrapping curved submanifolds with a non-constant gauge flux is provided by D-branes in WZW models. These involve target spaces which are group manifolds $\sG$ with non-trivial NS--NS $3$-form fluxes, and they contain D-branes that wrap conjugacy classes of $\sG$ carrying non-constant $2$-form flux~\cite{Alekseev:1999bs}. Here we focus on the simplest example $\sG=\sSU(2)$ for illustration, where the conjugacy classes are $2$-spheres $\FS^2\subset \FS^3$ labelled by the spin $j$ of a representation of the $\frsu(2)_k$ current algebra at WZW level $k\in\RZ_{>0}$. In the semi-classical limit with $k\to\infty$, the space of open string ground states forms a finite-dimensional matrix algebra, in contrast to the infinite-dimensional flat space case from Section~\ref{sec:usualNCLinfty}, and the noncommutative matrix product $f\,g$ now plays the analogue role of the Moyal--Weyl star-product. 

The noncommutative gauge theory supported by the two-dimensional worldvolume of the D-brane is thus no longer a continuum field theory but rather a matrix model~\cite{Alekseev:2000fd}, called a `fuzzy' field theory, or alternatively a field theory on the `fuzzy sphere' which is a matrix approximation of the standard $2$-sphere $\FS^2$ obtained by quantizing conjugacy classes in $\sSU(2)$, and which is covariant under $\sSU(2)$. Being finite-dimensional systems, these noncommutative field theories are particularly well tailored to the formalism developed in Section~\ref{sec:Linfty}, as they avoid the many functional analytic complications brought in by continuum field theories.

Let us first explain how to set up a differential calculus on the fuzzy sphere suitable for the formulation of fuzzy field theories. The algebra of functions on the fuzzy sphere $\FS_N^2$ is defined by
\begin{align}\label{eq:fuzzyalgebra}
\CC_N := (j)\otimes(j)^* \simeq {\sf Mat}(N) \ ,
\end{align}
where $(j)$ denotes the irreducible spin $j=(N-1)/2$ representation of $\frsu(2)$ for an integer $N>1$ and $(j)^*$ denotes the dual representation; the isomorphism with $N{\times} N$ complex matrices follows from the fact that $(j)$ is $N$-dimensional. Fix a basis of the Lie algebra $\frsu(2)$ where the structure constants are given by the Levi--Civita symbol $\epsilon_{ab}{}^c$, for $a,b,c=1,2,3$, and denote by $\frac1{r_N}\,X_{a}\in\CC_N$ the generators of $\frsu(2)$ in the spin~$j$ representation, where \smash{$r _N:=\frac1{\sqrt{j\,(j+1)}}$}. Then the matrices $X_{a}$ generate the algebra $\CC_N$ and satisfy the relations
\begin{align}
[X_{a},X_b] = \ii\,r_N\,\epsilon_{ab}{}^c\, X_c \ , \quad \delta^{ab}\,X_{a}\,X_b=\unit \qquad \mbox{and} \qquad X_{a}^\ast = X_{a}
\end{align}
of the fuzzy unit sphere, where 
\begin{align}
[\lambda,\rho] := \lambda\,\rho - \rho\,\lambda
\end{align}
denotes the matrix commutator of $\lambda,\rho\in\CC_N$, which makes $\CC_N$ into a Lie algebra, and ${}^*$ denotes Hermitian conjugation.

The fuzzy $2$-sphere $\FS_N^2$ has a three-dimensional $\frsu(2)$-covariant differential calculus given by the Chevalley--Eilenberg algebra
\begin{align}
\Omega^\bullet(\CC_N) := \CC_N\otimes\midwedge^\bullet\frsu(2)^* 
\end{align}
of $\frsu(2)$ with coefficients in the fuzzy sphere algebra \eqref{eq:fuzzyalgebra}, which as an $\frsu(2)$-module decomposes into irreducible representations as
\begin{align}
\CC_N \simeq \bigoplus_{J=0}^{N-1} \, (J) \ .
\end{align}
Let $\vartheta^{a}\in\frsu(2)^*\subset\Omega^1(\CC_N)$ for $a=1,2,3$ be a compatible coframe for the $\CC_N$-module of $1$-forms defined by the dual of the chosen Lie algebra basis, with $\vartheta^{a}\wedge\vartheta^b = -\vartheta^b\wedge\vartheta^{a}$, which generates the entire differential calculus $\Omega^\bullet(\CC_N)$ and commutes with $\Omega^0(\CC_N)=\CC_N$. 
The Chevalley--Eilenberg differential $\dd_\ce$ is then defined by 
\begin{align}
\dd_\ce \lambda := \tfrac1{r_N} \, [X_{a},\lambda] \, \vartheta^{a}
\end{align}
for $\lambda\in\Omega^0(\CC_N)=\CC_N$, together with the Maurer--Cartan equations
\begin{align}
\dd_\ce\vartheta^{a} = -\tfrac\ii2\,\epsilon^{a}{}_{bc} \, \vartheta^b\wedge\vartheta^c \ ,
\end{align}
and extended to the whole of $\Omega^\bullet(\CC_N)$ by the graded Leibniz rule. We extend the matrix commutator on $\CC_N$ to arbitrary homogeneous differential forms $\alpha,\beta\in\Omega^\bullet(\CC_N)$ by
\begin{align}\label{eq:matrixcommforms}
[\alpha,\beta] := \alpha\wedge\beta - (-1)^{|\alpha|\,|\beta|} \, \beta\wedge\alpha \ ,
\end{align}
which makes $\big(\Omega^\bullet(\CC_N),[-,-],\dd_\ce\big)$ into a differential graded Lie algebra, enabling the construction of field theories on the fuzzy sphere $\FS_N^2$.

With this noncommutative differential calculus it is now straightforward to formulate noncommutative gauge symmetries in the standard way: The \emph{fuzzy gauge transformation} of a gauge field $A\in\Omega^1(\CC_N)$ by a gauge parameter $\lambda\in\Omega^0(\CC_N)$ is given by the usual formula
\begin{align}\label{eq:fuzzygauge}
\delta_\lambda^\ce A = \dd_\ce\lambda - [\lambda,A] \ ,
\end{align}
and using the by now familiar arguments one shows that they close to the Lie algebra $(\CC_N,[-,-])$:
\begin{align}\label{eq:fuzzyclosure}
\big[\delta_{\lambda_1}^\ce,\delta_{\lambda_2}^\ce\big]_\circ = \delta^\ce_{[\lambda_1,\lambda_2]} \ .
\end{align}
The field strength defined by
\begin{align}\label{eq:fieldstrengthfuzzy}
F_{\tA}^\ce := \dd_\ce A + \tfrac12\,[A,A] = \dd_\ce A + A\wedge A \ \in \ \Omega^2(\CC_N)
\end{align}
transforms covariantly under this gauge symmetry:
\begin{align}
\delta_\lambda^\ce F_{\tA}^\ce = -[\lambda,F_{\tA}^\ce] \ .
\end{align}

The low-energy effective field theory on a D-brane of type $j$ wrapping this fuzzy $2$-sphere was derived in~\cite{Alekseev:2000fd} as a noncommutative gauge theory whose action functional is a sum of a Yang--Mills term and a Chern--Simons term, each of which is individually invariant under the gauge transformations \eqref{eq:fuzzygauge}.\footnote{One can extend this description to stack of $n>1$ D-branes of type $j$ by considering the differential graded noncommutative algebra $\Omega^\bullet(\CC_N,\fru(n)):=\Omega^\bullet(\CC_N)\otimes\fru(n)$ of $\fru(n)$-valued forms on the fuzzy $2$-sphere, with the differential $\dd_\ce\otimes\unit$, and the (graded) Lie bracket extended to the matrix commutator on ${\sf Mat}(N)\otimes{\sf Mat}(n)$, and similarly for the normalised trace below.} The underlying $L_\infty$-algebra was described in~\cite{Blumenhagen:2018kwq} in a flat holonomic coframe. Here we give an alternative treatment which uses the more natural covariant coframe $\vartheta^{a}$ appropriate to a curved space; Yang--Mills and Chern--Simons gauge theories in this formulation were first discussed in~\cite{Grosse:2000gd}. The advantage of this description is that, like the flat space Moyal--Weyl type noncommutative gauge theories of Section~\ref{sec:usualNCLinfty}, the two theories can be treated separately and the $L_\infty$-algebra structure exactly parallels the case of ordinary nonabelian gauge theories. 

\subsubsection*{Chern--Simons theory on the fuzzy sphere}

The $L_\infty$-algebra formulation of Chern--Simons theory on $\FS_N^2$ in the coframe formalism was discussed in~\cite{Nguyen:2021rsa}. As usual it is based on a differential graded Lie algebra whose underlying graded vector space is
\begin{align}
V = \Omega^\bullet(\CC_N)
\end{align}
with the non-zero brackets
\begin{align}
\ell_1(\alpha) = \dd_\ce\alpha \qquad \mbox{and} \qquad \ell_2(\alpha,\beta) = -[\alpha,\beta]
\end{align}
for $\alpha,\beta\in V$. Notably, in this framework the underlying cochain complex is markedly different from the classical case, as its differential is the Chevalley--Eilenberg differential $\dd_\ce$, rather than the de~Rham differential $\dd$. This is not so surprising since the geometry of the fuzzy sphere $\FS_N^2$ involves a three-dimensional differential calculus on a fuzzy version of a two-dimensional space. This is due to the imposition of $\frsu(2)$-covariance, which requires an embedding of the classical $2$-sphere in three-dimensional flat space: the extra dimension comes from the pullback of the tangent bundle $T_{\FR^3}\to\FR^3$ to $\FS^2$, which decomposes as the Whitney sum $T_{\FR^3}|_{\FS^2} = T_{\FS^2}\oplus N_{\FS^2}$ with $N_{\FS^2}\to\FS^2$ the normal bundle to the embedding $\FS^2\hookrightarrow\FR^3$. Higher-dimensional covariant calculi are a common feature of a broad range of examples in noncommutative geometry.

The gauge symmetry \eqref{eq:fuzzygauge}--\eqref{eq:fuzzyclosure} and field equation $F_{\tA}^\ce=0$ arise in the standard way from these brackets, as does the usual Bianchi identity 
\begin{align}
\dd^{\tA}_\ce F_{\tA}^\ce := \dd_\ce F_{\tA}^\ce + [A,F_{\tA}^\ce] = 0 
\end{align}
from the corresponding Noether identity of the $L_\infty$-algebra. We introduce a cyclic structure in the usual way by the pairing
\begin{align}\label{eq:fuzzypairing}
\langle\alpha,\beta\rangle = \int_{\FS_N^2} \, \alpha\wedge\beta
\end{align}
of forms $\alpha,\beta\in\Omega^\bullet(\CC_N)$ in complementary degrees, where the integration of top-forms $\Omega^3(\CC_N)$ is defined via the normalized $N{\times} N$ matrix trace
\begin{align}
\int_{\FS_N^2} \, \lambda \ \vartheta^1\wedge\vartheta^2\wedge\vartheta^3 := \frac{4\pi}N \, \Tr(\lambda) \ ,
\end{align}
where $\lambda\in\CC_N$. The Maurer--Cartan action functional 
\begin{align}\label{eq:MCfuzzyCS}
S(A) = \frac12\,\langle A,\ell_1(A)\rangle - \frac1{3!}\,\langle A,\ell_2(A,A)\rangle = \frac12\,\int_{\FS_N^2} \, A\wedge \dd_\ce A +\frac2{3}\,A\wedge A\wedge A 
\end{align}
is just the standard Chern--Simons  functional on the fuzzy sphere as written in~\cite{Grosse:2000gd}. It has a natural global $\sSU(2)$ symmetry given by rotations of the coframe $\vartheta^{a}$, and in the continuum limit $N\to\infty$ it describes a $BF$-type topological field theory on the classical $2$-sphere $\FS^2$. We shall give a more thorough account of general $BF$ theories in the $L_\infty$-algebra formalism in Section~\ref{sec:BFtheory}.

We can write the action functional \eqref{eq:MCfuzzyCS} in a more transparent form of a matrix model. For this, we write the gauge field $A\in\Omega^1(\CC_N)$ explicitly in component form $A=A_{a}\,\vartheta^{a}$ with $A_{a}\in{\sf Mat}(N)$ and define $\LL_{a}A_b:=\frac1{r_N}\,[X_{a},A_b]$. Then we can write \eqref{eq:MCfuzzyCS} as
\begin{align}
S(A) = \tfrac{2\pi}N \, \Tr\big(\epsilon^{abc} \, {\tt CS}_{abc}\big) \ ,
\end{align}
where
\begin{align}
{\tt CS}_{abc} =\LL_{a} A_b \, A_c +\tfrac13\,A_{a}\,[A_b,A_c] -\tfrac\ii2 \, \epsilon_{ab}{}^d \, A_d \, A_c \ .
\end{align}
This is the form of the fuzzy Chern--Simons functional derived in~\cite{Alekseev:2000fd} from a worldsheet conformal field theory analysis.\footnote{This form of the Chern--Simons functional can be extended to any semi-simple Lie algebra $\frg$ of higher rank by replacing the Levi--Civita symbol $\epsilon_{ab}{}^c$ with the structure constants $f_{ab}{}^c$ of a basis of $\frg$, in which case it is part of the low-energy effective field theory on a D-brane wrapping conjugacy classes in the corresponding compact semi-simple Lie group $\sG$~\cite{Alekseev:2000fd}. However, the noncommutative geometric origin of the action functional in these extensions is not clear.}

\subsubsection*{Yang--Mills theory on the fuzzy sphere}

To define Yang--Mills theory on $\FS_N^2$ we need an analogue of the Hodge operator $\ast_\ce:\Omega^k(\CC_N)\to\Omega^{3-k}(\CC_N)$, which is defined on the compatible coframe of $\Omega^\bullet(\CC_N)$ by
\begin{align}
\ast_\ce\unit := \tfrac1{3!} \, \epsilon_{abc}\,\vartheta^{a}\wedge \vartheta^b\wedge\vartheta^c \qquad \mbox{and} \qquad \ast_\ce\vartheta^{a} := \tfrac12\,\epsilon^{a}{}_{bc}\,\vartheta^b\wedge\vartheta^c \ , 
\end{align}
together with involutivity $(\ast_\ce)^2 = \unit$. The underlying graded vector space of the $L_\infty$-algebra is again
\begin{align}
	V = \Omega^\bullet(\CC_N) \ ,
\end{align}
and, with similar notation as in Section~\ref{sec:usualNCLinfty}, the differential is given by
\begin{align}
	\ell_1(\lambda) = \dd_\ce\lambda \ , \quad \ell_1(A) = \dd_\ce \ast_\ce  \dd_\ce 
	A \qquad  
	\mbox{and} \qquad \ell_1(\CA) = \dd_\ce \CA   \ ,
\end{align} 
the 2-brackets are
\begin{align}
	\begin{split} 
		\ell_{2}(\lambda_1,\lambda_2)= -[\lambda_1,\lambda_2] \qquad & , \qquad
		\ell_2 (\lambda, A)=-[\lambda, A] \ , \\[4pt]
		\ell_{2}(\lambda, \CA) = -[\lambda,\CA] \qquad , \qquad
		\ell_2(\lambda,\varLambda)&=-[\lambda,\varLambda] \qquad , \qquad
		\ell_{2}(A, \CA)= -[A,\CA] \ , \\[4pt]
		\ell_{2}(A_1, A_2)= - \dd_\ce \ast_\ce  [A_1, A_2] - & [A_1,\ast_\ce  \dd_\ce A_2] -  [\ast_\ce  \dd_\ce A_1, A_2] \ , 
	\end{split}
\end{align}
and the non-zero 3-bracket is
\begin{align}
	\begin{split}
		\ell_3(A_1,A_2,A_3) = -\big[A_1, \ast_\ce [A_2, A_3]\big]-\big[A_2,\ast_\ce [A_1, A_3]\big] - \big[\ast_\ce [A_1, A_2], A_3\big] \ .
	\end{split}
\end{align}

The Maurer--Cartan equation $\frF_{\tA}=0$ leads to the usual covariant equation of motion given by $\dd_\ce^{\tA}\ast_\ce F_{\tA}^\ce=0$, and, with the same cyclic inner product \eqref{eq:fuzzypairing}, the Maurer--Cartan action functional
\begin{align}
S(A) = \frac{1}{2}\, \langle A, \ell_{1}(A) \rangle - \frac{1}{3!}\, \langle A, \ell_2 (A,A)\rangle -\frac{1}{4!}\,  \langle A,\ell_3 (A,A,A) \rangle = \frac{1}{2}\, \int_{\FS_N^2}\, F_{\tA}^\ce \wedge  \ast_\ce  F_{\tA}^\ce 
\end{align}
is precisely the $\sSU(2)$-symmetric Yang--Mills functional on the fuzzy sphere as presented in~\cite{Grosse:2000gd}. Written explicitly in component form $F_{\tA}^\ce=\frac12\,F_{ab}\,\vartheta^{a}\wedge\vartheta^b$ with $F_{ab}\in{\sf Mat}(N)$, this becomes explicitly a matrix model
\begin{align}
S(A) = \tfrac\pi N \, \Tr\big(F_{ab} \, F^{ab}\big) \ ,
\end{align}
where 
\begin{align}
F_{ab} = \LL_{a} A_b - \LL_b A_{a} + [A_{a},A_b] - \epsilon_{ab}{}^c \, A_c
\end{align}
and $F^{ab}:= \delta^{ac}\,\delta^{bd}\,F_{cd}$, which is precisely the form of the fuzzy Yang--Mills functional derived in~\cite{Alekseev:2000fd}.

\subsection{Working with non-trivial fluxes}
\label{sec:nontrivialflux}

The two special classes of noncommutative field theories which we have considered thus far in this section have the special feature that they each admit a natural noncommutative differential calculus, which enabled a straightforward adaptation of techniques from ordinary nonabelian gauge theories to reformulate them algebraically in the formalism of standard $L_\infty$-algebras, even in the curved space example of the fuzzy sphere. As discussed in Section~\ref{sec:Intro}, this will not be the case in general for curved D-branes and for curved backgrounds. For a non-constant Poisson bivector $\theta = \frac12\,\theta^{\mu\nu}\,\partial_\mu\wedge\partial_\nu$, the deformation quantization of the worldvolume $M$ is described by the Kontsevich star-product~\cite{Kontsevich:1997vb} which up to second order in the deformation parameter $\hbar$ can be written in the form
\begin{align}\label{eq:Kontstar}
f\circledast g = f\star g - \tfrac{\hbar^2}{12}\,\theta^{\mu\rho}\,\partial_\rho\theta^{\nu\lambda} \, \big(\partial_\mu\partial_\nu f \, \partial_\lambda g + \partial_\mu\partial_\nu g \, \partial_\lambda f\big) + O(\hbar^3)
\end{align}
for fields $f$ and $g$ on $M$, where $f\star g$ is given by the standard Moyal--Weyl formula \eqref{eq:MoyalWeylstarprod}. 

The usual de~Rham differential $\dd$ is evidently not a derivation of this star-product unless $\theta$ is constant. This also implies that the naive definition \eqref{eq:stargauge} of noncommutative gauge transformations no longer satisfies the gauge closure property \eqref{eq:stargaugeclosure}. Moreover, in the case that the pullback of the NS--NS $3$-form flux to the D-brane worldvolume is non-zero, the bivector $\theta$ defines a \emph{twisted} Poisson structure and the star-product \eqref{eq:Kontstar} is no longer associative; in this case the star-commutator $[-,-]_\circledast$ fails to define a Lie bracket on $\Omega^0(M)$.

In~\cite{Blumenhagen:2018kwq} a homotopy algebraic approach to constructing generic noncommutative gauge theories was proposed and dubbed the `$L_\infty$-bootstrap'. According to this prescription, one starts with the natural brackets
\begin{align}
\ell_1(\lambda) = \dd\lambda \qquad \mbox{and} \qquad \ell_2(\lambda_1,\lambda_2) = -[\lambda_1,\lambda_2]_\circledast 
\end{align}
on the vector space $V^0 = \Omega^0(M)$ of gauge parameters, and attempts to construct the rest of the $L_\infty$-algebra structure by consistently solving the homotopy Jacobi identities order by order in the deformation parameter $\hbar$. The approach of~\cite{Kupriyanov:2021cws,Kupriyanov:2021aet}, based on symplectic embeddings of (almost) Poisson structures, improves this perspective because it can in principle compute all orders expressions, which are sometimes asymptotic expansions of analytic functions known in closed form. However, while these approaches suggest novel algebraic constructions which go beyond earlier approaches, they suffer from a more serious problem: so far they have only been developed in detail for the semi-classical limit of a full noncommutative gauge theory, where the star-commutator is approximated by the underlying almost Poisson bracket. It is not currently clear how these methods can be extended to construct complete noncommutative gauge theories, beyond their semi-classical approximations.

In the remainder of this contribution we will explain an alternative homotopy algebraic construction of noncommutative field theories, which does not handle generic (almost) Poisson structures, but is suggestive of a more systematic and elegant approach to dealing with the general problems discussed above. Rather than attempting to fit noncommutative gauge symmetry into the standard structure of an $L_\infty$-algebra, we adapt the $L_\infty$-algebra itself by deforming it in a way which is compatible with the underlying noncommutative algebra of differential forms. This yields a new perspective on the implementation of general symmetries in noncommutative field theories while producing new deformations of the standard gauge theories, and it introduces a new type of homotopy algebraic structure which to our knowledge has not been previously discussed in the mathematics literature. 

\section{Braided gauge symmetry and braided $L_\infty$-algebras}
\label{sec:braidedLinfty}

In this section we come to the main topic of this paper. We shall introduce the notion of braided gauge symmetry using Drinfel'd twist deformation techniques, comparing it to the more conventional star-gauge symmetries that were discussed in Section~\ref{sec:usualNCLinfty}. After going back to our simple prototypical example of Chern--Simons gauge theory from Section~\ref{sec:whatisgauge} and describing its braided noncommutative deformation, we introduce braided $L_\infty$-algebras as the organising algebraic principle behind noncommutative field theories with braided gauge symmetry. Throughout we stress the novelties behind these `braided field theories' as compared to the standard noncommutative field theories, and we discuss a braided version of the BV formalism which should serve as a useful tool for the quantization of these theories.

\subsection{Closing noncommutative gauge transformations}

In Section~\ref{sec:nontrivialflux} we pointed out issues involved in attempting to close an algebra of  star-gauge transformations \eqref{eq:stargauge} for generic (almost) Poisson structures. Let us discuss a somewhat different but equally well-known closure issue with these noncommutative gauge symmetries, which will serve to motivate the constructions which follow later on in this section. 

As in Section~\ref{sec:usualNCLinfty}, consider a noncommutative field theory defined with the Moyal--Weyl star-product \eqref{eq:MoyalWeylstarprod} for a constant Poisson bivector $\theta$ on flat space $M=\FR^{1,d-1}$. Let $\frg$ be a matrix Lie algebra, and endow the exterior algebra $\Omega^\bullet(M,\frg)=\Omega^\bullet(M)\otimes\frg$ of $\frg$-valued forms with the graded star-commutator
\begin{align}\label{eq:starcommfrg}
[\alpha\starcom\beta]_\frg := \alpha\wedge_\star\beta - (-1)^{|\alpha|\,|\beta|} \, \beta\wedge_\star\alpha
\end{align}
for homogeneous differential forms $\alpha,\beta\in\Omega^\bullet(M,\frg)$, where here the star-products are composed with the matrix multiplication in $\frg$. Again we define the \emph{star-gauge transformation} of a gauge field $A\in\Omega^1(M,\frg)$ by a gauge parameter $\lambda\in\Omega^0(M,\frg)$ as the naive deformation of a classical gauge transformation:
\begin{align}
\delta_\lambda^\star A = \dd\lambda - [\lambda\starcom A]_\frg \ .
\end{align}

The problem now is that these star-gauge variations do not generally close on the Lie algebra $\frg$: whereas the commutator of two gauge variations still closes onto the star-commutator as \smash{$\big[\delta_{\lambda_1}^\star,\delta_{\lambda_2}^\star\big]_\circ = \delta_{[\lambda_1\starcom\lambda_2]_\frg}^\star$}, the star-commutator itself does not close:
\begin{align}
[\lambda_1\starcom\lambda_2]_\frg = \lambda_1\star\lambda_2 - \lambda_2\star\lambda_1 \ \notin \ \Omega^0(M,\frg) \ .
\end{align}
Moreover, the naive deformation of the field strength
\begin{align}\label{eq:naiveNCcurv}
F_{\tA}^\star = \dd A +\tfrac12\,[A\starcom A]_\frg = \dd A + A\wedge_\star A \ \notin \ \Omega^2(M,\frg)
\end{align}
is no longer valued in $\frg$.
In other words, the star-commutator \eqref{eq:starcommfrg} does not generally define a Lie bracket on the exterior algebra $\Omega^\bullet(M,\frg)$. An easy calculation in a basis of $\frg$ shows that the exception is when $\frg$ is a matrix Lie algebra which is also closed under the anticommutator of matrices, such as $\frg=\mathfrak{gl}(n)$ or $\frg=\fru(n)$; if the gauge theory is coupled to matter fields then those must further be restricted to the trivial, (anti-)fundamental or adjoint representations of $\frg$. This is a well-known issue in noncommutative gauge theory, noted originally in~\cite{Terashima:2000xq} and subsequently elucidated by~\cite{Bonora:2000td,Chaichian:2001mu}. We note that exactly the same problem arises in the fuzzy field theories from Section~\ref{sec:fuzzyNCLinfty}: the extension of the graded matrix commutator \eqref{eq:matrixcommforms} to ${\sf Mat}(N)\otimes\frg$ does not close in $\frg$ unless $\frg$ is also closed under anticommutators.

Let us briefly discuss the two most common workarounds of the problem of defining Moyal--Weyl star-gauge theories for generic  Lie algebras $\frg$, each of which has its own drawbacks. Firstly, one can notice that the star-gauge closure takes place not in the Lie algebra $\frg$ but rather in its universal enveloping algebra $\sU\frg$, which is the free unital algebra over $\FR$ generated by $\frg$ modulo the two-sided ideal generated by $X\,Y - Y\,X - [X,Y]_\frg$ for all $X,Y\in\frg$. This leads to the notion of \emph{enveloping algebra-valued gauge symmetry}~\cite{Jurco:2000ja,AschCast}, generated by extended gauge parameters $\lambda\in\Omega^0(M,\sU\frg)$ acting on extended gauge fields $A\in\Omega^1(M,\sU\frg)$. While this extension is appealing because it makes sense for arbitrary (not necessarily matrix) Lie algebras $\frg$, its drawback is that it introduces infinitely many new degrees of freedom for the gauge fields $A$ due to the infinite dimensionality of the algebra $\sU\frg$. These new fields have no clear meaning in the underlying  gauge theory at $\hbar=0$, and so these noncommutative gauge theories do not possess good classical limits.

Alternatively, one can use the \emph{Seiberg--Witten map}~\cite{Seiberg:1999vs}, which was originally found as an equivalence of low-energy effective descriptions of open strings in a constant $B$-field background in terms of ordinary as well as noncommutative gauge theories. This identifies the noncommutative and classical gauge orbits  through a field redefinition which maps ordinary gauge parameters and fields $(\lambda,a)$ to noncommutative gauge parameters and fields $\big(\hat A(a),\hat\Lambda(\lambda,a)\big)$ defined by
\begin{align}
\hat A(a+\delta_\lambda a) = \hat A(a) + \delta_{\hat\Lambda(\lambda,a)}^\star\hat A(a) \ .
\end{align}
If the gauge theory is coupled to matter fields, then there is a further map from ordinary matter fields $\phi$ to noncommutative matter fields $\hat\Phi(\phi,a)$. 
In this case no new degrees of freedom are introduced, while the noncommutative variables are functions of the corresponding ordinary variables and the background deformation. Interesting new interaction vertices appear  which enable the construction of noncommutative gauge theories~\cite{Jurco:2001rq}; this has been used extensively over the years in phenomenological explorations of noncommutative field theories (see e.g.~\cite{Szabo:2009tn} for a review). 

The Seiberg--Witten map has a clear geometric interpretation which enables its extension to arbitrary Poisson structures~\cite{Jurco:2001my,Jurco:2001kp} and even to twisted Poisson structures~\cite{Aschieri:2002fq,Mylonas:2012pg}. It can also be interpreted as an $L_\infty$-quasi-isomorphism of the underlying $L_\infty$-algebras~\cite{Blumenhagen:2018shf}. However, the drawback of this approach is that it does not produce the intrinsically noncommutative differential graded algebra underlying a noncommutative gauge theory, along the lines we anticipated in Section~\ref{sec:Intro}, but rather maps the problem into a non-local classical field theory.

In the following we review a more recent approach to the construction of noncommutative gauge theories which circumvents these problems: these field theories are defined analytically in closed form for arbitrary gauge algebras $\frg$ and matter field representations without the introduction of spurious degrees of freedom, and they possess good classical limits to the corresponding classical field theories. Their construction is based on a precise and systematic homotopy algebraic organising principle which avoids the sometimes cumbersome ordering ambiguities in definitions of interaction vertices due to noncommutativity of fields (e.g. $\Phi_1\star A\star\Phi_2\neq\Phi_1\star\Phi_2\star A$): the entire noncommutative field theory is uniquely defined by the structure of an underlying `braided $L_\infty$-algebra' which is constructed as an unambiguous and precise noncommutative deformation of the standard $L_\infty$-algebra underlying the classical field theory. 

\subsubsection*{Noncommutative diffeomorphisms}

To help motivate how the braided algebraic structure appears, let us bring gravity into the present discussion and attempt to analyse its noncommutative deformation through the methods just explained (recall the motivation from Section~\ref{sec:Intro}). The symmetries underlying general relativity are spacetime diffeomorphisms which are generated by vector fields in $\frv:=\sfGamma(TM)$, which is a Lie algebra with the usual Lie bracket of vector fields which we denote by $[-,-]_\frv$. In the Moyal--Weyl deformation, we define the star-bracket $[-\starcom-]_\frv$ in the enveloping algebra $\sU\frv$ by regarding the vector space $\sfGamma(TM)$ as a module over the noncommutative algebra of functions $\big(C^\infty(M,\FR),\star\big)$, with the holonomic frame on $M=\FR^{1,d-1}$ obeying $f\star\partial_\mu = f\cdot\partial_\mu$ and $\partial_\mu\star f=\partial_\mu f$ for $f\in C^\infty(M,\FR)$, and $\partial_\mu\star\partial_\nu = \partial_\mu\,\partial_\nu = \partial_\nu\star\partial_\mu$ in $\sU\frv$. As above, if $\xi_1,\xi_2\in\frv$ are vector fields, then an easy calculation in components reveals that the star-bracket $[\xi_1\starcom\xi_2]_\frv\in \sU\frv$ is no longer a vector field in general. However, no analogue of the Seiberg--Witten map is generally known for deformed diffeomorphisms, which are instead most naturally defined through Drinfel'd twist techniques~\cite{TwistApproach,Aschieri:2005zs}. We now proceed to review these techniques, and then demonstrate how they can be fruitfully applied to internal gauge symmetries as well.

\subsection{Drinfel'd twist deformation theory}
\label{sec:Drinfeldtwist}

We shall now briefly review the basics of the Drinfel'd twist deformation formalism, which is a method to consistenly deform symmetries (described by a Hopf algebra). Here we only recall the salient aspects that we will use explicitly in the following; see~\cite{MajidBook,BeggsMajidbook} for a more detailed and complete account. A lot of what we say in the remainder of this section can be applied in a much more general context, and in particular to incorporate the fuzzy field theories discussed in Section~\ref{sec:fuzzyNCLinfty} (see~\cite{Nguyen:2021rsa}). For the examples we treat in the present paper, however, we will always start from a continuum field theory on a manifold $M$ and suitably deform it analogously to the Moyal--Weyl type noncommutative field theories. Hence we start from the Lie algebra of vector fields $\frv:=\sfGamma(TM)$, which generate infinitesimal diffeomorphisms of $M$, and later on specialize to suitable subalgebras appropriate to the spacetime symmetries of a given field theory example at hand.

The universal enveloping algebra $\sU\frv$ of the Lie algebra $\frv$ is naturally a cocommutative Hopf algebra with coproduct $\Delta:\sU\frv\to \sU\frv\otimes \sU\frv$ defined on
generators by
\begin{align}
\begin{split}
\Delta(\xi)&= \xi\otimes 1 + 1\otimes \xi\qquad \mbox{and} \qquad
             \Delta(1)=1\otimes 1 \ , 
\end{split}
\end{align}
for all $\xi\in \frv$. The map $\Delta$ is
extended as an algebra homomorphism to all of $\sU\frv$ (analogously to \eqref{eq:DeltaV}). Cocommutativity is the property $\tau\circ\Delta=\Delta$, where
$\tau$ is the transposition which interchanges the factors in a tensor
product of vector spaces; this is equivalent to saying that $\sU\frv$ has a triangular structure with trivial
$\RR$-matrix $1\otimes1$. The Hopf algebra structure is completed by defining a counit
and antipode whose explicit forms will not be needed in this paper. These structure maps together satisfy a compatibility condition with the product and the unit $1$ of $\sU\frv$ as a unital algebra~\cite{MajidBook,BeggsMajidbook}.

There is a symmetric
monoidal category ${}_{\sU\frv}\CCM$ whose objects are left
$\sU\frv$-modules and whose morphisms are equivariant maps (see
e.g.~\cite{Barnes:2014ksa}), that is, linear maps which commute with the action of
$\frv=\sfGamma(TM)$ via the trivial coproduct $\Delta$. Since $\sU\frv$ is a cocommutative Hopf
algebra, the braiding isomorphism of ${}_{\sU\frv}\CCM$ is just the trivial
transposition $\tau$.
A \emph{$\sU\frv$-module algebra} is an algebra in the category
${}_{\sU\frv}\CCM$. By this we mean an associative algebra $(\CA,\mu)$ with a left
$\sU\frv$-action $\triangleright: \sU\frv\otimes \CA
\rightarrow \CA$ which is compatible with the algebra multiplication
via the coproduct $\Delta$, that is
\begin{align}
X\triangleright\mu(a\otimes b)
  = \mu\big(\Delta(X)\triangleright(a\otimes b)\big)
\end{align}
for all $X\in \sU\frv$ and $a,b\in\CA$, where $\mu:\CA\otimes\CA\to\CA$
is the product on $\CA$. In the following we will drop the symbol
$\triangleright$ to simplify the notation, and usually omit the adjective `left' with the understanding that only left modules are considered in this paper. 

\begin{example}\label{ex:modulealg}
In our concrete applications to field theory later on, $(\CA,\mu)$ will be, for instance, the space of functions $\CA=C^\infty(M,\FR)$ on the manifold $M$ with $\mu=\,\cdot\,$ the pointwise product, the space of differential forms $\CA=\Omega^\bullet(M)$ with $\mu=\wedge$ the exterior product, or more generally the space of tensor fields with the tensor product. These are all modules over the Lie
algebra of vector fields $\frv$, where the (left) action of $\xi\in\frv$ on $a\in\CA$ is given by the Lie derivative $\LL_\xi$ along $\xi$:
\begin{align}
\xi(a) := \LL_\xi a \ .
\end{align}
This makes
$(\CA,\mu)$ into a $\sU\frv$-module
algebra, by virtue of the Leibniz rule for the Lie derivative:
\begin{align}
\xi \big(\mu(a\otimes b)\big) = \LL_{\xi}\circ\mu(a\otimes b)=\mu(\LL_{\xi} a\otimes b) + \mu(a\otimes \LL_{\xi} b) = \mu \circ \Delta(\xi) (a\otimes b) \ .
\end{align}
This is of course just the basis of diffeomorphism symmetry: The algebras ($\CA,\mu)$ are covariant under the action
of the universal enveloping algebra of infinitesimal diffeomorphisms
$\sU\frv$.
\end{example}

If $V$ is a real vector space, we denote by $V[[\hbar]]$ the vector space of formal power series in a deformation parameter
$\hbar$ with coefficients in the complexification $V_\FC=V\otimes \FC$;\footnote{We should more accurately write $V_\FC[[\hbar]]$, but we drop the subscript ${}_\FC$ to simplify the notation, as no confusion should arise.} it is naturally a module over
$\FC[[\hbar]]$. If $V$ and $W$ are real vector spaces, then
$V[[\hbar]]\otimes W[[\hbar]]\simeq(V\otimes W)[[\hbar]]$.\footnote{To keep the notation simple, we do not distinguish between the appropriate topological tensor product needed for formal power series extensions and the usual tensor product of vector spaces over $\FC$.} With these conventions, we denote by
$\sU\frv[[\hbar]]$ the formal power series extension of the
cocommutative Hopf algebra $\sU\frv$, with operations extended linearly over $\FC$ and applied term by term to the coefficients of series.

\begin{definition}\label{def:twist}
A \emph{Drinfel'd twist} on a manifold $M$ is an invertible element
$\CF\in \sU\frv[[\hbar]]\otimes \sU\frv[[\hbar]]$ satisfying the cocycle
condition
\begin{align}\label{eq:twistcocycle}
\CF_{12}\,(\Delta\otimes \unit)\CF=\CF_{23}\,(\unit\otimes \Delta)\CF \ ,
\end{align}
where $\CF_{12}=\CF\otimes 1$ and $\CF_{23}=1\otimes \CF$, together
with the normalization condition $\CF=1\otimes 1 + O(\hbar)$.
\end{definition}
We write the power series expansion of the twist as $\CF=:\sff^{k}\otimes
\sff_{k}\in \sU\frv[[\hbar]]\otimes \sU\frv[[\hbar]]$, with the sum over
$k$ understood, and likewise for the inverse twist $\CF^{-1}=: \bar{\sff}^{k}\otimes
\bar{\sff}_{k}\in \sU\frv[[\hbar]]\otimes \sU\frv[[\hbar]]$ which satisfies a
 condition similar to \eqref{eq:twistcocycle} that can be expressed as
\begin{align}\label{eq:invcocycle}
\bar\sff^k\otimes\Delta(\bar\sff_k)\,(\bar\sff^l\otimes\bar\sff_l) = \Delta(\bar\sff^k)\,(\bar\sff^l\otimes\bar\sff_l)\otimes\bar\sff_k \ .
\end{align}
The normalization condition ensures that the twist deformations induced by $\CF$ below have good classical limits at $\hbar=0$.

A Drinfel'd twist $\CF$ defines a new Hopf algebra structure on the
universal enveloping algebra $\sU\frv[[\hbar]]$, which we denote by
$\sU_\tCF\frv$. As algebras $\sU_\tCF\frv=\sU\frv[[\hbar]]$, while the new coproduct $\Delta_\tCF$ is given by
\begin{align}\label{eq:DeltaCF}
\Delta_{\tCF}(X):= \CF \, \Delta(X) \, \CF^{-1} \ ,
\end{align}
for all $X\in \sU\frv[[\hbar]]$. 
This new Hopf algebra is not cocommutative in general,
but the cocommutativity is controlled up to a
\emph{braiding} given by the invertible \emph{universal $\RR$-matrix} $\RR\in
\sU\frv[[\hbar]]\otimes \sU\frv[[\hbar]]$ induced by the twist as
\begin{align}
\RR=\CF_{21}\, \CF^{-1}=:\sfR^k\otimes\sfR_k \ ,
\end{align}
where $\CF_{21}=\tau(\CF)=\sff_k\otimes\sff^k$ is the twist with its legs
swapped. Explicitly
\begin{align}
\tau\circ\Delta_{\tCF}(X)=\RR\,\Delta_{\tCF}(X)\,\RR^{-1} \ .
\end{align}
It is easy to see that the $\RR$-matrix is triangular, that is
\begin{align}
  \RR_{21} = \RR^{-1} = \sfR_k\otimes\sfR^k \ ,
\end{align}
  and moreover that 
\begin{align}
(\Delta_{\tCF}\otimes \unit) \RR = \RR_{13}\, \RR_{23} \qquad \mbox{and} \qquad
(\unit\otimes \Delta_{\tCF})\RR= \RR_{13}\, \RR_{12} \ , 
\end{align}
where $\RR_{13}=\sfR^k\otimes 1 \otimes\sfR_k$, or equivalently
\begin{align}\label{eq:Rmatrixidsw}
\Delta_\tCF(\sfR^{k})\otimes \sfR_{k}= \sfR^{l}\otimes \sfR^{k}\otimes
  \sfR_{l}\, \sfR_{k} \qquad \mbox{and} \qquad \sfR^{k}\otimes
  \Delta_\tCF({\sfR_{k}}) =
  \sfR^{l}\,\sfR^{k}\otimes \sfR_{k} \otimes \sfR_{l} \ .
\end{align}
These identities tell us that passing an element ``at once'' over a pair of elements in a triple tensor product $\sU\frv[[\hbar]]\otimes \sU\frv[[\hbar]]\otimes\sU\frv[[\hbar]]$ is the same as passing successively over ``each one'' individually.

Drinfel'd twist deformation quantization consists in twisting the enveloping
Hopf algebra $\sU\frv$ to a non-cocommutative Hopf algebra $\sU_\tCF\frv$,
while simultaneously twisting all of its modules~\cite{SpringerBook}. A Drinfel'd twist $\CF$ defines a symmetric monoidal
category ${}_{\sU_\tCF\frv}\CCM$ of left
$\sU_\tCF\frv$-modules and equivariant maps, that is, linear maps which
commute with the action of $\sfGamma(TM)$ via the
twisted coproduct $\Delta_\tCF$. The category ${}_{\sU_\tCF\frv}\CCM$ is functorially equivalent to ${}_{\sU\frv}\CCM$~\cite{Barnes:2014ksa}. The braiding isomorphism $\tau_{\tCR}$ of
${}_{\sU_\tCF\frv}\CCM$ is now non-trivial and given by composing the transposition
$\tau$ with the action of the inverse $\RR^{-1}$ of the universal $\RR$-matrix. Since $\RR$ is
triangular, $\RR_{21}=\RR^{-1}$, the braiding is symmetric,
i.e. the braiding isomorphism $\tau_{\tCR}$ squares to the identity
morphism.

If $(\CA,\mu)$ is a (left) $\sU\frv$-module algebra, then we
can deform the product $\mu$ on $\CA$ by precomposing it with the
inverse of the twist
$\CF$ to get a new product
\begin{align}\label{eq:mustar}
\mu_\star (a\otimes b) = \mu\circ\CF^{-1}(a\otimes b) =
  \mu\big(\bar\sff^k(a)\otimes\bar\sff_k(b)\big) \ ,
\end{align}
for $a,b\in\CA$, where on the right-hand side we extend $\mu$ (linearly over $\FC$) to
$\CA[[\hbar]]\otimes\CA[[\hbar]] \simeq(\CA\otimes\CA)[[\hbar]]$ by
applying it term by term to the coefficients of a formal power series.
The cocycle condition \eqref{eq:twistcocycle} on $\CF$ guarantees that this
produces an associative star-product $\mu_\star$ on $\CA[[\hbar]]$, and it generally defines a
noncommutative $\sU_\tCF\frv$-module algebra $(\CA[[\hbar]],\mu_\star)$,
that is, an algebra 
in the category of $\sU_\tCF\frv$-modules, because it carries a representation of the twisted Hopf algebra $\sU_\tCF\frv$:
\begin{align}
\begin{split}
X\big(\mu_\star(a\otimes b)\big) &= X\big(\mu\circ\CF^{-1}(a\otimes b)\big) \\[4pt]
&=\mu\circ\Delta(X)\big(\CF^{-1}(a\otimes b)\big) \\[4pt]
&=\mu\circ\CF^{-1}\big(\Delta_\tCF(X)(a\otimes b)\big) = \mu_\star\big(\Delta_\tCF(X)(a\otimes b)\big) \ ,
\end{split}
\end{align}
for all $X\in \sU\frv$ and $a,b\in\CA$.

If the algebra $(\CA,\mu)$ is commutative, i.e. $\mu\circ\tau=\mu$,
then $(\CA[[\hbar]],\mu_\star)$ is \emph{braided commutative}, i.e.~$\mu_\star\circ\tau_{\tCR}=\mu_\star$. The noncommutativity in this case is controlled by the triangular $\RR$-matrix as
\begin{align}
\mu_\star(a\otimes b) = \mu_\star\big(\sfR_k(b)\otimes\sfR^k(a)\big) \ ,
\end{align}
which is easily seen by recalling that $\RR=\CF_{21}\,
\CF^{-1}$. 

Let us now illustrate the formalism above by giving some concrete examples of Drinfel'd twists which commonly appear in noncommutative field theory.

\begin{example}\label{ex:MWtwist}
The \emph{Moyal--Weyl twist} on $M=\FR^{1,d-1}$ is given by
\begin{align}\label{eq:MWtwist}
\begin{split}
\CF &=
  \exp\big(-\tfrac{\mathrm{i}\,\hbar}2\,\theta^{\mu\nu}\,\partial_\mu\otimes\partial_\nu\big) \\[4pt]
  &=1\otimes1 + \sum_{n=1}^\infty\, \Big(-\frac{\ii\,\hbar}2\Big)^n \, \frac1{n!} \, \theta^{\mu_1\nu_1}\cdots\theta^{\mu_n\nu_n} \, \partial_{\mu_1}\cdots\partial_{\mu_n} \, \otimes \,  \partial_{\nu_1}\cdots\partial_{\nu_n}=:\sff^k\otimes\sff_{ k}
  \ ,
\end{split}
\end{align}
where $(\theta^{\mu\nu})$ is a constant $d{\times}d$ antisymmetric real-valued
matrix. This twist is built on the enveloping Hopf algebra $\sU\FR^{1,d-1}\subset\sU\sfGamma(T\FR^{1,d-1})$ of
the abelian Lie algebra of infinitesimal spacetime translations. In this case $\CF_{21}=\CF^{-1}$ so that the universal
$\RR$-matrix is given by
\begin{align}
\RR = \CF^{-2} =
  \exp\big(\,\mathrm{i}\,\hbar\,\theta^{\mu\nu}\,\partial_\mu\otimes\partial_\nu\big) =: \sfR^k\otimes\sfR_k
  \ .
\end{align}

The twist deformation $\big(C^\infty(\FR^{1,d-1})[[\hbar]],\mu_\star\big)$ of the algebra of functions on $\FR^{1,d-1}$ gives precisely the Moyal--Weyl star-product \eqref{eq:MoyalWeylstarprod}. More generally, the twist deformation $\big(\Omega^\bullet(\FR^{1,d-1})[[\hbar]], \wedge_\star\big)$ of the exterior algebra of differential forms on $\FR^{1,d-1}$ gives the noncommutative differential calculus discussed in Section~\ref{sec:usualNCLinfty}; for the holonomic coframe  $\LL_{\partial_\nu}\dd x^\mu=0$, so that indeed
$\dd x^\mu\wedge_\star\dd x^\nu = \dd x^\mu\wedge \dd x^\nu $ and $\lambda\star\dd x^\mu = \dd x^\mu\star \lambda = \lambda\cdot\dd x^\mu $ for $\lambda\in\Omega^0(\FR^{1,d-1})$.\footnote{For these and other considerations in this paper, the Cartan formula $\LL_\xi=\dd\circ\imath_\xi + \imath_\xi\circ\dd$ for the action of the Lie derivative on forms is very useful, where $\imath_\xi$ is the interior multiplication with the vector field $\xi\in\frv$.} Similarly, for the holonomic frame $\LL_{\partial_\mu}\partial_\nu = [\partial_\mu,\partial_\nu]_\frv = 0$, so that indeed $\lambda\star\partial_\mu = \lambda\cdot\partial_\mu$.

The Moyal--Weyl twist \eqref{eq:MWtwist} has a natural extension from $\FR^{1,d-1}$ to any manifold $M$ with a non-constant Poisson bivector $\theta$ whose components in a local holonomic frame satisfy
\begin{align}
\theta^{\mu\rho}\,\partial_\rho\theta^{\nu\lambda} = 0 \ ,
\end{align}
see \eqref{eq:Kontstar}. An important class of examples arises in the noncommutative gauge theories on D-branes in non-geometric T-duals to string theory compactifications with duality twists, see e.g.~\cite{Hull:2019iuy,Aschieri:2020uqp}.

\end{example}

\begin{example}\label{ex:angulartwist}
The \emph{angular twist} on $M=\FR^{1,d-1}$ is given by~\cite{DimitrijevicCiric:2018blz}
\begin{align}
\begin{split}
\CF &= \exp\big(-\tfrac{\mathrm{i}\,\hbar}2\,(\partial_0\otimes\sM_{\theta} - \sM_{\theta}\otimes\partial_0)\big)  \\[4pt]
&= 1\otimes1 + \sum_{n=1}^\infty \, \Big(-\frac{\ii\,\hbar}2\Big)^n  \ \sum_{l=0}^n \, \frac{(-1)^l}{l! \, (n-l)!} \, (\partial_0)^{n-l}\, (\sM_{\theta})^l \, \otimes \, (\partial_0)^l\, (\sM_{\theta})^{n-l} =: \sff^k\otimes\sff_k \ ,
\end{split}
\end{align}
where 
\begin{align}
\sM_\theta:= \theta^{ij}\,\sM_{ij} \qquad \mbox{with} \quad \sM_{ij} = x_i\,\partial_j - x_j\,\partial_i
\end{align}
for a constant $(d{-}1){\times}(d{-}1)$ antisymmetric real matrix $(\theta^{ij})$ and $i,j=1,\dots,d-1$. Here
we denote coordinates on $\FR^{1,d-1}$ by $(x^\mu) = \big(x^0,(x^i)\big) = (x^0,x^1,\dots,x^{d-1})$ with $x^0$ the time coordinate. This twist is built on the enveloping Hopf subalgebra of $\sU\sfGamma(T\FR^{1,d-1})$ generated by the direct product Lie algebra $\FR\times\mathfrak{so}(d-1)$ of infinitesimal time translations and spatial rotations. The corresponding triangular $\RR$-matrix is again given by $\RR=\CF^{-2}$. 

For the corresponding twist deformation $\big(\Omega^\bullet(\FR^{1,d-1})[[\hbar]], \wedge_\star\big)$ of the exterior algebra, the holonomic coframe again has undeformed exterior products $\dd x^\mu\wedge_\star \dd x^\nu = \dd x^\mu\wedge\dd x^\nu$. To describe the deformation of the $C^\infty(\FR^{1,d-1})$-bimodule structure it is convenient to perform a spatial rotation of the coordinates which brings the matrix $(\theta^{ij})$ into its Jordan canonical form with non-zero skew-eigenvalues $\theta^{a}$ for $a=1,\dots,r$, where $2r$ is the rank of $(\theta^{ij})$. Then
\begin{align}
\begin{split}
\dd x^{2a-1}\star\lambda &= \cos(\ii\,\hbar\,\theta^{a}\,\partial_0)\lambda\star \dd x^{2a-1} + \sin(\ii\,\hbar\,\theta^{a}\,\partial_0)\lambda\star\dd x^{2a} \ , \\[4pt]
\dd x^{2a}\star\lambda &= -\sin(\ii\,\hbar\,\theta^{a}\,\partial_0)\lambda\star \dd x^{2a-1} + \cos(\ii\,\hbar\,\theta^{a}\,\partial_0)\lambda\star\dd x^{2a} \ ,
\end{split}
\end{align}
for $a=1,\dots,r$ and $\lambda\in\Omega^0(\FR^{1,d-1})$, 
whereas $\dd x^0\star\lambda=\lambda\star\dd x^0$ and $\dd x^i\star\lambda = \lambda\star\dd x^i$ for $i=2r+1,\dots,d-1$.
\end{example}

\begin{example}\label{ex:abeliantwist}
Both the Moyal--Weyl twist of Example~\ref{ex:MWtwist} and the angular twist of Example~\ref{ex:angulartwist} are special instances of a more general class of twists called \emph{abelian twists}~\cite{Aschieri:2005zs}. These can be defined on \emph{any} manifold $M$ and are given by
\begin{align}
\begin{split}
\CF &= \exp\big(-\tfrac{\mathrm{i}\,\hbar}2\,\theta^{ab}\,X_{a}\otimes X_b\big) \\[4pt]
  &=1\otimes1 + \sum_{n=1}^\infty\, \Big(-\frac{\ii\,\hbar}2\Big)^n \, \frac1{n!} \, \theta^{a_1b_1}\cdots\theta^{a_nb_n} \, X_{a_1}\cdots X_{a_n} \, \otimes \,  X_{b_1}\cdots X_{b_n}=:\sff^k\otimes\sff_{ k}
  \ ,
\end{split}
\end{align}
where $(\theta^{ab})$ is a constant  antisymmetric real-valued
matrix and $\{X_{a}\}\subset\sfGamma(TM)$ is a set of mutually commuting vector fields, $[X_{a},X_b]_\frv=0$. These twists are built from an abelian subalgebra of the Lie algebra $\frv$ of vector fields on $M$. For such twists the universal $\RR$-matrix is always given as $\RR=\CF^{-2}$. They are a common choice in noncommutative geometry and field theory because they share many of the special simplifying properties enjoyed by the Moyal--Weyl twist. In particular, in an open neighbourhood of every point of $M$ outside a set of measure zero, an abelian twist is equivalent to a Moyal--Weyl twist via a suitable change of local frame~\cite{Aschieri:2009qh}. This is useful for practical considerations, as it means that the legs of the twist $\CF$ commute with the local basis vector fields in almost every open set.
\end{example}

\begin{example}\label{ex:Jordantwist}
Twists need not only be associated to commuting vector fields. As an example of a Drinfel'd twist which is \emph{not} an abelian twist, consider the \emph{Jordanian twist} on $M=\FR^{1,d-1}$ given by~\cite{Aschieri:2017ost}
\begin{align}
\begin{split}
\CF = \exp\big(-\sD \otimes \log(1-\ii\,\hbar\,\partial_0)\big) = 1\otimes1 + \sum_{n=1}^\infty \, \frac{(\ii\,\hbar)^n}{n!} \, \sD^{\underline n} \, \otimes \, (\partial_0)^n =: \sff^k\otimes\sff_k \ ,
\end{split}
\end{align}
where 
\begin{align}
\sD:=x^\mu\,\partial_\mu \ ,
\end{align}
and the lower factorial $X^{\underline n}\in\sU\frv$ of a vector field $X\in\frv$ for $n\geq1 $ is defined by
\begin{align}
X^{\underline n} = X\,(X-1\big)\cdots\big(X-(n-1)\big) \ .
\end{align}
This is built on the enveloping Hopf algebra of the semi-direct product Lie algebra $\FR\ltimes\FR$ of infinitesimal dilatations and time translations. The corresponding universal $\RR$-matrix is given by
\begin{align}
\begin{split}
\RR &= \CF_{21} \, \CF^{-1} = \exp\big(-\log(1-\ii\,\hbar\,\partial_0)\otimes \sD\big) \exp\big(\sD\otimes \log(1-\ii\,\hbar\,\partial_0)\big) \\[4pt]
&= 1\otimes1 + \sum_{n=1}^\infty \, (-\ii\,\hbar)^n  \ \sum_{l=0}^n \, \frac{(-1)^l}{l! \, (n-l)!} \, (\partial_0)^{n-l}\, \sD^{\underline l} \, \otimes \, (\partial_0)^l\, \sD^{\underline{n-l}}
=: \sfR^k\otimes\sfR_k \ .
\end{split}
\end{align}

The holonomic coframe $\dd x^\mu$ of the deformed exterior algebra $\big(\Omega^\bullet(\FR^{1,d-1})[[\hbar]], \wedge_\star\big)$ once again has trivial star-exterior products $\dd x^\mu \wedge_\star\dd x^\nu = \dd x^\mu\wedge \dd x^\nu$. The $C^\infty(\FR^{1,d-1})$-bimodule structure is non-trivial and given by
\begin{align}
\lambda\star\dd x^\mu = \dd x^\mu\star (1-\ii\,\hbar\,\partial_0)^{-1}\lambda \ ,
\end{align}
for all $\lambda\in\Omega^0(\FR^{1,d-1})$. 

These twists appear in $\kappa$-deformed field theories on Minkowski spacetimes which are covariant under the action of a  $\kappa$-Poincar\'e quantum group~\cite{Aschieri:2017ost}. Jordanian twists can be more generally associated to the Borel subalgebra of any Lie algebra $\frg\subset\frv=\sfGamma(TM)$~\cite{Aschieri:2005zs}. 
By further combining these with the abelian twists of Example~\ref{ex:abeliantwist},  we obtain more general Drinfel'd twists which provide realizations of the extended Jordanian twists of~\cite{Kulish:1998be}. Such twists arise in the holographic dual gauge theories to backgrounds obtained via Yang--Baxter deformations of string theory in $\mathbbm{AdS}_5{\times}\FS^5$~\cite{vanTongeren:2015uha}.
\end{example}

\subsection{Braided gauge symmetry and noncommutative kinematics}
\label{sec:braidedgaugesym}

Let us now spell out the kinematical ingredients of braided noncommutative gauge symmetry, incorporating dynamics later on. Let $\CF\in\sU\frv[[\hbar]]\otimes\sU\frv[[\hbar]]$ be any Drinfel'd twist on a manifold $M$, and let $\frg$ be any Lie algebra. As in Section~\ref{sec:whatisgauge}, let $\Omega^\bullet(M,\frg)=\Omega^\bullet(M)\otimes\frg$ be the exterior algebra of $\frg$-valued differential forms on $M$ with graded Lie bracket $[-,-]_\frg$ given by the tensor product of exterior multiplication with the Lie bracket of $\frg$. The graded Lie algebra $\big(\Omega^\bullet(M,\frg),[-,-]_\frg\big)$ is naturally a $\sU\frv$-module with $\sU\frv$ acting trivially on $\frg$ and via the Lie derivative on $\Omega^\bullet(M)$ (see Example~\ref{ex:modulealg}), that is, it is a graded Lie algebra in the category ${}_{\sU\frv}\CCM$. 

Drinfel'd twist deformation then produces a $\sU_\tCF\frv$-module $\big(\Omega^\bullet(M,\frg)[[\hbar]],[-,-]_\frg^\star\big)$ with the star-bracket
\begin{align}
[\alpha,\beta]_\frg^\star:= [-,-]_\frg\circ\CF^{-1}(\alpha\otimes \beta) = [\bar\sff^k(\alpha) , \bar\sff_k(\beta)]_\frg \ .
\end{align}
This bracket makes $\Omega^\bullet(M,\frg)[[\hbar]]$ into a \emph{braided} (or \emph{quantum}) \emph{graded Lie algebra} in the sense of~\cite{Woronowicz:1989rt,Majid:1993yp}, that is, a graded Lie algebra in the category
${}_{\sU_\tCF\frv}\CCM$~\cite{Barnes:2015uxa}. It is braided graded antisymmetric: 
\begin{align}
[\alpha_{1},\alpha_{2}]_\frg^{\star}=-(-1)^{|\alpha_1|\,|\alpha_2|} \,
  [\sfR_{k}(\alpha_{2}),\sfR^{k}(\alpha_{1})]_\frg^\star 
  \ ,
\end{align}
and satisfies the braided graded Jacobi identity:
\begin{align}\label{eq:braidedJacobiforms}
[\alpha_{1},[\alpha_{2},\alpha_{3}]_\frg^{\star}]_\frg^{\star}
  =[[\alpha_{1},\alpha_{2}]_\frg^{\star},\alpha_{3}]_\frg^{\star} +
  (-1)^{|\alpha_1|\,|\alpha_2|}\,
  [\sfR_{k}(\alpha_{2}),[\sfR^{k}(\alpha_{1}) ,\alpha_{3}]_\frg^{\star}]_\frg^{\star}
  \ ,
\end{align}
for all homogeneous $\alpha_{1},\alpha_{2},\alpha_{3} \in
\Omega^\bullet(M,\frg)$. This braided graded Lie algebra has a natural braided Lie subalgebra $\big(\Omega^0(M,\frg)[[\hbar]],[-,-]_\frg^\star\big)$ which will play an important role below.

It is important to note that this definition makes sense for \emph{any} Lie algebra $\frg$. In the special instance that $\frg$ is a matrix Lie algebra, the braided commutator does \emph{not} coincide with the graded star-commutator \eqref{eq:starcommfrg}:
\begin{align}\label{eq:braidedcomm}
[\alpha,\beta]_\frg^\star=\alpha\wedge_\star\beta - (-1)^{|\alpha|\,|\beta|} \,  \sfR_k(\beta)\wedge_\star\sfR^k(\alpha) \ \neq \ [\alpha\starcom\beta]_{\frg} \ ,
\end{align}
where we used $\CF_{21}^{-1}=\CF^{-1}\,\RR^{-1}=\CF^{-1}\,\RR_{21}$. In particular, if $\frg$ is an abelian Lie algebra, then the braided commutator \eqref{eq:braidedcomm} always vanishes on $\Omega^0(M,\frg)[[\hbar]]$, while the star-commutator \eqref{eq:starcommfrg} does not. 

Recall that the de~Rham differential $\dd:\Omega^\bullet(M)\to\Omega^\bullet(M)$ makes $\big(\Omega^\bullet(M,\frg),[-,-]_\frg\big)$ into a differential graded Lie algebra, as a consequence of the Leibniz rule on the underlying exterior algebra. By the Cartan structure equations, $\dd\circ\LL_\xi=\LL_\xi\circ\dd$ for all $\xi\in\frv$, so the exterior derivative commutes with the Lie derivatives that enter into the definition of the star-exterior product via the inverse twist $\CF^{-1}$. It follows that the standard undeformed de~Rham differential is also an ordinary graded derivation of the underlying star-exterior algebra $\big(\Omega^\bullet(M)[[\hbar]],\wedge_\star\big)$:
\begin{align}
\dd(\alpha\wedge_{\star} \beta) = \dd\alpha\wedge_\star\beta +
  (-1)^{|\alpha|}\,\alpha\wedge_\star\dd\beta \ .
\end{align}
This makes $\big(\Omega^\bullet(M,\frg)[[\hbar]],[-,-]_\frg^\star,\dd\big)$ into a differential graded braided Lie algebra, and provides a natural noncommutative differential calculus for the construction of field theories on $M$.

We are now ready to define notions of braided gauge fields and matter fields. The idea is to deform representations of the Lie algebra $\frg$ in which fields transform to braided representations of $\frg$, of which there are two types, regarded as acting from the `left' and from the `right'. The \emph{left} and \emph{right braided gauge transformations} of a gauge field $A\in\Omega^1(M,\frg)$ by a gauge parameter $\lambda\in\Omega^0(M,\frg)$ are given respectively by twisting the action of a classical gauge variation $\delta_\lambda$ as 
\begin{align}\label{eq:braidedvarA}
\begin{split}
{}^\lact\delta_\lambda^{\br}A &:= \delta_{\bar\sff^k(\lambda)}\bar\sff_k(A) =\dd\lambda-[\lambda,A]_\frg^\star \ , \\[4pt] {}^\ract\delta_\lambda^{\br}A &:= {}^\lact\delta_{\sfR_k(\lambda)}^\br\sfR^k(A)=\delta_{\bar\sff_k(\lambda)}\bar\sff^k(A)=\dd\lambda+[A,\lambda]_\frg^\star \ ,
\end{split}
\end{align}
where in the second line we again used $\CF^{-1}\,\RR_{21} = \CF_{21}^{-1}$.
Since $[A,\lambda]_\frg^\star = - [\sfR_k(\lambda),\sfR^k(A)]_\frg^\star$, these are different transformations. In contrast, star-gauge transformations do not see this left/right distinction, as $[A\starcom\lambda]_\frg = -[\lambda\starcom A]_\frg$. 

Similarly, if $\phi\in\Omega^p(M,W)$ is a $p$-form matter field in a linear representation $W$ of $\frg$, then its left and right gauge transformations are defined respectively by
\begin{align}\label{eq:braidedvarphi}
{}^\lact\delta_\lambda^{\br} \phi=-\lambda\star \phi \qquad \mbox{and} \qquad {}^\ract\delta_\lambda^{\br} \phi=-\sfR_k(\lambda)\star \sfR^k(\phi) \ ,
\end{align} 
where here $\lambda\star\phi := \bar\sff^k(\lambda)\cdot\bar\sff_k(\phi)$ is given by the tensor product of multiplication of forms by functions and the $\frg$-action on $W$, and $\sfR_k(\lambda)\star \sfR^k(\phi) = \bar\sff_k(\lambda)\cdot\bar\sff^k(\phi)$. As left and right braided gauge symmetries behave in a completely analogous way~\cite{Ciric:2021rhi}, we will consider only left variations from now on and drop the superscript ${}^\lact$ from the notation for simplicity; we will usually also omit the adjective `left'.

By setting $\delta_\lambda^\br\big|_{\Omega^0(M,\frg)} = 0$ and $\CV:=\Omega^\bullet(M,\frg)[[\hbar]]\oplus\Omega^\bullet(M,W)[[\hbar]]$, a braided gauge transformation is extended to a map \smash{$\delta_\lambda^\br:\CV\to \CV$} of bidegree~$(0,0)$ by commuting with arbitrary maps $\CV\to\CV$ and as a braided derivation of operations $\mu_\star$ defined on $\CV\otimes \CV$ in the category ${}_{\sU_\tCF\frv}\CCM$, that is, it obeys the \emph{braided Leibniz rule}
\begin{align}\label{eq:braidedLeibniz}
\delta_\lambda^\br\mu_\star(\upsilon_1\otimes \upsilon_2) = \mu_\star\big(\delta_\lambda^\br\upsilon_1\otimes\upsilon_2\big) + \mu_\star\big(\sfR_k(\upsilon_1)\otimes\delta^\br_{\sfR^k(\lambda)}\upsilon_2\big)
\end{align}
for all $\upsilon_1,\upsilon_2\in\CV$. Generally, only braided (not ordinary) derivations are compatible with braided symmetry properties of such maps $\mu_\star$~\cite{Ciric:2021rhi}. This rule also defines the tensor products of braided representations. They furthermore close under the \emph{braided commutator}
\begin{align}\label{eq:braidedclose}
\big[\delta_{\lambda_1}^\br,\delta_{\lambda_2}^\br\big]_\circ^\star  :=  \delta_{\lambda_1}^\br\circ\delta_{\lambda_2}^\br  - \delta_{\sfR_k(\lambda_2)}^\br\circ\delta^\br_{\sfR^k(\lambda_1)} = \delta^\br_{[\lambda_1,\lambda_2]_\frg^\star} \ .
\end{align}
The braided commutator generally makes the braided derivations of a $\sU_\tCF\frv$-module algebra into a braided Lie algebra~\cite{Barnes:2015uxa}. That \eqref{eq:braidedvarA} and \eqref{eq:braidedvarphi} give the representation \eqref{eq:braidedclose} of the braided Lie algebra $\big(\Omega^0(M,\frg)[[\hbar]],[-,-]_\frg^\star\big)$ on $\CV$ is proven in~\cite{Ciric:2021rhi}.

Braided noncommutative kinematics is now defined by the \emph{braided left and right covariant derivatives}
\begin{align}\label{eq:leftrightcovderiv}
\dd_{\star\lact}^{\tA}\phi := \dd \phi+A\wedge_\star \phi \qquad \mbox{and} \qquad \dd_{\star\ract}^{\tA}\phi := \dd \phi+\sfR_k(A)\wedge_\star \sfR^k(\phi)
\end{align}
in $\Omega^{p+1}(M,W)[[\hbar]]$, where $A\wedge_\star\phi = \bar\sff^k(A)\wedge\bar\sff_k(\phi)$ is given by the tensor product of exterior multiplication of forms with the action of $\frg$ on the representation $W$, and $\sfR_k(A)\wedge_\star \sfR^k(\phi) = \bar\sff_k(A)\wedge\bar\sff^k(\phi)$. They are both \emph{braided covariant}:
\begin{align}
\delta_\lambda^\br\big( \dd_{\star\lact}^{\tA}\phi\big) = -\lambda\star\big( \dd_{\star\lact}^{\tA}\phi\big) \qquad \mbox{and} \qquad \delta_\lambda^\br\big( \dd_{\star\ract}^{\tA}\phi\big) = -\lambda\star\big( \dd_{\star\ract}^{\tA}\phi\big)\ ,
\end{align}
as proven in~\cite{Ciric:2021rhi}. 

Similarly, one defines the \emph{braided curvature}
\begin{align}\label{eq:braidedcurvature}
F_{\tA}^\br := \dd A+\tfrac12\,[A,A]_\frg^\star \ \in \ \Omega^2(M,\frg)[[\hbar]] \ ,
\end{align}
where now there is no left/right distinction due to braided symmetry of the bracket $[-,-]_\frg^\star$ on $1$-forms, i.e. $[A,A]_\frg^\star = [\sfR_k(A),\sfR^k(A)]_\frg^\star$. The braided curvature is also braided covariant:
\begin{align}\label{eq:braidedcovcurv}
\delta_\lambda^\br F_{\tA}^\br=-[\lambda,F_{\tA}^\br]_\frg^\star \ .
\end{align}
Note that for a matrix Lie algebra $\frg$, one can write
\begin{align}
F^\br_{\tA} = \dd A + \tfrac12\,\big(A\wedge_\star A + \sfR_k(A)\wedge_\star\sfR^k(A)\big)  \ \neq \ F_{\tA}^\star \ ,
\end{align}
which differs from the usual field strength \eqref{eq:naiveNCcurv} in conventional noncommutative gauge theories (see e.g. \eqref{eq:fieldstrengthstargauge} and \eqref{eq:fieldstrengthfuzzy}).

\subsubsection*{Braided diffeomorphisms}

The discussion above carries through \emph{verbatum} to spacetime symmetries of field theories, such as the diffeomorphism symmetry of gravity, which is in fact where the formalism was originally developed (see e.g.~\cite{Aschieri:2005zs,SpringerBook,NAGravity}) and which inspired the present treatment of gauge symmetries. This is generated infinitesimally by vector fields $\xi\in\frv=\sfGamma(TM)$ acting via the Lie derivative $\LL_\xi$ on generic tensor fields $T$: $\delta_\xi T = \LL_\xi\, T$ (see Example~\ref{ex:modulealg}). The Cartan structure equations tell us that this gives a representation of the Lie algebra $(\frv,[-,-]_\frv)$ of vector fields on tensor fields: $[\LL_{\xi_1},\LL_{\xi_2}]_\circ = \LL_{[\xi_1,\xi_2]_\frv}$. 

This is obviously a $\sU\frv$-module, and following the twisting procedure above we define {braided diffeomorphisms} $\delta_\xi^\br T = \LL_\xi^\star\, T$ through the \emph{braided Lie derivative} $\LL_\xi^\star$ given by
\begin{align}
\LL_\xi^\star\, T := \LL_{\bar \sff^k(\xi)}\bar \sff_k( T) \ .
\end{align}
The braided Lie derivative is a braided derivation of the star-tensor product $\otimes_\star = \otimes\circ\CF^{-1}$ which closes the braided Lie algebra $\big(\frv[[\hbar]],[-,-]_\frv^\star\big)$:
\begin{align}
\big[\LL_{\xi_1}^\star,\LL_{\xi_2}^\star\big]_\circ^\star = \LL_{[\xi_1,\xi_2]_\frv^\star}^\star \ .
\end{align}
Note that the braided Lie bracket of vector fields may be represented as
\begin{align}
[\xi_1,\xi_2]_\frv^\star = \xi_1\star\xi_2 - \sfR_k(\xi_2)\star\sfR^k(\xi_1) \ \neq \ [\xi_1\starcom\xi_2]_\frv
\end{align}
in the twist deformation of the universal enveloping algebra $\sU\frv$, which by construction is again a vector field in $\frv[[\hbar]]$, in contrast to the star-bracket $[\xi_1\starcom\xi_2]_\frv$.

It is worth stressing once more that these braided diffeomorphisms are \emph{not} the same as the `twisted diffeomorphisms' of~\cite{TwistApproach}. Twisted diffeomorphisms act through the twisted enveloping Hopf algebra $\sU_\tCF\frv$, whose underlying algebra is the same as that of the classical Hopf algebra $\sU\frv$. In particular, an infinitesimal twisted diffeomorphism of a single field simply follows the classical rule with the usual Lie derivative $\LL_\xi$ so that they also close on the classical Lie algebra $(\frv,[-,-]_\frv)$ of diffeomorphisms, while on star-products of fields they follow a deformed Leibniz rule defined by the twisted coproduct $\Delta_\tCF(\LL_\xi)$ from \eqref{eq:DeltaCF}.

\subsection{Example: Braided Chern--Simons theory}
\label{sec:braidedCS}

As a simple illustration of a field theory with braided gauge symmetry where we can explore some of the novelties that arise, let us go back to the beginning of our story with the simplest prototypical example, namely the Chern--Simons gauge theory from Section~\ref{sec:whatisgauge}, and consider its braided noncommutative deformation. We define the braided noncommutative version of the Chern--Simons functional \eqref{eq:CSaction} in the obvious way as 
\begin{align}\label{eq:CSbraidedaction}
S_\star(A)=\int_M\,\Tr_\frg\Big(\frac12\,A\wedge_\star \dd A +\frac1{3!}\,A\wedge_\star[A,A]^\star_\frg \Big) \ ,
\end{align}
for a gauge field $A\in\Omega^1(M,\frg)$.

In order to be able to use this in a standard variational principle, we need a technical restriction on the Drinfel'd twist $\CF$: we require that the usual integral over $M$ is \emph{graded cyclic} with respect to the deformed exterior
product $\wedge_\star$, that is,
\begin{align}\label{eq:intcyclic}
\int_M\,\alpha\wedge_\star\beta= (-1)^{|\alpha|\,|\beta|} \,
  \int_M\,\beta\wedge_\star\alpha \ ,
\end{align}
for all homogeneous forms $\alpha,\beta\in\Omega^\bullet(M)[[\hbar]]$ in complementary degrees. For example, this holds for abelian twists (see Example~\ref{ex:abeliantwist}), whereby $\RR^{-1}=\CF^{2}$ and \eqref{eq:intcyclic} follows from integration by parts with respective to the Lie derivative.

Using this graded cyclicity repeatedly, along with invariance of the quadratic form $\Tr_\frg$ on $\frg$ and the $\RR$-matrix identities \eqref{eq:Rmatrixidsw} as well as the braided Leibniz rule \eqref{eq:braidedLeibniz}, one checks explicitly that the action functional \eqref{eq:CSbraidedaction} is invariant under the (left) braided gauge transformations 
\begin{align}
\delta_\lambda^\br A = \dd\lambda-[\lambda,A]_\frg^\star \ ,
\end{align}
for an arbitrary gauge parameter $\lambda\in\Omega^0(M,\frg)$, that is, $\delta_\lambda^\br S_\star(A)=0$. This holds for arbitrary quadratic Lie algebras $\frg$ without the introduction of extra degrees of freedom for the gauge fields $A$. On the other hand, by the ordinary Leibniz rule its variation is given by
\begin{align}\label{eq:deltaSstarCS}
\delta S_\star(A) = \int_M \, \Tr_\frg\Big(\delta A\wedge_\star\dd A + \frac12\,\delta A\wedge_\star[A,A]_\frg^\star\Big)
\end{align}
for arbitrary variations $\delta A\in\Omega^1(M,\frg)$ of the gauge fields, as in the classical case. This leads to  the equations of motion
\begin{align}
\frF_{\tA}^\star = F_{\tA}^\br = 0 
\end{align}
describing a flat braided $\frg$-connection, which are braided covariant by virtue of \eqref{eq:braidedcovcurv}:
\begin{align}
\delta_\lambda^\br\,\frF_{\tA}^\star = -[\lambda,\frF_{\tA}^\star]_\frg^\star \ .
\end{align}

So far everything we have said is just the obvious braided deformation of what was said in the classical case from Section~\ref{sec:whatisgauge}, as in the conventional noncommutative field theories of Section~\ref{sec:NCLinfty}. At this point, however, the braided field theory deviates from the naive expectations. Despite the covariance of the field equations, braided gauge symmetries \emph{do not} produce new solutions: An infinitesimal gauge variation produces the field equation
\begin{align}
\frF^\star_{A+\delta_\lambda^\br A} = \dd(A+\delta_\lambda^\br A) + \tfrac12\,[A,A]_\frg^\star + \tfrac12\,\big([\delta_\lambda^\br A,A]_\frg^\star + [A,\delta_\lambda^\br A]_\frg^\star\big) + O(\lambda^2) \ ,
\end{align}
and due to the braided Leibniz rule \eqref{eq:braidedLeibniz} this cannot be combined to the form $\frF^\star_{\tA}+\delta_\lambda^\br\,\frF_{\tA}^\star+O(\lambda^2)$. That is,
\begin{align}
\delta_\lambda^\br\,\frF_{\tA}^\star \ \neq \ \frF^\star_{A+\delta_\lambda^\br A} - \frF_{\tA}^\star + O(\lambda^2) \ ,
\end{align}
and hence a braided gauge transformation of a solution $A$ to $\frF_{\tA}^\star=0$ may produce a field configuration with \smash{$\frF^\star_{A+\delta_\lambda^\br A}\neq0$}. In other words, there is generally no moduli space of flat $\frg$-connections modulo braided gauge transformations in this sense, and the space of physical states cannot be described as in the commutative case.\footnote{The trivial (and uninteresting) exception is of course when $\frg$ is an abelian Lie algebra and the resulting Chern--Simons theory is a non-interacting field theory, whose classical and braided versions are the same.} We will say more about this unusual yet consistent feature of the ``braided moduli space'' from a more technical standpoint in Section~\ref{sec:braidedMC}.

How then should we interpret, or even justify, the meaning of the term `braided gauge symmetry'? The key point is our alternative interpretation of gauge symmetries via Noether's second theorem, which for the standard Chern--Simons theories from Sections~\ref{sec:whatisgauge} and~\ref{sec:usualNCLinfty} coincides with the Bianchi identity for the curvature 2-form. The latter is known to be generally violated in braided noncommutative differential geometry (see e.g.~\cite{Barnes:2016cjm,NAGravity,Aschieri:2020ifa}): The braided version of the Bianchi identity is readily calculated by taking the exterior derivative of the expression \eqref{eq:braidedcurvature} and rewriting it in terms of $F_{\tA}^\br$ to get the differential identity
\begin{align}
\dd F_{\tA}^\br -\tfrac12\,[F_{\tA}^\br,A]_\frg^\star+\tfrac12\,[A,F_{\tA}^\br]_\frg^\star + \tfrac14\,\big([[A,A]_\frg^\star,A]_\frg^\star - [A,[A,A]_\frg^\star]_\frg^\star\big) = 0 \ .
\end{align}
The last term vanishes in the classical case with trivial $\RR$-matrix $\RR=1\otimes1$ by symmetry of the Lie bracket $[-,-]_\frg$ on $1$-forms $\Omega^1(M,\frg)$, but it is non-vanishing in general.
We can rewrite the modified Bianchi identity in terms of the left and right covariant derivatives \eqref{eq:leftrightcovderiv} as
\begin{align}\label{eq:braidedNoetherCS}
\tfrac12\,\big(\dd_{\star\lact}^{\tA}\frF_{\tA}^\star+\dd_{\star\ract}^{\tA}\frF_{\tA}^\star\big)=-\tfrac14\,[\sfR_k(A),[\sfR^k(A),A]_\frg^\star]_\frg^\star \ ,
\end{align}
where we used the braided (odd) Jacobi identity \eqref{eq:braidedJacobiforms} on 1-forms with $\alpha_1=\alpha_2=\alpha_3=A$.
Again this identity holds off-shell, but it differs from its classical form by inhomogeneous terms involving the gauge field $A$ itself. 

In analogy to the classical case, it is natural to interpret this as the `braided Noether identity' that should follow via braided gauge invariance of the action functional. This is indeed the case, but one can no longer use the formula \eqref{eq:deltaSstarCS} as in the classical case, because braided gauge transformation obey a braided Leibniz rule and so are no longer special directions of a general field variation, which obeys the standard Leibniz rule; that is,
\begin{align}
\delta_\lambda^\br S_\star(A) \ \neq \ \int_M \, \Tr_\frg\big(\delta_\lambda^\br A\wedge_\star\frF_{\tA}^\star\big) \ .
\end{align}
Nonetheless, the braided Noether identity \eqref{eq:braidedNoetherCS} follows from $\delta_\lambda^\br S_\star(A)=0$ by using cyclicity and integration by parts, along with ${\rm ad}(\frg)$-invariance of the inner product $\Tr_\frg$ and the braided Jacobi identity, to isolate the gauge parameter $\lambda$ while making the Euler--Lagrange derivative $\frF_{\tA}^\star$ appear; that is,
\begin{align}
\delta_\lambda^\br S_\star(A) = -\int_M \, \Tr_\frg\bigg(\lambda\wedge_\star\Big(\frac12\,\big(\dd_{\star\lact}^{\tA}\frF_{\tA}^\star+\dd_{\star\ract}^{\tA}\frF_{\tA}^\star\big)+\frac14\,[\sfR_k(A),[\sfR^k(A),A]_\frg^\star]_\frg^\star\Big)\bigg) \ .
\end{align}
This provides a braided version of Noether's second theorem: braided gauge symmetries still induce interdependence among the equations of motion, and hence gauge redundancies. This justifies the interpretation of local braided symmetries as ``gauge''. Note that there is only one braided Noether identity, corresponding to the fact that left and right braided gauge transformations are not independent (see~\eqref{eq:braidedvarA}). 

We now notice that, analogously to Section~\ref{sec:whatisgauge}, the formulation of braided Chern--Simons gauge theory can be written entirely in terms of the `bracket' operations
\begin{align}
\ell_1^\star=\dd \qquad \mbox{and} \qquad \ell_2^\star=-[-,-]_\frg^\star
\end{align}
which define a differential graded braided Lie algebra structure on the graded vector space $V:=\Omega^\bullet(M,\frg)[[\hbar]]$, as well as the (strictly) cyclic inner product
\begin{align}
\langle\alpha,\beta\rangle_\star := \langle-,-\rangle\circ\CF^{-1}(\alpha\otimes\beta) = \int_M \, \Tr_\frg(\alpha\wedge_\star\beta)
\end{align}
on $V$, where strict cyclicity follows using the classical properties together with the graded commutativity \eqref{eq:intcyclic} and the $\RR$-matrix identities \eqref{eq:Rmatrixidsw}. That is, the braided gauge transformations of the fields, the equations of motion and the gauge redundancy via the braided Noether identities are encoded respectively by
\begin{align}\label{eq:braidedCSLinfty}
\begin{split}
\delta_\lambda^\br A = \ell_1^\star(\lambda) + \ell_2^\star(\lambda,A) & \qquad , \qquad \frF_{\tA}^\star = \ell^\star_1(A) - \tfrac12\,\ell_2^\star(A,A) \ , \\[4pt] \dsf^\star_{\tA} \frF_{\tA}^\star := \ \ell_1^\star(\frF_{\tA}^\star)+\tfrac12\,\big(\ell^\star_2(\frF_{\tA}^\star,A) - & \, \ell_2^\star(A,\frF_{\tA}^\star)\big) + \tfrac14\,\ell_2^\star\big(\sfR_k(A),\ell_2^\star(\sfR^k(A),A)\big) \ = \ 0 \ ,
\end{split}
\end{align}
with gauge closure and braided covariance represented by
\begin{align}
\big[\delta_{\lambda_1}^\br,\delta_{\lambda_2}^\br\big]^\star_\circ A = \delta^\br_{-\ell_2^\star(\lambda_1,\lambda_2)}A \qquad \mbox{and} \qquad \delta^\br_\lambda\, \frF_{\tA}^\star = \ell_2^\star(\lambda,\frF_{\tA}^\star) \ .
\end{align}
The braided Chern--Simons functional is encoded as
\begin{align}\label{eq:CSactionellbraided}
S_\star(A) = \tfrac12\,\langle A,\ell_1^\star(A)\rangle_\star - \tfrac1{3!}\,\langle A,\ell_2^\star(A,A)\rangle_\star \ ,
\end{align}
from which the equations of motion $\frF_{\tA}^\star=0$ are derived by varying and using cyclicity of the inner product.

The structure maps in this formulation of the braided gauge theory are all evidently simply twist deformations, via the procedure of Section~\ref{sec:Drinfeldtwist}, of those from Section~\ref{sec:whatisgauge}.
This motivates an extension of the notion of an $L_\infty$-algebra to the braided setting, and the construction of more general noncommutative field theories with braided gauge symmetry by twisting the formalism from Section~\ref{sec:LinftyCFT}. This extension will occupy the remainder of this section.

\subsection{Braided $L_\infty$-algebras}
\label{sec:Linftytwist}

In this section we shall work with $L_\infty$-algebras as presented in Section~\ref{sec:Linfty} in terms of brackets and relations, and define a braided generalization in these terms, deferring a discussion of their incarnations as braided versions of coalgebras and Chevalley--Eilenberg algebras until Section~\ref{sec:braidedBV}. The key observation is that the definition of a classical $L_\infty$-algebra takes place in the category of vector spaces and linear maps, which is a symmetric monoidal category with braiding given by the trivial transposition isomorphism $\tau$. This definition can be made in \emph{any} symmetric monoidal category with a non-trivial braiding, just like the definition of a Lie algebra. In this paper we will only need the definition in the category ${}_{\sU_\tCF\frv}\CCM$ of left $\sU_\tCF\frv$-modules, and we restrict to this setting for definiteness.

\begin{definition}\label{def:braidedLinfty}
A \emph{braided $L_\infty$-algebra} is a $\RZ$-graded $\sU_\tCF\frv$-module\footnote{We regard $\sU_\tCF\frv$ itself as a $\RZ$-graded $\sU_\tCF\frv$-module sitting in degree~$0$.} $V=\bigoplus_{k\in\RZ}\,V^k$ with a collection $\ell_n:V^{\otimes n}\to V$, $n\geq1$ of equivariant graded braided antisymmetric multilinear maps of degree $2-n$ that satisfy the \emph{braided homotopy Jacobi identities}
\begin{align}\label{eq:homotopyJacobibraided}
\sum_{i=1}^n \, (-1)^{i\,(n-i)} \, \ell_{n-i+1}\circ\big(\ell_i\otimes\unit^{\otimes n-i}\big) \ \circ \ \sum_{\sigma\in{\rm Sh}(i;n-i)} \, {\rm sgn}(\sigma) \, \tau_{\tCR}^\sigma = 0
\end{align}
for each $n\geq1$, where $\tau_{\tCR}^\sigma:V^{\otimes n}\to V^{\otimes n}$ denotes the action of the permutation $\sigma$ via the symmetric braiding $\tau_{\tCR}$ of the category ${}_{\sU_\tCF\frv}\CCM$ times the Koszul sign multiplication which interchanges factors in a tensor product of graded $\sU_\tCF\frv$-modules. 
\end{definition}

Let us unravel this definition in a bit more detail. Graded braided antisymmetry is the property
\begin{align}\label{eq:antisymmbraided}
\ell_n(v_1,\dots,v_n) = -(-1)^{|v_i|\,|v_{i+1}|} \ \ell_n\big(v_1,\dots,v_{i-1},\sfR_k(v_{i+1}),\sfR^k(v_i),v_{i+2},\dots,v_n\big)
\end{align}
for $i=1,\dots,n-1$ and homogeneous elements $v_1,\dots,v_n\in V$. The permutation action $\tau_{\tCR}^\sigma$ includes, in addition to the usual Koszul signs, appropriate insertions of the $\RR$-matrix as in \eqref{eq:antisymmbraided}. The braided homotopy Jacobi identities \eqref{eq:homotopyJacobibraided} are unchanged from their classical counterparts for $n=1,2$, that is, \eqref{eq:homotopyn=1} and \eqref{eq:homotopyn=2} continue to hold. Thus, similarly to the classical case, every braided $L_\infty$-algebra has an underlying cochain complex $(V,\ell_1)$ in the category ${}_{\sU_\tCF\frv}\CCM$ and $\ell_2:V\otimes V\to V$ is a cochain map in ${}_{\sU_\tCF\frv}\CCM$. The first homotopy relation which differs from the classical case is for $n=3$, which is a modification of \eqref{eq:homotopyn=3} in accordance with the non-trivial braiding \eqref{eq:antisymmbraided}:
\begin{align}
\begin{split}
& \ell_2\big(\ell_2(v_1,v_2),v_3\big) - (-1)^{|v_2|\,|v_3|}\,
  \ell_2\big(\ell_2(v_1,\sfR_k(v_3)),\sfR^k(v_2)\big) \\
& \hspace{4cm} + (-1)^{(|v_2|+|v_3|)\,|v_1|}\,
  \ell_2\big(\ell_2(\sfR_k(v_2),\sfR_l(v_3)),\sfR^l\,\sfR^k(v_1)\big) \\[4pt]
& \hspace{1cm} = -\ell_3\big(\ell_1(v_1),v_2,v_3\big) - (-1)^{|v_1|}\,
  \ell_3\big(v_1, \ell_1(v_2), v_3\big) - (-1)^{|v_1|+|v_2|}\,
  \ell_3\big(v_1,v_2, \ell_1(v_3)\big) \\
& \hspace{4cm} -\ell_1\big(\ell_3(v_1,v_2,v_3)\big) \ .
\end{split}
\end{align}

When the only non-zero bracket is $\ell_2$, a braided $L_\infty$-algebra is a braided Lie algebra in the sense of~\cite{Majid:1993yp} and it coincides with the notion of Lie algebra in ${}_{\sU_\tCF\frv}\CCM$ discussed by~\cite{Barnes:2015uxa}. Thus braided $L_\infty$-algebras are homotopy coherent generalizations of braided Lie algebras. In particular, the cohomology $H^\bullet(V,\ell_1)$ of a braided $L_\infty$-algebra is a $\RZ$-graded braided Lie algebra.

In the following we regard the one-dimensional vector space $\FC[[\hbar]]$ as a trivial graded left $\sU_\tCF\frv$-module, sitting in degree~$0$.

\begin{definition}
A \emph{cyclic braided $L_\infty$-algebra} is a braided $L_\infty$-algebra $(V,\{\ell_n\})$ together with a non-degenerate graded braided symmetric cochain map $\langle-,-\rangle:V\otimes V\to \FC[[\hbar]]$ that satisfies the \emph{braided cyclicity condition}
\begin{align}
\langle-,-\rangle \, \circ \, (\unit\otimes\ell_n) = {\rm sgn}(\sigma) \, \langle-,-\rangle \, \circ \, (\unit\otimes\ell_n) \, \circ \, \tau_{\tCR}^\sigma \ ,
\end{align}
for all $n\geq1$ and for all cyclic permutations $\sigma\in C_{n+1}$.
\end{definition}

Evaluated on elements $v_0,v_1,\dots,v_n\in V$, cyclicity in the category ${}_{\sU_\tCF\frv}\CCM$ reads as
\begin{align}
\big\langle v_0,\ell_n(v_1,v_2,\dots,v_n)\big\rangle = \pm \,  \big\langle \sfR_{k_0}\,\sfR_{k_1}\cdots\sfR_{k_{n-1}}( v_n),\ell_n\big(\sfR^{k_0}(v_0),\sfR^{k_1}(v_1),\dots,\sfR^{k_{n-1}}(v_{n-1})\big)\big\rangle \ .
\end{align}

In this paper we are predominantly interested in braided $L_\infty$-algebras that result from twist deformations of classical $L_\infty$-algebras, via the techniques described in Section~\ref{sec:Drinfeldtwist}. This can be done starting from an $L_\infty$-algebra $(V,\{\ell_n\})$ in the category ${}_{\sU\frv}\CCM$ of left $\sU\frv$-modules. This means that $V=\bigoplus_{k\in\RZ}\,V_k$ is a $\RZ$-graded
$\sU\frv$-module and the $n$-brackets $\ell_n:V^{\otimes n}\to V$ are
equivariant maps. Given any 
Drinfel'd twist $\CF\in \sU\frv[[\hbar]]\otimes \sU\frv[[\hbar]]$, we can
deform the brackets $\ell_n$ to new brackets $\ell_n^\star:V[[\hbar]]^{\otimes n}\to V[[\hbar]]$ which
are morphisms in the twisted representation category
${}_{\sU_\tCF\frv}\CCM$ of $\sU_\tCF\frv$-modules. Following the
standard prescription (\ref{eq:mustar}), we set
$\ell_1^\star:=\ell_1$ and
\begin{align}\label{eq:ellnstardef}
	\ell_n^\star(v_1\otimes\cdots\otimes v_n) :=
	\ell_n(v_1\otimes_\star\cdots\otimes_\star v_n)
\end{align}
for $n\geq2$, where $v\otimes_\star v':=\CF^{-1}(v\otimes
v')=\bar\sff^k(v)\otimes\bar\sff_k(v')$ for $v,v'\in V$. The functorial equivalence between the symmetric monoidal categories ${}_{\sU\frv}\CCM$ and ${}_{\sU_\tCF\frv}\CCM$ then implies the following central result, whose proof is found in~\cite[Proposition~4.8]{Ciric:2021rhi}.

\begin{proposition}\label{prop:braidedfromclassical}
	If $(V,\{\ell_n\})$ is an $L_\infty$-algebra in the category
	${}_{\sU\frv}\CCM$, then $(V[[\hbar]],\{\ell_n^\star\})$ is a braided
	$L_\infty$-algebra. 
\end{proposition}

This result is a far reaching generalization of what we did in Section~\ref{sec:braidedCS}, where the differential graded Lie algebra $\big(\Omega^\bullet(M,\frg),\ell_1,\ell_2\big)$ underlying Chern--Simons gauge theory on a three-manifold $M$ was twist deformed to the differential graded braided Lie algebra $\big(\Omega^\bullet(M,\frg)[[\hbar]],\ell_1^\star,\ell_2^\star\big)$ of braided Chern--Simons theory on $M$. We will discuss how this procedure extends to more general field theories organised by generic $L_\infty$-algebras in Section~\ref{sec:braidedfieldtheory}.

Similarly, we can twist deform cyclic structures on $L_\infty$-algebras to obtain cyclic braided $L_\infty$-algebras. Let $(V,\{\ell_n\},\langle-,-\rangle)$ be a cyclic
$L_\infty$-algebra in the category ${}_{\sU\frv}\CCM$ of left
$\sU\frv$-modules, with $\FR$ regarded as a trivial graded $\sU\frv$-module sitting in degree~$0$. This means that $(V,\{\ell_n\})$ is an $L_\infty$-algebra in ${}_{\sU\frv}\CCM$ and the cyclic inner product 
$\langle-,-\rangle:V\otimes V\to\FR$ is $\sU\frv$-invariant:
\begin{align}\label{eq:paringinv}
X\langle v_1\otimes v_2\rangle = \langle\Delta(X)(v_1\otimes v_2)\rangle = 0 \ ,
\end{align}
for all $X\in \sU\frv$ and $v_1,v_2\in V$. Following the standard prescription \eqref{eq:mustar}, we can twist
deform the inner product $\langle-,-\rangle$ to a new pairing
$\langle-,-\rangle_\star$ defined by
\begin{align}\label{eq:twistpairing}
\langle v_1,v_2\rangle_\star := \langle-,-\rangle\circ \CF^{-1}(v_1\otimes v_2) = \langle
  \bar\sff^k(v_1),\bar\sff_k(v_2) \rangle \ .
\end{align}

\begin{proposition}
If $(V,\{\ell_n\},\langle-,-\rangle)$ is a cyclic $L_\infty$-algebra in the category
	${}_{\sU\frv}\CCM$, then \\ $(V[[\hbar]],\{\ell_n^\star\},\langle-,-\rangle_\star)$ is a cyclic braided
	$L_\infty$-algebra. 
\end{proposition}
\begin{proof}
Non-degeneracy of the pairing $\langle-,-\rangle_\star:V[[\hbar]]\otimes V[[\hbar]]\to \FC[[\hbar]]$ follows from the  untwisted non-degeneracy and an order by order argument in the formal deformation parameter $\hbar$. Graded symmetry of the cyclic inner product $\langle-,-\rangle$
implies that the twisted inner product $\langle-,-\rangle_\star$ is naturally
graded braided symmetric. 

To check braided cyclicity, we compute
\begin{align}\label{eq:braidedcyccomp}
\begin{split}
\big\langle v_0,\ell^\star_n(v_1,\dots,v_n)\big\rangle_\star &= \big\langle
                                                              \bar\sff^k(v_0)\otimes\bar\sff_k\big(\ell_n(\bar\sff^l(v_1\otimes_\star \cdots\otimes_\star
                                                              v_{n-1})\otimes\bar\sff_l(v_n))\big)\big\rangle
  \\[4pt]
  &=
    \big\langle(\unit\otimes\ell_n)\big(\bar\sff^k(v_0)\otimes\Delta(\bar\sff_k)\big[\bar\sff^l(v_1\otimes_\star\cdots\otimes_\star
    v_{n-1})\otimes\bar\sff_l(v_n)\big]\big)\big\rangle
  \\[4pt]
  &=
    \big\langle(\unit\otimes\ell_n)\big(\Delta(\bar\sff^k)\big[\bar\sff^l(v_0)\otimes \bar\sff_l(v_1\otimes_\star\cdots\otimes_\star
    v_{n-1})\big]\otimes\bar\sff_k(v_n)\big)\big\rangle \\[4pt]
  &= \pm\,
    \big\langle\bar\sff_k(v_n)\otimes
    \ell_n\big(\Delta\big(\bar\sff^k)\big[\bar\sff^l(v_0)\otimes\bar\sff_l(v_1\otimes_\star\cdots\otimes_\star
    v_{n-1})\big]\big)\big\rangle \\[4pt]
  &=\pm\,
    \big\langle\bar\sff_k(v_n)\otimes \bar\sff^k\big(\ell_n(\bar\sff^l(v_0)\otimes\bar\sff_l(v_1\otimes_\star\cdots\otimes_\star
    v_{n-1}))\big)\big\rangle \\[4pt]
  &=\pm\,
    \big\langle\bar\sff^k\,\sfR_m(v_n)\otimes\bar\sff_k\,\sfR^m\big(\ell_n(\bar\sff^l(v_0)\otimes\bar\sff_l(v_1\otimes_\star\cdots\otimes_\star
    v_{n-1}))\big)\big\rangle \\[4pt]
  &=\pm\, \big\langle\sfR_m(v_n)\otimes\sfR^m\big(\ell_n^\star(v_0\otimes
    v_1\otimes\cdots\otimes v_{n-1})\big)\big\rangle_\star \\[4pt]
  &= \pm\, \big\langle(\unit\otimes\ell_n^\star)\big(\sfR_m(v_n)\otimes \Delta_\tCF(\sfR^m)\big[v_0\otimes(v_1\otimes\cdots\otimes v_{n-1})\big]\big)\big\rangle_\star \\[4pt]
  &= \pm\,\big\langle \sfR_{k_0}\,\sfR_{k_1}\cdots\sfR_{k_{n-1}}(v_n),\ell_n^\star\big(\sfR^{k_0}(v_0),\sfR^{k_1}(v_1),\dots,\sfR^{k_{n-1}}(v_{n-1})\big)\big\rangle_\star \ .
\end{split}
\end{align}
In the first and seventh lines we use the definitions of the
star-pairing and the star-brackets, in the second and fifth lines we use commutativity
of $\ell_n$ with the $\sU\frv$-action, in the third line we use the cocycle property \eqref{eq:invcocycle} for the inverse of the Drinfel'd twist along with $\sU\frv$-invariance \eqref{eq:paringinv} of
the classical pairing, in the fourth line we use
classical cyclicity, in the sixth line we use the definition of the
$\RR$-matrix, in the eighth line we use commutativity of $\ell_n^\star$ with the $\sU_\tCF\frv$-action, and in the ninth line we use the $\RR$-matrix identities \eqref{eq:Rmatrixidsw}.
\end{proof}

However, this level of generality is not 
suitable for applications to field theory, as the loss of {\it strict} graded
symmetry and cyclicity would lead to problems with the variational principle for the corresponding action functionals. Following the concrete example of braided Chern--Simons theory from Section~\ref{sec:braidedCS}, we need to restrict the types of Drinfel'd twists by abstracting the graded cyclicity property \eqref{eq:intcyclic} for integration
of twisted differential forms.

\begin{definition}\label{def:compatibletwist}
A Drinfel'd twist $\CF\in \sU\frv[[\hbar]]\otimes \sU\frv[[\hbar]]$ is
\emph{compatible} with a cyclic structure $\langle-,-\rangle:V\otimes
V\to\FR$ on an $L_\infty$-algebra in ${}_{\sU\frv}\CCM$ if
\begin{align}\label{eq:compatibletwist}
\langle
  v_1,v_2\rangle_\star = \langle \sfR_k(v_1),\sfR^k(v_2)\rangle_\star = (-1)^{|v_1|\,|v_2|} \, \langle v_2,v_1\rangle_\star
\end{align}
for all homogeneous $v_1,v_2\in V$.
\end{definition}

In other words, a compatible Drinfel'd twist turns braided graded symmetry of the inner product into \emph{strict} graded symmetry. It also turns the braided cyclicity condition into the strict cyclicity condition (see \eqref{eq:cyclicity}), as inferred by

\begin{corollary}\label{cor:compatiblestrict}
Let $(V,\{\ell_n\},\langle-,-\rangle)$ be a cyclic $L_\infty$-algebra
in the category ${}_{\sU\frv}\CCM$, and let $\CF$ be a compatible Drinfel'd
twist. Then
$(V[[\hbar]],\{\ell_n^\star\},\langle-,-\rangle_\star)$ is a strictly
cyclic braided $L_\infty$-algebra.
\end{corollary}

\begin{proof}
The seventh line of the computation \eqref{eq:braidedcyccomp} shows
\begin{align}\label{eq:strictcyclic}
\begin{split}
\big\langle v_0,\ell^\star_n(v_1,v_2,\dots,v_n)\big\rangle_\star &= \pm\, \big\langle\sfR_k(v_n),\sfR^k\big(\ell_n^\star(v_0,
    v_1,\dots, v_{n-1})\big)\big\rangle_\star \\[4pt]
&= \pm\, \big\langle v_n,\ell_n^\star(v_0,
    v_1,\dots, v_{n-1})\big\rangle_\star \ ,
\end{split}
\end{align}
where in the second line we used \eqref{eq:compatibletwist}.
\end{proof}

\subsection{Braided $L_\infty$-algebras of noncommutative field theories}
\label{sec:braidedfieldtheory}

We shall now define noncommutative field theories in the braided $L_\infty$-algebra formalism by
employing a twist deformation of the approach discussed in
Section~\ref{sec:LinftyCFT}, whereby a classical field theory is defined by its
4-term cyclic $L_\infty$-algebra as the initial input.\footnote{Here we treat only theories with irreducible gauge symmetries. We shall discuss an example with higher gauge symmetries in Section~\ref{sec:braidedBFtheory}.} We call a field theory obtained in this
way a \emph{braided field theory}. That is, a braided field theory is
defined by a 4-term cyclic braided $L_\infty$-algebra as the initial input. 

Starting from a classical field theory on a manifold $M$, described algebraically by a 4-term cyclic $L_\infty$-algebra $(V,\{\ell_n\},\langle-,-\rangle)$ (over $\FR$) with underlying graded vector space
\begin{align}
V=V^0\oplus V^1\oplus V^2\oplus V^3
\end{align}
in the representation category ${}_{\sU\frv}\CCM$, following the
prescription of Section~\ref{sec:Linftytwist} we construct the
corresponding 4-term cyclic braided $L_\infty$-algebra
$(V[[\hbar]],\{\ell_n^\star\},\langle-,-\rangle_\star)$ (over $\FC[[\hbar]]$) in the twisted representation category
${}_{\sU_\tCF\frv}\CCM$. Since $\ell_1^\star=\ell_1$, this braided $L_\infty$-algebra has the same underlying cochain complex \eqref{eq:cochaincomplex}. In other words the noncommutative field theories obtained in this way have the same underlying free field theory as their classical versions, as in the case of the star-gauge theories of Section~\ref{sec:usualNCLinfty}. Only the interactions are different and they are encoded in the higher braided brackets $\ell_n^\star$ with $n\geq2$.

The (left) braided gauge transformation of a field $A\in V^1$ by a gauge parameter $\lambda\in V^0$ is defined by
\begin{align}\label{eq:LgtL} 
\delta_\lambda^\br A := \ell_1^\star(\lambda) + \sum_{n\geq1} \, \frac{(-1)^{n\choose 2}}{n!}\, \ell_{n+1}^\star(\lambda,A^{\otimes n}) 
 \ ,
\end{align}
as in the classical case \eqref{gaugetransfA}, where now the positioning of $\lambda$ inside the slots of the brackets $\ell_{n+1}^\star$ matters because of braided symmetry \eqref{eq:antisymmbraided}. The braided gauge variations are extended to maps $\delta^\br_\lambda:V\to V$ of degree~$0$ in the usual way as braided derivations. For the field theories we work with in this paper, whereby the algebra of gauge variations closes off-shell with field-independent gauge parameters, they close a braided Lie algebra under the braided commutator:
\begin{align}\label{eq:braidedclosure}
\big[\delta_{\lambda_1}^{\br},\delta_{\lambda_2}^{\br}\big]_\circ^\star = \delta_{-\ell^\star_2(\lambda_1,\lambda_2)}^{\br} \ .
\end{align}

The braided closure property \eqref{eq:braidedclosure} is shown in~\cite{Ciric:2021rhi} using the braided homotopy Jacobi identity \eqref{eq:homotopyJacobibraided} evaluated on $(\lambda_1,\lambda_2,A^{\otimes n})$, order by order in the expansion in $A^{\otimes n}$. The calculations are considerably more involved than in the classical case~\cite{BVChristian}, due to the braided Leibniz rule \eqref{eq:braidedLeibniz} obeyed by the braided gauge variations, and because the braided symmetry properties \eqref{eq:antisymmbraided} prevent one from combining various sums and simplifying the homotopy relations. The identities \eqref{eq:Rmatrixidsw} are instrumental in manipulating the proliferation of $\RR$-matrix factors which ensue. We take the validity of this result, and the others which follow below, as a highly non-trivial vindication that braided gauge symmetries and covariant dynamics are controlled by the tightly knit algebraic structure of a braided $L_\infty$-algebra.

The equations of motion are dictated by a braided version of the Maurer--Cartan equations $\frF_{\tA}^\star=0$ in $V^2[[\hbar]]$ for a field $A\in V^1$, where
\begin{align}\label{eq:braidedeom}
\frF^\star_{\tA} := \ell_1^\star(A) + \sum_{n \geq2} \, \frac{(-1)^{n\choose 2}}{n!}\, \ell^\star_{n}(A^{\otimes n}) \ ,
\end{align} 
exactly as in the classical case \eqref{EOM}. The dynamics are braided covariant, that is, the expression \eqref{eq:braidedeom} is covariant under the braided gauge variations~\eqref{eq:LgtL}:
\begin{align}\label{eq:lefteomcov}
\delta_\lambda^\br\, \frF^\star_{\tA} = \ell_2^\star(\lambda,\frF_{\tA}^\star) + \sum_{n \geq1} \, \frac{(-1)^{n\choose 2}}{(n+1)!} \ \sum_{i=0}^{n} \, (-1)^i \, \ell^\star_{n+2}\big(\lambda,A^{\otimes i},\frF^\star_{\tA},A^{\otimes n-i}\big) \ ,
\end{align}
for all gauge parameters $\lambda\in V^0$, as shown in~\cite{Ciric:2021rhi} order by order using the braided homotopy Jacobi identity \eqref{eq:homotopyJacobibraided} evaluated on $(\lambda,A^{\otimes n})$. In the classical case, where $\RR=1\otimes1$, the expression \eqref{eq:lefteomcov} agrees with \eqref{gaugetransfF}.

The absence of a (naive) moduli space of solutions to $\frF_{\tA}^\star=0$ that we observed in the case of braided Chern--Simons theory can be seen in this formalism to be a generic feature of noncommutative field theories with braided gauge symmetries: Although \eqref{eq:lefteomcov} implies that $\delta_\lambda^\br\frF_{\tA}^\star=0$ on solutions $A\in V^1$ of the field equations $\frF_{\tA}^\star=0$, computing the first order variation in any infinitesimal gauge parameter $\lambda\in V^0$ gives
\begin{align}
\frF^\star_{A+\delta_\lambda^\br A} = \ell_1^\star\big(A+\delta_\lambda^\br A\big) + \sum_{n\geq2} \, \frac{(-1)^{n\choose 2}}{n!} \, \Big( \ell_n^\star(A^{\otimes n}) + \sum_{i=0}^{n-1} \, \ell_n^\star\big(A^{\otimes i},\delta_\lambda^\br A, A^{\otimes n-1-i}\big) \Big) + O(\lambda^2) \ .
\end{align}
This cannot be brought to the form $\frF^\star_{\tA}+\delta_\lambda^\br\,\frF_{\tA}^\star+O(\lambda^2)$ because of the braided Leibniz rule \eqref{eq:braidedLeibniz}, and it will generally differ from the result of \eqref{eq:lefteomcov} by $\RR$-matrix factors in the sum over~$i$. Hence \smash{$\frF^\star_{A+\delta_\lambda^\br A} \neq \frF^\star_{\tA} + \delta_\lambda^\br\, \frF^\star_{\tA} + O(\lambda^2)$} in general, and it does not make sense to take the (naive) quotient of the subspace in $V^1$ of classical solutions by braided gauge transformations. We will return to this point in Section~\ref{sec:braidedMC} where we explain how a braided version of the Maurer--Cartan moduli space ${}^{\sU_\tCF\frv}\CCM\CCC(V[[\hbar]],\{\ell^\star_n\})$ can still be defined from a more abstract perspective.

From a physical perspective, generalized braided gauge symmetries are well-defined through a braided version of the Noether identities which hold off-shell and encode gauge redundancies, analogously to the classical case. These identities follow from the braided $L_\infty$-algebra formalism, which also naturally controls kinematical violations through the braided homotopy Jacobi identities, such as the well-known failure of the Bianchi identity in braided geometry that we discussed in Section~\ref{sec:braidedCS}. Analogously to the classical case, the braided Noether identities follow from a weighted sum over $n\geq1$ of the braided homotopy Jacobi identity \eqref{eq:homotopyJacobibraided} evaluated on $A^{\otimes n}$. However, because of the braided symmetry properties \eqref{eq:antisymmbraided}, we do not have available the simplification \eqref{eq:homotopyJacobidegree1} which took place in the classical case and the weighted sum is very cumbersome to work with. 

For the field theories we treat in this paper, where the gauge transformations are at most linear in the field $A$, the braided Noether identities are encoded by~\cite{Ciric:2021rhi} 
\begin{align}\label{braidedNoether}
\begin{split}
\dsf^\star_{\tA}\frF^\star_{\tA}  &:= \ell_1^\star(\frF^\star_{\tA}) + \frac12\,\big(\ell_2^\star(\frF^\star_{\tA},A) - \ell^\star_2(A,\frF^\star_{\tA})\big) - \sum_{n\geq3} \, \frac{(-1)^{n\choose 2}}{n!} \, \ell_1^\star\big(\ell^\star_{n}(A^{\otimes n})\big) \\
& \quad \, -\sum_{n\geq2} \, \frac{(-1)^{n\choose 2}}{2\,n!} \, \Big(\ell^\star_2\big(\ell^\star_n(A^{\otimes n}),A\big) - \ell_2^\star\big(A,\ell_n^\star(A^{\otimes n})\big)\Big) \ = \ 0 \ .
\end{split}
\end{align}
This uses the simplification provided by the  braided homotopy Jacobi identity \eqref{eq:homotopyJacobibraided} evaluated on $A^{\otimes n+1}$ which in this case reads
\begin{align}\label{eq:homotopybraidedAn+1}
\begin{split}
\sum_{i=1}^{n-1} \, \ell_2^\star\big(\ell_n^\star(A^{\otimes n-i},\sfR_k(A^{\otimes i})),\sfR^k(A)\big) &= \ell_{2}^\star\big(A,\ell_{n}^\star(A^{\otimes n})\big)-\ell_2^\star\big(\ell_{n}^\star(A^{\otimes n}),A\big) \\ & \quad \, - (-1)^{n} \, \ell_{1}^\star \big(\ell_{n+1}^\star(A^{\otimes n+1})\big) \ ,
\end{split}
\end{align}
for all $n\geq2$.
Unlike the classical case \eqref{eq:Noether}, the braided Noether identity \eqref{braidedNoether} is no longer linear in the field equations $\frF_{\tA}^\star$, while the inhomogeneous terms involving solely the field $A$ cancel in the classical limit $\RR=1\otimes1$ after using \eqref{eq:homotopybraidedAn+1}. This generalizes the braided Chern--Simons Noether identity in the second line of \eqref{eq:braidedCSLinfty}, which is recovered after applying \eqref{eq:homotopybraidedAn+1}.

While the braided Noether identities evidently exhibit gauge redundancies through inhomogeneous differential identities among the field equations, we can understand their precise correspondence to braided gauge symmetries by appealing to a variational principle for the braided field theory. Using the cyclic inner product $\langle-,-\rangle_\star$ on our braided $L_\infty$-algebra we define a braided version of the Maurer--Cartan action functional by
\begin{align}\label{eq:braidedaction}
S_\star(A) = \frac12\,\langle A,\ell_1^\star(A)\rangle_\star + \sum_{n\geq2} \, \frac{(-1)^{n\choose 2}}{(n+1)!}\, \langle A, \ell^\star_{n}(A^{\otimes n})\rangle_\star \ ,
\end{align}
for all fields $A\in V^1$, exactly as in the classical case \eqref{action}. 

If $\CF$ is a compatible Drinfel'd twist, in the sense of Definition~\ref{def:compatibletwist}, then using the ordinary Leibniz rule for general field variations $\delta A$ on the tangent space $TV^1$ to the space of fields we can easily compute the variation of \eqref{eq:braidedaction} exactly as in the classical case to get
\begin{align}
\begin{split}
\delta S_\star(A) &= \sum_{n\geq1} \, \frac{(-1)^{n\choose 2}}{(n+1)!} \, \Big( \langle\delta A,\ell_n^\star(A^{\otimes n})\rangle_\star + \sum_{i=0}^{n-1} \, \langle A,\ell_n^\star(A^{\otimes i},\delta A,A^{\otimes n-i-1})\rangle_\star\Big) \\[4pt]
&= \sum_{n\geq1} \, \frac{(-1)^{n\choose 2}}{n!} \, \langle\delta A,\ell^\star_n(A^{\otimes n})\rangle_\star = \langle \delta A,\frF^\star_{\tA}\rangle_\star \ ,
\end{split}
\end{align}
where in the second equality we used strict cyclicity for compatible twists (see Corollary~\ref{cor:compatiblestrict}) to set all terms in the sum over $i$ equal to the first term. Hence by non-degeneracy of the pairing, the extrema of the action functional \eqref{eq:braidedaction} are precisely the solutions $A\in V^1$ of the equations of motion $\frF^\star_{\tA}=0$.

Braided gauge invariance of the action functional \eqref{eq:braidedaction} is a bit more subtle, because the braided derivation property of $\delta_\lambda^\br$ is not generally compatible with the strict symmetry \eqref{eq:compatibletwist} of the inner product. Nevertheless, the gauge variation $\delta_\lambda^\br S_\star(A)$ of the action functional \eqref{eq:braidedaction} \emph{is} well-defined, because the strict cyclicity condition \eqref{eq:strictcyclic} ensures compatibility with strict symmetry of the pairing order by order in tensor powers of the fields $A^{\otimes n}$ for $n\geq 1$~\cite{Ciric:2021rhi}, that is,
\begin{align}\label{eq:deltaLsym}
\delta_\lambda^\br\big\langle A,\ell^\star_n\big(A^{\otimes n}\big)\big\rangle_\star = \delta_\lambda^\br\big\langle\ell^\star_n\big(A^{\otimes n}\big),A\big\rangle_\star \ .
\end{align}
However, in contrast to the classical case, a braided gauge variation $\delta_\lambda^\br A$ of a field $A\in V^1$ along a gauge parameter $\lambda\in V^0$ is no longer a special direction of a general variation $\delta A$ of the field in the tangent space $TV^1$, because the gauge variations are now defined as \emph{braided} derivations, in contrast to the general variations. In other words, in the braided case the first equality in \eqref{eq:gtNoether} simply no longer holds in general:
\begin{align}\label{eq:deltalambdapair}
\delta_\lambda^\br S_\star(A) \ \neq \ \langle \delta_\lambda^\br A,\frF^\star_{\tA}\rangle_\star \ ,
\end{align}
which can be seen explicitly by direct calculation.

Braided gauge invariance of the action functional, that is, $\delta_\lambda^\br S_\star(A)=0$ for all $\lambda\in V^0$, follows from an order by order calculation using the braided homotopy Jacobi identity \eqref{eq:homotopyJacobibraided} applied to $A^{\otimes n+1}$~\cite{Ciric:2021rhi}; this holds irrespectively of the strict or braided cyclicity of the inner product, that is, it is true for generic twists $\CF$. Despite \eqref{eq:deltalambdapair}, the gauge variation of the action functional \eqref{eq:braidedaction} can moreover be written in the form
\begin{align}
\delta_\lambda^\br S_\star(A) = -\langle\lambda,\dsf_{\tA}^\star \frF^\star_{\tA}\rangle_\star \ ,
\end{align}
exactly as in the classical case \eqref{eq:gtNoether}, by using cyclicity together with the $\RR$-matrix identities \eqref{eq:Rmatrixidsw} to isolate the gauge parameter $\lambda$. In this sense the operation $\sfd^\star_{\tA}$ is the `braided adjoint' to $-\delta^\br_\lambda$ with respect to the cyclic inner product, and the braided Noether identities $\dsf^\star_{\tA} \frF^\star_{\tA}=0$ then follow from braided gauge invariance $\delta_\lambda^\br S_\star(A)=0$ of the action functional \eqref{eq:braidedaction}. We stress, however, that $\sfd^\star_A$ is not what would normally be called the adjoint operator with respect to a non-degenerate pairing, as for example in~\eqref{eq:level1NoetherOp}, as this notion would take as starting point the right-hand side $\langle \delta_\lambda^\br A,\frF^\star_A\rangle_\star $ of the non-equality~\eqref{eq:deltalambdapair}. Since $\langle \delta_\lambda^\br A,\frF^\star_A\rangle_\star\neq 0$, the usual adjoint operator does not compose to zero with the braided Maurer--Cartan operator, as is reflected directly in the braided Noether identities \eqref{braidedNoether}.

The procedure described here yields systematic constructions of \emph{new} examples of noncommutative field theories which are braided deformations of classical field theories. These braided theories contain no new degrees of freedom compared with their classical versions, while still possessing good classical limits at $\hbar=0$. Unlike our motivating example of Section~\ref{sec:braidedCS}, in some cases the braided field theories exhibit some surprising features that are not simply the naive deformations of those found in the classical cases, in startling contrast to the conventional noncommutative field theories discussed in Section~\ref{sec:NCLinfty}. In addition to the braided Chern--Simons theory discussed in Section~\ref{sec:braidedCS}, the examples of braided Einstein--Cartan--Palatini theories of gravity in three and four dimensions are worked out in detail in~\cite{Ciric:2021rhi}. 

In the subsequent sections we present three new examples: the standard noncommutative scalar field theory (regarded as a braided field theory), a braided version of $BF$ theory in arbitrary dimensionality, and a braided version of noncommutative Yang--Mills theory for arbitrary gauge algebras. In the remainder of the present section we address some further technical points surrounding the constructions above.

\subsection{Reality conditions}

We have not yet addressed the reality conditions that are needed to make physical sense of the field equations as well as of the action functional of a braided field theory. For this, we need to restrict to special twists.
\begin{definition}\label{def:Hermitian}
A Drinfel'd twist $\CF\in \sU\frv[[\hbar]]\otimes \sU\frv[[\hbar]]$ is \emph{Hermitian} if it satisfies the reality condition $\CF^*=\CF_{21}$, where ${}^*$ denotes complex conjugation. 
\end{definition}
For the inverse twist $\CF^{-1}=\bar\sff^k\otimes\bar\sff_k$ the condition of Definition~\ref{def:Hermitian} reads 
	\begin{align}\label{eq:Twistherm}
		\bar\sff^k{}^*\otimes\bar\sff_k{}^* =
		\bar\sff_k\otimes \bar\sff^k \ .
\end{align}
In particular, this turns the conjugation isomorphism into an anti-involution of the braided exterior algebra, that is,
\begin{align}
(\alpha\wedge_\star\beta)^* = (-1)^{|\alpha|\,|\beta|} \, \beta^*\wedge_\star\alpha^*
\end{align}
for homogeneous forms $\alpha,\beta\in\Omega^\bullet(M)[[\hbar]]$, and in this case we say that $\wedge_\star$ is a \emph{Hermitian star-product}. A large class of examples is provided by abelian twists (see Example~\ref{ex:abeliantwist}), which are Hermitian when we use the convention $\hbar^\ast=\hbar$ for the formal deformation parameter. 

\begin{example}\label{ex:realcovder}
In the setting of Section~\ref{sec:braidedgaugesym}, let $A\in\Omega^1(M,\frg)[[\hbar]]$ be a braided gauge field and $\phi\in\Omega^p(M,W)[[\hbar]]$ a braided matter field in a braided representation $W$ of a gauge algebra $\frg$, which are both real: $A^*=A$ and $\phi^*=\phi$. If the twist $\CF$ is Hermitian, then the complex conjugate of the left braided covariant derivative of $\phi$ from \eqref{eq:leftrightcovderiv} is given by
\begin{align}
\begin{split}
\big(\dd_{\star\lact}^{\tA}\phi\big)^* &= \big(\dd\phi + A\wedge_\star\phi\big)^* \\[4pt]
&= \dd\phi^* + \big(\bar\sff^k(A)\wedge\bar\sff_k(\phi)\big)^\ast \\[4pt]
&= \dd\phi^* + \bar\sff_k(A^*)\wedge\bar\sff^k(\phi^*) \\[4pt]
&= \dd\phi + \sfR_k(A)\wedge_\star\sfR^k(\phi) = \dd_{\star\ract}^{\tA}\phi \ ,
\end{split}
\end{align}
where in the last line we used $\CF_{21}^{-1} = \CF^{-1}\,\RR^{-1}$.
Thus neither the left nor right braided covariant derivative preserves the reality  of the matter field $\phi$ in general. However, their symmetrized combination
\begin{align}
\tfrac12\,\big(\dd_{\star\lact}^{\tA} + \dd_{\star\ract}^{\tA}\big)
\end{align}
does preserve reality, and in examples this is the combination that will naturally appear, as guaranteed by the general formalism below (cf.\ Section~\ref{sec:braidedCS}).
\end{example}

Let $(V,\{\ell_n\},\langle -,- \rangle)$ be a cyclic 4-term
	$L_\infty$-algebra over $\FR$ encoding a classical Lagrangian field theory on a manifold $M$. This theory describes real dynamical fields, that is, they belong to a real vector space $V^1$ which is  a module over $C^{\infty}(M,\FR)$. Similarly, the field equations are real and the action functional is real-valued. Twist deformation via a compatible Hermitian Drinfel'd twist $\CF$ yields a 4-term strictly cyclic braided $L_\infty$-algebra
	$(V[[\hbar]],\{\ell_n^\star\}, \langle-,-\rangle_\star)$ over $\FC[[\hbar]]$. We will now show that the reality  of the classical field theory survives this deformation.

Firstly, the braided field equations $\frF_{\tA}^\star\in V^2[[\hbar]]$ are real when restricted to real fields $A\in V^1[[\hbar]]$. Conjugating the brackets $\ell_{n}^\star(A^{\otimes n}) \in V^{2}[[\hbar]]$ yields brackets in the complex conjugate vector space $V^2[[\hbar]]^*$ given by
\begin{align}
		\begin{split}
			\ell_{n}^\star(A^{\otimes n})^{*}&=\ell_n\big(\bar\sff^{k_1}(A) \otimes \bar\sff_{k_1}(\bar \sff^{k_2}( A)\otimes \bar \sff_{k_2}(\cdots \bar \sff^{k_{n-1}}( A) \otimes \bar\sff_{k_{n-1}}(A))\cdots))\big)^* 
			\\[4pt]
			&= \ell_{n}\big(\bar\sff_{k_1}(A^*) \otimes \bar\sff^{k_1}(\bar \sff_{k_2}( A^*)\otimes \bar \sff^{k_2}(\cdots \bar \sff_{k_{n-1}}( A^* )\otimes \bar\sff^{k_{n-1}}( A^*))\cdots))\big)\\[4pt]
			&=\ell_{n}\big(\sfR_{k_1}(A^*) \otimes_\star \sfR^{k_1}(\sfR_{k_2}( A^*) \otimes_\star \sfR^{k_2}(\cdots \sfR_{k_{n-1}} (A^*)\otimes_\star \sfR^{k_{n-1}}( A^*))\cdots)) \big)\\[4pt]
			&= \ell_n^\star\big(\sfR_{k_1}( A^*) \otimes \sfR^{k_1}(\sfR_{k_2}( A^*) \otimes \sfR^{k_2}(\cdots \sfR_{k_{n-1}}( A^*)\otimes \sfR^{k_{n-1}}( A^*))\cdots)) \big) \\[4pt]
			&= \ell_{n}^\star(A^*{}^{\otimes n}) \ ,
		\end{split}
	\end{align}
where $A^*\in V^1[[\hbar]]^*$. In the second line we repeatedly used Hermiticity of the twist and compatibility of the classical brackets with complex conjugation, ${}^*\circ \ell_{n} =\ell_{n}\circ {}^*$ (as they are $\FR$-multilinear maps extended linearly over the complexification of the underlying vector space), and in the third line we used the definition of the $\CR$-matrix $\CR=\CF_{21}\,\CF^{-1}$. In the last two lines we used the definition of the braided brackets along with their braided graded antisymmetry. This shows 
\begin{align}
(\frF_{\tA}^\star)^* = \frF_{A^*}^\star \ ,
\end{align}
and in particular the field equations are real on the subspace of fields where the conjugation isomorphism $A\mapsto A^*$ is the identity.

Secondly, the braided action functional $S_\star:V^1[[\hbar]]\to\FC[[\hbar]]$ is real-valued when restricted to real fields. This follows similarly to the computation above:
\begin{align}\label{eq:realbraidedaction}
		\begin{split}
\big[\langle A, \ell_{n}^\star (A^{\otimes n})\rangle_\star\big]^* 
			&= \big[\big\langle \bar\sff^k (A), \bar\sff_k\big( \ell_{n}^\star (A^{\otimes n})\big)\big\rangle \big]^*	\\[4pt]
			&= \big\langle \bar \sff_k  (A^*), \bar \sff^k \big(\ell_{n}^\star (A^{\otimes n})^*\big)\big\rangle \\[4pt]
			&= \big\langle \sfR_k(A^*), \sfR^k\big( \ell_{n}^\star (A^*{}^{\otimes n})\big)\big\rangle_\star = \langle A^*, \ell_{n}^\star (A^*{}^{\otimes n})\rangle_\star \ ,
		\end{split}
	\end{align}
where in the last equality we additionally used the compatibility of the twist. This shows that $S_\star(A)^* = S_\star(A^*)$, and the required reality follows.

\subsection{Braided BV formalism}
\label{sec:braidedBV}

Let us now briefly explain how the notion of cyclic braided $L_\infty$-algebra $(V[[\hbar]],\{\ell_n^\star\},\langle-,-\rangle_\star)$ in terms of star-brackets and braided homotopy relations is related to braided versions of the algebraic and coalgebraic formulations from Section~\ref{sec:whatisLinfty}. The Drinfel'd twist deformation prescription defines the {braided symmetric algebra} with deformed symmetric product $\odot_\star:=\odot\circ\CF^{-1}$, which we denote by $\Sym_\br(V[1])$. On generators $v,v'\in V[1]$ this product is braided graded commutative:
\begin{align}
v\odot_\star v' = (-1)^{|v|\,|v'|} \, \sfR_k(v')\odot_\star\sfR^k(v) \ .
\end{align} 

Following~\cite{Nguyen:2021rsa}, it is straightforward to develop the Chevalley--Eilenberg algebra of a braided $L_\infty$-algebra (as a braided noncommutative deformation of Definition~\ref{def:CEalgebra}), that is, a braided version of the BV formalism presented in Section~\ref{sec:BV}. For this, we extend the braided $L_\infty$-algebra structure on $V[[\hbar]]$ to $\Sym_\br(V[2])\otimes V[[\hbar]]$ via the brackets
\begin{align}
\begin{split}
\ell_n^{\star\,{\rm ext}}(\zeta_1\otimes v_1,\dots,\zeta_n\otimes v_n) :&\!= \ell_n^{\,\rm ext}\big((\zeta_1\otimes v_1)\otimes_\star\cdots\otimes_\star(\zeta_n\otimes v_n)\big) \\[4pt]
&=  \pm\,\big(\zeta_1\odot_\star\sfR_{k_1^{\textrm{\tiny1}}}(\zeta_2)\odot_\star\cdots\odot_\star\sfR_{k_{n-1}^{1}}\cdots \sfR_{k_1^{n-1}}(\zeta_n)\big) \\
& \hspace{2cm} \, \otimes \, \ell_n^{\star}\big(\sfR^{k_{n-1}^{1}}\cdots\sfR^{k_1^{1}}(v_1),\dots,\sfR^{k_1^{n-1}}(v_{n-1}),v_n\big) \ ,
\end{split}
\end{align}
for all $n\geq1$, $\zeta_1,\dots,\zeta_n\in\Sym_\br(V[2])$ and $v_1,\dots,v_n\in V$. Similarly, we extend the cyclic structure via the pairing
\begin{align}
\begin{split}
\langle\zeta_1\otimes v_1,\zeta_2\otimes v_2\rangle_\star^{\rm ext} := \langle(\zeta_1\otimes v_1)\otimes_\star(\zeta_2\otimes v_2)\rangle^{\rm ext} = \pm\,\big(\zeta_1\odot_\star\sfR_k(\zeta_2)\big) \, \langle \sfR^k(v_1),v_2\rangle_\star \ .
\end{split}
\end{align}

We can now define a braided version of the BV differential $Q_\BV^\star:\Sym_\br (V[1])^*\to \Sym_\br (V[1])^*$ by 
\begin{align}
Q_\BV^\star\xi = -\sum_{n\geq1} \, \frac{(-1)^{n\choose 2}}{n!} \, \ell_n^{\star\,\rm ext}(\xi^{\otimes n})  \ ,
\end{align}
where the contracted coordinate functions
\begin{align}
\xi := \tau^\alpha\otimes\tau_\alpha \ \in \ \Sym_\br(V[2])\otimes V[[\hbar]]
\end{align}
of degree~$1$ are defined by a choice of basis $\{\tau_\alpha\}\subset V$ with dual basis $\{\tau^\alpha\}\subset V^*\simeq V[3]$. Similarly, we define a braided version of the BV symplectic $2$-form
\begin{align}
\omega_\BV^\star := -\tfrac12\,\langle\delta\xi,\delta\xi\rangle_\star^{\rm ext} \ \in \ \Omega^2(V[1])[[\hbar]] \ ,
\end{align}
whose inverse is a braided graded Poisson bracket $\{-,-\}_\BV^\star$, that is, a braided graded Lie bracket which is a braided graded derivation on $\Sym_\br(V[2])$ in each of its slots, and which moreover is compatible with the differential $Q_\BV^\star$. The action of the braided BV differential can then be represented in terms of this braided Poisson bracket as
\begin{align}
Q_\BV^\star = \{S^\star_\BV,-\}^\star_\BV \ ,
\end{align}
where the braided BV action functional
\begin{align}
S^\star_\BV = \sum_{n\geq1} \, \frac{(-1)^{n\choose 2}}{(n+1)!} \, \langle\xi,\ell_n^{\star\,\rm ext}(\xi^{\otimes n})\rangle_\star^{\rm ext} \ \in \ \Sym_\br(V[2])
\end{align}
of degree~$0$ satisfies the {classical master equation}
\begin{align}
\{S^\star_\BV,S^\star_\BV\}^\star_\BV = 0 \ .
\end{align}

A crucial point here is that the classical master equation, or equivalently nilpotency $(Q_\BV^\star)^2=0$, follows from precisely the same calculation as in the ordinary case~\cite[Section~4.3]{BVChristian}. This is because the contracted coordinate functions $\xi$ are $\sU_\tCF\frv$-invariant elements of $\Sym_\br(V[2])\otimes V[[\hbar]]$,\footnote{Invariance of $\xi$ follows from either a direct calculation or by the $\sU_\tCF\frv$-isomorphism \smash{$V\otimes V^*\simeq{\rm Hom}_{{}_{\sU_\tCF\frv}\CCM}(V,V)$}, under which $\xi$ corresponds to the identity morphism on $V$.} and hence all appearances of $\RR$-matrices in the properties of the extended brackets $\ell_n^{\star\,\rm ext}$ and of the extended pairing $\langle-,-\rangle_\star^{\rm ext}$ disappear when evaluated on tensor powers of the $\sU_\tCF\frv$-invariant element $\xi$. Another important feature is that this differential graded algebra perspective circumvents the issues discussed in Section~\ref{sec:Linftytwist} with defining a moduli space ${}^{\sU_\tCF\frv}\CCM\CCC(V[[\hbar]],\{\ell^\star_n\})$ of classical solutions in the braided setting. This is because it characterises this space by its differential graded algebra ${}^{\sU_\tCF\frv}\CCA_\MC$ of equivariant functions, and so the (coarse) moduli space may be defined through the degree~$0$ cohomology of the differential graded braided commutative algebra $\big(\Sym_\br(V[1])^*,Q_\BV^\star\big)$. These two observations will be combined together in Section~\ref{sec:braidedMC} below to sketch a precise description of the moduli space ${}^{\sU_\tCF\frv}\CCM\CCC(V[[\hbar]],\{\ell^\star_n\})$ which resolves the technical issues mentioned in Section~\ref{sec:Linftytwist}.

In a similar vein, one can derive a duality with differential graded braided cocommutative coalgebras (as a braided noncommutative deformation of Definition~\ref{def:coalgebra}). The braided symmetric algebra $\Sym_\br(V[1])$ has a natural free noncocommutative coalgebra structure with coproduct $\Delta_{\textrm{\tiny$V$}}^\star:\Sym_\br(V[1])\to\Sym_\br(V[1])\otimes\Sym_\br(V[1])$ obtained by Drinfel'd twist deformation of the cocommutative coproduct \eqref{eq:DeltaV} in the usual way, replacing permutation actions $\tau^\sigma:V[1]^{\otimes n}\to V[1]^{\otimes n}$ with $\tau_\br^\sigma$ for $\sigma\in S_n$. The braided graded symmetric multilinear maps
\begin{align}
b_n^\star := s\circ\ell_n^\star\circ(s^{-1})^{\otimes n}
\end{align}
of degree~$1$ define a map $\sum_{n\geq1} \, b_n^\star:\Sym_\br(V[1])\to (V[1])[[\hbar]]$, which uniquely extends  to a coderivation
\begin{align}
D_\star:\Sym_\br(V[1])\longrightarrow\Sym_\br(V[1])
\end{align}
of degree~$1$. This is a differential, $(D_\star)^2=0$, as a consequence of the braided homotopy Jacobi identities, and it too provides a natural setting for the moduli space of classical solutions ${}^{\sU_\tCF\frv}\CCM\CCC(V[[\hbar]],\{\ell^\star_n\})$ as the degree~$0$ cohomology of the corresponding differential graded braided cocommutative coalgebra $(\Sym_\br(V[1]),D_\star)$. 

The coalgebra formulation also provides a natural interpretation of a braided $L_\infty$-algebra as a deformation of an $L_\infty$-algebra in the category ${}_{\sU\frv}\CCM$~\cite{Jonketalk}: The classical differential graded coalgebra $(\Sym(V[1]),D)$ naturally inherits the structure of a cocommutative Hopf algebra compatible with the differential $D$,\footnote{More precisely, here one needs to work with the unreduced symmetric algebra and curved $L_\infty$-algebras.} which can be twisted to a new noncocommutative Hopf algebra by the techniques of Section~\ref{sec:Drinfeldtwist} with a compatible differential $D_\tCF$. The braided $L_\infty$-algebra $(V[[\hbar]],D_\star)$ may then be interpreted as an $L_\infty$-module for the resulting twisted $L_\infty$-algebra.

\subsection{Moduli spaces of classical solutions}
\label{sec:braidedMC}

We conclude this section by briefly sketching how to potentially make sense of the moduli space of classical solutions ${}^{\sU_\tCF\frv}\CCM\CCC(V[[\hbar]],\{\ell^\star_n\})$ in braided field theory, which we refer to as the `braided Maurer--Cartan moduli space'. We shall be intentionally brief and gloss over many technical details in order to illustrate the main ideas;\footnote{Nevertheless, this is the most technical part of the paper. The reader not interested in the formal descriptions of moduli spaces may safely skip this section, which is not essential to the rest of the paper.} a more thorough and rigorous construction will be presented elsewhere. We start by recalling the description of the classical Maurer--Cartan moduli space $\CCM\CCC(V,\{\ell_n\})$ of an $L_\infty$-algebra $(V,\{\ell_n\})$, whose construction is originally due to Getzler~\cite{getzler09}, see also~\cite{vallette20} for a more explicit alternative construction and~\cite{Manetti05} for the special case when $(V,\ell_1,\ell_2)$ is a differential graded Lie algebra (as in the example of Chern--Simons gauge theory).

We begin with some preliminary notation. Let $\dAlg$ denote the category of differential graded commutative algebras over $\FR$, whose morphisms are algebra homomorphisms of degree~$0$ which intertwine the differentials.\footnote{\label{fn:inftycategory} The appropriate homotopical analogue generalizing the spectrum of a commutative algebra to differential graded commutative algebras takes place in the $\infty$-category of `simplicial' sets, but this is not needed in our presentation below and we shall work solely with the category $\Set$ of sets. Similarly, to define appropriate notions of representability of functors we should work with the $\infty$-category of `Artin' differential graded commutative algebras in non-positive degrees, but we do not indicate this detail and stick to the full source category $\dAlg$ below.}
Let $\CL=(V,\{\ell_n\})$ be a classical $L_\infty$-algebra (in the category of real vector spaces). Given a differential graded commutative algebra $\sA=(\CA,\dsf_{\textrm{\tiny$\CA$}})$ (with the product on $\CA$  denoted by juxtaposition), there is a tensor product $L_\infty$-algebra $\sA\otimes\CL=\big(\CA\otimes V,\big\{\ell_n^\tsA\big\}\big)$ with homogeneous subspaces $(\CA\otimes V)^k=\bigoplus_{l\in\RZ}\, \CA^{-l}\otimes V^{l+k}$ for $k\in\RZ$ and brackets given by
\begin{align}\label{eq:tensorLinfty}
\begin{split}
\ell_1^\tsA(a_1\otimes v_1) = (\dsf_{\textrm{\tiny$\CA$}}+\ell_1) (a_1\otimes v_1) :&\!= \dsf_{\textrm{\tiny$\CA$}}a_1\otimes v_1 + (-1)^{|a_1|} \, a_1\otimes \ell_1(v_1) \ , \\[4pt] 
\ell_n^\tsA(a_1\otimes v_1,\dots,a_n\otimes v_n) &= \pm\,(a_1\cdots a_n)\otimes\ell_n(v_1,\dots,v_n) \ ,
\end{split}
\end{align}
for all $n\geq 2$, $a_1,\dots,a_n\in \CA$ and $v_1,\dots,v_n\in V$. This generalizes the extended $L_\infty$-algebra construction of Section~\ref{sec:BV}.

Following~\cite{vallette20}, the (coarse) Maurer--Cartan moduli space $\CCM\CCC(\CL)$ of the $L_\infty$-algebra $\CL=(V,\{\ell_n\})$ is the set of solutions $A\in V^1$ of the Maurer--Cartan equation $\frF_A=0$, modulo the gauge equivalence relation $A\sim A'$ if there exists a gauge parameter $\lambda\in V^0$ and a path of fields $A(t)\in V^1$ for $t\in[0,1]$ with $A(0)=A$, $A(1)=A'$ which satisfies the differential equation
\begin{align}\label{eq:MCODE}
\frac{\dd}{\dd t} A(t) = \delta_\lambda A(t) 
\end{align}
for all $t\in[0,1]$. The infinitesimal version of \eqref{eq:MCODE} (its evaluation at $t=0$) gives the usual infinitesimal gauge transformation $\delta_\lambda A$ of $A$. With the given initial condition and an appropriate completeness condition on the $L_\infty$-algebra $\CL$, the differential equation \eqref{eq:MCODE} always admits a solution $A(t)$ which satisfies the Maurer--Cartan equation for all $t\in[0,1]$, that is, $\frF_{A(t)}=0$. In the case that $\CL$ is a differential graded Lie algebra, this definition is equivalent to the equivalence relation arising from the (integrated) group action of $V^0$ on $V^1$ (see e.g.~\cite{Manetti05}). 

The moduli space $\CCM\CCC(\CL)$ is a (formal derived) stack that is described by its `functor of points' $\CCF_\CL:\dAlg\to\Set$, as we now explain. This defines a presheaf on the category of differential graded affine schemes over $\FR$, and so from this perspective the moduli space is described not as a set of objects but rather as a space which parametrizes those objects.

\begin{definition}\label{def:MCfamilies}
Let $\CL=(V,\{\ell_n\})$ be an $L_\infty$-algebra and $\sA=(\CA,\dsf_{\textrm{\tiny$\CA$}})$ a differential graded commutative algebra.
\begin{myenumerate} 
\item[(a)] A \emph{family of classical solutions} in $\CL$ parametrized by $\sA$ is an element $\alpha=a^l\otimes v_l$, with $a^l\in\CA^{-l}$ and $v_l\in V^{l+1}$ for $l\in\RZ$, which solves the Maurer--Cartan equation 
\begin{align}
(\dsf_{\textrm{\tiny$\CA$}}+\ell_1)\alpha + \sum_{n\geq 2} \, \frac{(-1)^{n\choose 2}}{n!} \, \ell_n^\tsA(\alpha^{\otimes n}) = 0
\end{align}
in the $L_\infty$-algebra $\sA\otimes\CL$.
\item[(b)] Two families $\alpha$ and $\alpha'$ of classical solutions in $\CL$ are \emph{gauge equivalent} if there exists $\lambda=b^l\otimes w_l$, with $b^l\in \CA^{-k}$ and $ w_l\in V^k$ for $l\in\RZ$, and a path $\alpha(t)=a^l(t)\otimes v_l(t)$ in the $L_\infty$-algebra $\sA\otimes \CL$ which satisfies the differential equation 
\begin{align}
\frac\dd{\dd t}\alpha(t) = (\dsf_{\textrm{\tiny$\CA$}} + \ell_1)\lambda + \sum_{n \geq1} \, \frac{(-1)^{n\choose 2}}{n!} \, \ell_{n+1}^\tsA\big(\lambda,\alpha(t)^{\otimes n}\big) \ ,
\end{align}
such that $\alpha=\alpha(0)$ and $\alpha'=\alpha(1)$.
\end{myenumerate}
\end{definition}

For a family parametrized by a one-point space, that is, for $\sA = \FR_0 := (\FR,0)$ with $\FR$ sitting in degree~$0$, this definition reduces to the notion of gauge equivalence of classical solutions discussed above.

\begin{definition}\label{def:MCspace}
Let $\CL$ be an $L_\infty$-algebra.
\begin{myenumerate}
\item[(a)] The \emph{moduli functor} is the covariant functor
$
\CCF_\CL:\dAlg\to\Set
$
which associates to every differential graded commutative algebra $\sA$ the set $\CCF_\CL(\sA)$ of gauge equivalence classes of families of classical solutions in $\CL$ parametrized by $\sA$, and to every morphism $f:\sA\to\sB$ of differential graded commutative algebras it associates the pushforward $\CCF_\CL(f)$ sending families parametrized by $\sA$ to families parametrized by $\sB$.
\item[(b)] A (\emph{fine}) \emph{Maurer--Cartan moduli space} $\CCM\CCC(\CL)$ is a representation of the moduli functor, that is, a pair $(\hat\sA,\mathsf{\Phi})$ where $\hat\sA$ is an object of the category $\dAlg$ and ${\mathsf\Phi}:{\rm Hom}_\dAlg(\hat\sA,-)\to\CCF_\CL$ is a natural transformation of functors.
\end{myenumerate}
\end{definition}

The idea behind this definition is that the functor \smash{${\rm Hom}_\dAlg(\hat\sA,-):\dAlg\to\Set$} describes all possible ways that any other differential graded affine scheme may be mapped to the moduli space, which completely captures the (derived) geometric structure of the moduli scheme given by the spectrum of the differential graded commutative algebra $\hat\sA$. The distinguished element \smash{$\hat\alpha={\sf\Phi}(\unit_{\hat A})\in\CCF_\CL(\hat\sA)$} is called the \emph{universal classical solution}: For every pair $(\sA,\alpha)$ with $\sA$ an object of $\dAlg$ and $\alpha\in\CCF_\CL(\sA)$, there is a unique morphism ${\sf\Lambda}\in{\rm Hom}_\dAlg(\hat\sA,\sA)$ such that $(\CCF_\CL\circ{\sf\Lambda})(\hat\alpha)=\alpha$. When it exists, a universal object \smash{$(\hat\sA,\hat\alpha)$} is unique up to a unique isomorphism, and in this sense it is the space of \emph{all} classical solutions modulo gauge transformations. For a more accurate and explicit description of the Maurer--Cartan moduli space using simplicial methods, see~\cite{vallette20}. 

This description of the moduli space of classical solutions can be adapted to the braided noncommutative setting in the following way, which uses a description of functors of points similar to those studied in~\cite{Barnes:2016bmg} in the context of toric noncommutative geometry. For this, let now $\CL=(V,\{\ell_n\})$ be an $L_\infty$-algebra in the category ${}_{\sU\frv}\CCM$ of left $\sU\frv$-modules. Similarly, we work with differential graded commutative algebras in the category ${}_{\sU\frv}\CCM$. These form the category ${}^{\sU\frv}\dAlg$ of differential graded commutative left $\sU\frv$-module algebras, whose morphisms are equivariant algebra homomorphisms of degree~$0$ intertwining the differentials. The description of the Maurer--Cartan moduli space $\CCM\CCC(\CL)$ proceeds exactly as above, by restricting the moduli functor $\CCF_\CL:\dAlg\to\Set$ to the subcategory~${}^{\sU\frv}\dAlg$.

Given a Drinfel'd twist $\CF\in \sU\frv[[\hbar]]\otimes \sU\frv[[\hbar]]$, by Proposition~\ref{prop:braidedfromclassical} we obtain a braided $L_\infty$-algebra $\CL^\star=(V[[\hbar]],\{\ell_n^\star\})$. Similarly, by the usual constructions of Section~\ref{sec:Drinfeldtwist} we obtain a (functorially equivalent) category ${}^{\sU_\tCF\frv}\dAlg$ of differential graded braided commutative algebras. An object $\sA^\star=(\CA,\dsf_{\textrm{\tiny$\CA$}})$ of this category consists of a graded braided commutative algebra $\CA$ (with the product on $\CA$ denoted by $\star$), that is,
\begin{align}
a\star a' = (-1)^{|a|\,|a'|} \, \sfR_k(a')\star \sfR^k(a) \ ,
\end{align}
for all homogeneous $a,a'\in\CA$, together with a differential $\dsf_{\textrm{\tiny$\CA$}}:\CA\to\CA$ of degree~$1$ which is a graded derivation, that is, $\dsf_{\textrm{\tiny$\CA$}}(a\star a') = \dsf_{\textrm{\tiny$\CA$}}a \star a' + (-1)^{|a|} \, a\star\dsf_{\textrm{\tiny$\CA$}}a'$. There is a braided tensor product $L_\infty$-algebra $\sA^\star\otimes\CL^\star=\big(\CA\otimes V[[\hbar]],\big\{\ell_n^{\star\tsA}\big\}\big)$ whose differential is unchanged from the classical case \eqref{eq:tensorLinfty}, while the higher brackets of \eqref{eq:tensorLinfty} are modified to 
\begin{align}
\begin{split}
\ell_n^{\star\tsA}(a_1\otimes v_1,\dots,a_n\otimes v_n) 
&=  \pm\,\big(a_1\star\sfR_{k_1^{\textrm{\tiny1}}}(a_2)\star\cdots\star\sfR_{k_{n-1}^{1}}\cdots \sfR_{k_1^{n-1}}(a_n)\big) \\
& \hspace{2cm} \, \otimes \, \ell_n^{\star}\big(\sfR^{k_{n-1}^{1}}\cdots\sfR^{k_1^{1}}(v_1),\dots,\sfR^{k_1^{n-1}}(v_{n-1}),v_n\big) \ ,
\end{split}
\end{align}
for all $n\geq 2$, $a_1,\dots,a_n\in \CA$ and $v_1,\dots,v_n\in V$. This generalizes the extended braided $L_\infty$-algebra construction of Section~\ref{sec:braidedBV}.

\begin{definition}
Let $\CL^\star=(V[[\hbar]],\{\ell^\star_n\})$ be a braided $L_\infty$-algebra and $\sA^\star=(\CA,\dsf_{\textrm{\tiny$\CA$}})$ a differential graded braided commutative algebra.
\begin{myenumerate} 
\item[(a)] A \emph{family of classical solutions} in $\CL^\star$ parametrized by $\sA^\star$ is a $\sU_\tCF\frv$-invariant element $\alpha=a^l\otimes v_l$, with $a^l\in\CA^{-l}$ and $v_l\in V^{l+1}[[\hbar]]$ for $l\in\RZ$, which solves the braided Maurer--Cartan equation 
\begin{align}
(\dsf_{\textrm{\tiny$\CA$}}+\ell^\star_1)\alpha + \sum_{n\geq 2} \, \frac{(-1)^{n\choose 2}}{n!} \, \ell_n^{\star\tsA}(\alpha^{\otimes n}) = 0
\end{align}
in the braided $L_\infty$-algebra $\sA^\star\otimes\CL^\star$.
\item[(b)] Two families $\alpha$ and $\alpha'$ of classical solutions in $\CL^\star$ are \emph{braided gauge equivalent} if there exists a $\sU_\tCF\frv$-invariant element $\lambda=b^l\otimes w_l$, with $b^l\in \CA^{-k}$ and $ w_l\in V^k[[\hbar]]$ for $l\in\RZ$, and a $\sU_\tCF\frv$-invariant path $\alpha(t)=a^l(t)\otimes v_l(t)$ in the braided $L_\infty$-algebra $\sA^\star\otimes \CL^\star$ which satisfies the differential equation 
\begin{align}
\frac\dd{\dd t}\alpha(t) = (\dsf_{\textrm{\tiny$\CA$}} + \ell_1^\star)\lambda + \sum_{n \geq1} \, \frac{(-1)^{n\choose 2}}{n!} \, \ell_{n+1}^{\star\tsA}\big(\lambda,\alpha(t)^{\otimes n}\big) \ ,
\end{align}
such that $\alpha=\alpha(0)$ and $\alpha'=\alpha(1)$.
\end{myenumerate}
\end{definition}

This is our key definition: part (b) of the definition now makes sense non-trivially, because all $\RR$-matrix insertions disappear  when evaluated on the $\sU_\tCF\frv$-invariant elements of the set of tensor product solutions $\alpha$ to the braided Maurer--Cartan equations $\frF_\alpha^\star=0$ in $\sA^\star\otimes\CL^\star$, and so braided gauge transformations act as in the classical setting of Definition~\ref{def:MCfamilies}. Note that this does not preclude noncommutative deformations of the solutions compared to Definition~\ref{def:MCfamilies}, as both $\sA^\star$ and $\CL^\star$ are objects in the symmetric monoidal category ${}_{\sU_\tCF\frv}\CCM$, and so each individually involve braided operations: only in their braided tensor product $\sA^\star\otimes\CL^\star$ do we look for $\sU_\tCF\frv$-invariant elements. This also clarifies what goes wrong in the setting of Section~\ref{sec:braidedfieldtheory}: to recover that situation we would set $\sA^\star=\FR_0$ in this definition, which means that the classical solutions should be $\sU_\tCF\frv$-invariant fields. Then all issues concerning the action of braided gauge transformations, the braided Noether identities and related matters disappear. Of course we are mostly interested in non-invariant fields, but to detect those one has to consider differential graded braided commutative algebras $\sA^\star$ which carry a non-trivial $\sU_\tCF\frv$-action. In this sense, the braided moduli space does not contain any (interesting) `points', as we further exemplify below.

The description of the braided Maurer--Cartan moduli space is now a simple extension of the classical case, as the latter definitions make sense for covariant functors from any (locally small) category to the category of sets. The definition of moduli functor from part~(a) of Definition~\ref{def:MCspace} generalizes immediately to give a moduli functor ${}^{\sU_\tCF\frv}\CCF_{\CL^\star}:{}^{\sU_\tCF\frv}\dAlg\to\Set$ which sends each differential graded braided commutative algebra $\sA^\star$ to the set ${}^{\sU_\tCF\frv}\CCF_{\CL^\star}(\sA^\star)$ of braided gauge equivalence classes of families of classical solutions in the braided $L_\infty$-algebra $\CL^\star$ parametrized by $\sA^\star$. Part~(b) then generalizes easily in the following manner.

\begin{definition}
Let $\CL^\star$ be a braided $L_\infty$-algebra. A \emph{braided Maurer--Cartan moduli space} ${}^{\sU_\tCF\frv}\CCM\CCC(\CL^\star)$ is a pair $(\hat\sA^\star,\hat\alpha^\star)$, where $\hat\sA^\star$ is an object of the category ${}^{\sU_\tCF\frv}\dAlg$ and $\hat\alpha^\star\in {}^{\sU_\tCF\frv}\CCF_{\CL^\star}(\hat A^\star)$, which is a universal object representing the moduli functor ${}^{\sU_\tCF\frv}\CCF_{\CL^\star}:{}^{\sU_\tCF\frv}\dAlg\to\Set$.
\end{definition}

This definition has the same meaning as in the classical case,\footnote{By Yoneda's lemma, this definition is bijectively equivalent to the existence of a natural transformation \smash{${\sf\Phi}^\star:{\rm Hom}_{{}^{\sU_\tCF\frv}\dAlg}(\hat\sA^\star,-)\to{}^{\sU_\tCF\frv}\CCF_{\CL^\star}$} with $\hat\alpha^\star={\sf\Phi}^\star(\unit_{\hat\sA^\star})$.}  but with a crucial difference: since the moduli space is necessarily an object in the source category of the functor ${}^{\sU_\tCF\frv}\CCF_{\CL^\star}$, and we allow noncommutative parameter spaces (regarded as dual to differential graded braided commutative algebras), the braided Maurer--Cartan moduli space is naturally a (differential graded) `noncommutative scheme', which is the tradeoff for making sense of the braided gauge equivalence classes. Of course, the highly non-trivial aspect of this construction, which we do not address here, is whether there actually exists a representation of the functor ${}^{\sU_\tCF\frv}\CCF_{\CL^\star}$ which brings our arguments above to fruition, and hence provides a realization of the moduli space of classical solutions in braided field theory.

Noncommutative moduli spaces have previously appeared in braided noncommutative geometry as solution spaces to noncommutative instanton equations in~\cite{Brain:2009it}, where it is shown that the noncommutativity of the parameter spaces is a gauge artefact: one may always recover an equivalent description in terms of the usual notion of classical moduli space by a suitable choice of gauge transformation and a noncommutative quotient construction. It would be interesting to understand the braided Maurer--Cartan moduli space (if it really exists in this sense) in a similar vein, which would elucidate the physical meaning of the abstract arguments presented here.

This braided Maurer--Cartan moduli space may also be extracted from the braided BV formalism of Section~\ref{sec:braidedBV} and its braided noncommutative algebra of functions ${}^{\sU_\tCF\frv}\CCA_\MC$. For this, one considers the functor of points ${}^{\sU_\tCF\frv}\CCF_{\CL^\star}^{\textrm{\tiny$\vee$}}:{}^{\sU_\tCF\frv}\dAlg \to \Set$ defined on objects by ${}^{\sU_\tCF\frv}\CCF_{\CL^\star}^{\textrm{\tiny$\vee$}}(\sA^\star) = {\rm Hom}_{{}^{\sU_\tCF\frv}\dAlg}\big({}^{\sU_\tCF\frv}\CCA_\MC,\sA^\star\big)$. Taking again $\sA^\star=\FR_0$, one only detects $\sU_\tCF\frv$-invariant points, whereas the interesting fields appear for non-trivial differential graded braided commutative algebras $\sA^\star$.

\section{Noncommutative scalar field theory}
\label{sec:NCscalar}

The $L_\infty$-algebra framework also accommodates field theories without gauge symmetries. This simply translates to trivial homogeneous subspaces $V^{0}=V^{3}=\{0\}$ and corresponding zero brackets. The appearance of the $L_\infty$-algebra structure is then tautological, as it is just given by polynomial expansion of the field equations. In this section we will review the formalism in the example of an interacting scalar field theory, and subsequently twist deform this $L_\infty$-algebra to obtain the standard noncommutative scalar field theory. This example also serves to illustrate how to handle more general theories which are not diffeomorphism invariant, via twist deformation quantization along the global spacetime symmetries of the classical field theory. It moreover shows how to incorporate global symmetries into the braided $L_\infty$-algebra framework, which we describe following the classical treatment of~\cite{BVChristian}.\footnote{Of course, we are not interested in a moduli space of globally (braided) symmetric solutions (as these are not gauge symmetries). In principle, one is interested instead in the weakly conserved currents corresponding to these infinitesimal symmetries, which would follow from a braided version of Noether's first theorem. We do not develop this interesting point in the present paper.}

\subsection{Scalar field theory in the $L_\infty$-algebra formalism}
\label{sec:scalartheory}

Let $(M,g)$ be a $d$-dimensional oriented Lorentzian manifold.\footnote{The discussion is presented for Lorentzian signature in order to compare later on with previous points of view on global symmetries. The more general pseudo-Riemannian cases differ only by a sign in the Laplace--Beltrami operator.} Let $\dd \vol_{g}\in \Omega^{d}(M)$ be the induced volume form and $\ast_\hodge : \Omega^{k}(M) \to \Omega^{d-k}(M)$ the corresponding Hodge duality operator. Denote by ${\square}:= \ast_\hodge \dd \ast_\hodge  \dd : C^{\infty}(M,\FR)\to C^{\infty}(M,\FR)$ the Laplace--Beltrami operator acting on smooth functions.

The action functional for an interacting real scalar field $\phi \in C^{\infty}(M,\FR)$ is defined by
\begin{align}
S(\phi)= \int_{M}\, \Big(\frac{1}{2}\,\phi \,\big(\square-m^{2}\big)\phi - \sum_{n\geq3}\,\frac{\kappa_{n-1}}{n!}\,\phi^n \Big) \ \dd \vol_{g} \ ,
\end{align}
for a mass parameter $m^2\geq0$ and coupling constants $\kappa_{n-1}\in \FR$. The corresponding field equations are given by $\frF_\phi=0$ where
\begin{align}
	\frF_\phi = \big(\square-m^{2}\big)\phi -\sum_{n\geq2}\,\frac{\kappa_{n}}{n!}\,\phi^n \ ,
\end{align}
from which the cyclic $L_{\infty}$-algebra is immediately read off. 

The underlying graded vector space is concentrated in degrees $1$ and $2$, comprised of two copies of the space of functions on $M$:
\begin{align}
V= V^1\oplus V^2 \qquad \mbox{with} \quad V^1=V^2=C^\infty(M,\FR) \ ,
\end{align}
and we denote its homogeneous elements by $\phi \in V^{1}$ and $\varPhi\in V^{2}$. The underlying cochain complex $(V,\ell_1)$ has as differential the Klein--Gordon operator, and the non-trivial brackets are
\begin{align}\label{eq:SCbrackets}
	\ell_{1}(\phi)= \big(\square-m^{2}\big)\phi \qquad \mbox{and} \qquad
	\ell_{n}(\phi_{1},\dots,\phi_{n})= -(-1)^{n\choose 2} \, \kappa_n \, \phi_{1} \cdots \phi_{n} \ ,
\end{align}
where $n\geq2$, while the cyclic pairing is 
\begin{align} \label{eq:SCpairing}
	\langle\phi,\varPhi\rangle := \int_M\, \phi\ast_\hodge \varPhi = \int_M\,
	\phi \ \varPhi \ \dd \vol_{g} \ .
\end{align}
The homotopy Jacobi identities follow trivially for degree reasons.
Cyclicity with respect to $\ell_n$ for $n\geq2$ is tautological, due to commutativity of pointwise multiplication of functions. Cyclicity with respect to $\ell_1$ follows since $\square$ is a (formally) self-adjoint operator with respect to the inner product \eqref{eq:SCpairing} on $C^\infty(M,\FR)$. 

\subsection{Noncommutative scalar field theory in the braided $L_\infty$-algebra formalism}
\label{sec:braidedscalar}

Suppose now that the spacetime $(M,g)$ has a non-trivial set of isometries. That is, there exists a non-zero Lie subalgebra of Killing vector fields $\frk \subset \frv=\sfGamma(TM)$, such that $\LL_{\xi}\, g = 0$ for all $\xi\in \frk$. This implies that
\begin{align}
	\LL_\xi \circ \ast_\hodge  = \ast_\hodge  \circ \LL_\xi 
\end{align}
for all $\xi \in \frk$, and so the Laplace--Beltrami operator $\square= \ast_\hodge  \dd \ast_\hodge  \dd$ also commutes with the action of all Killing vector fields. It follows that all brackets \eqref{eq:SCbrackets} of the scalar field $L_{\infty}$-algebra commute with the action of $\frk$:
\begin{align}
	\LL_\xi \circ \ell_{n}= \ell_{n} \circ \LL_\xi
\end{align}
for all $\xi\in \frk$ and $n\geq1$. In other words, $(V, \{\ell_n\})$ is an $L_\infty$-algebra in the category of $\sU\frk$-modules. 

By the construction of Section~\ref{sec:Linftytwist}, a choice of a \emph{Killing} Drinfel'd twist $\mathcal{F}\in \sU\frk[[\hbar]]\otimes \sU\frk[[\hbar]]$ results in a 2-term braided $L_{\infty}$-algebra $\big(V[[\hbar]], \{\ell^\star_{n}\}\big)$ with non-zero brackets 
\begin{align}\label{eq:BraidedSCbrackets}
\begin{split}
	\ell_{1}^\star(\phi)= \big(\square- m^{2}\big)\phi  \qquad \mbox{and} \qquad
	\ell_{n}^\star(\phi_{1},\dots,\phi_{n}) = -(-1)^{n\choose 2} \, \kappa_n \, \phi_{1}\star\cdots \star\phi_{n} \ ,
	\end{split}
\end{align}
where $n\geq2$ and $\star$ denotes the star-product of functions. By the prescription of Section~\ref{sec:braidedfieldtheory}, the braided field equations landing in $V^2[[\hbar]]=C^{\infty}(M,\FR)[[\hbar]]$ are encoded by
\begin{align}\label{eq:braidedSCfieldeqs}
	\frF^{\star}_\phi= \sum_{n\geq1} \, \frac{(-1)^{n\choose 2}}{n!} \, \ell_n^\star(\phi^{\otimes n}) = \big(\square-m^{2}\big)\phi - \sum_{n\geq 2}\,\frac{\kappa_n}{n!}\,\phi^{\star n} \ ,
\end{align}
where $\phi^{\star n}:=\phi\star\cdots\star\phi$ ($n$ times). This is just the usual noncommutative deformation of scalar field equations. 

The cyclic inner product \eqref{eq:SCpairing} is invariant under the action of $\frk$, since $\LL_{\xi} \, \dd \vol_{g} = 0$ for all Killing vector fields $\xi\in \frk$. By twisting we obtain the braided cyclic pairing 
\begin{align} \label{eq:braidedSCpairing}
	\langle\phi,\varPhi\rangle_\star := \int_M\,
	\phi \star \varPhi \ \dd \vol_{g} \ ,
\end{align}
from which we may define the braided action functional following the prescription of Section~\ref{sec:braidedfieldtheory}:
\begin{align}\label{eq:NCscalaraction}
	\begin{split}
		S_{\star}(\phi):=\sum_{n\geq1} \, \frac{(-1)^{n\choose 2} }{(n+1)!}\,\langle\phi,\ell_{n}^\star (\phi^{\otimes n}) \rangle_\star =\int_{M}\, \Big(\frac{1}{2}\,\phi \star \big(\square- m^{2}\big)\phi - \sum_{n\geq 3} \, \frac{\kappa_{n-1}}{n!}\,\phi^{\star n} \Big) \ \dd \vol_{g} \ ,
	\end{split}
\end{align}
which, again, is the usual action functional for noncommutative scalar field theory; note that, for fields of suitable asymptotic decay, the kinetic term $\phi\star\square\,\phi \ \dd \vol_g$ may be integrated by parts to $\dd\phi\wedge_\star \ast_\hodge  \dd\phi$. Of course, in order to derive the field equations as the critical locus of this braided action functional in the usual sense, we have to restrict to compatible Drinfel'd twists which make the pairing strictly cyclic, as discussed in Section~\ref{sec:braidedfieldtheory}. 

We emphasize that even standard noncommutative scalar field theory is naturally encoded in a \emph{braided} $L_{\infty}$-algebra, even though no braided gauge symmetries are present; in other words, it is also an example of a braided field theory. 
This can be compared to the standard classical $L_\infty$-algebra formulation of noncommutative scalar field theory: in order to maintain strict symmetry of the brackets, as opposed to braided symmetry, one needs to instead symmetrize star-products of fields to write the higher brackets as
\begin{align}
\frac1{n!} \, \sum_{\sigma\in S_n} \, \ell_n^{\star}(\phi_{\sigma(1)},\dots,\phi_{\sigma(n)}) = -(-1)^{n\choose 2} \, \frac{\kappa_n}{n!} \, \sum_{\sigma\in S_n} \, \phi_{\sigma(1)}\star\cdots\star\phi_{\sigma(n)} \ ,
\end{align}
for $n\geq2$, where the sum runs over all permutations of degree $n$.
At the classical level this results in the same noncommutative field theory, but at the quantum level different interaction vertices arise. In particular, in the braided case one loses the distinction between planar and non-planar Feynman diagrams from dropping the symmetrization~\cite{Nguyen:2021rsa}, raising the possibility that no UV/IR mixing occurs in braided quantum field theory, which as a consequence would have drastically improved renormalizability features over the standard noncommutative quantum field theory~\cite{Szabo:2001kg}. The braided $L_\infty$-algebra formalism also alters the implementation of global symmetries, as we discuss in detail in Section~\ref{sec:globalsymscalar} below.

\subsection{Global braided symmetry}
\label{sec:globalsymscalar}

The global symmetries of the classical scalar field theory are the isometries of the spacetime $(M,g)$, which are generated by the Killing Lie algebra $\frk$ acting on fields via the Lie derivative. In the noncommutative scalar field theory, only the subalgebra of $\frk$ that commutes with the legs $\sff^k,\sff_k\in \sU\frk$ of the Killing twist $\CF$ generates global symmetries in this sense. However, irrespective of the twist, the noncommutative scalar field theory is invariant under the \emph{braided} Lie algebra of twisted Killing symmetries $(\frk[[\hbar]],[-,-]_\frk^\star)$ acting on fields via the {braided} Lie derivative. This is a finite-dimensional ``global'' or ``rigid'' braided Lie algebra of symmetries of the field theory. We shall now explain how this is also naturally understood within the braided $L_\infty$-algebra framework.

To incorporate these global symmetries, we extend the 2-term $L_\infty$-algebra of Section~\ref{sec:scalartheory} to a 4-term $L_\infty$-algebra by adding in non-zero homogeneous subspaces of degrees $0$ and $3$, with corresponding brackets for the action of the Lie algebra $\frk$.\footnote{\label{fn:globalsym} Strictly speaking, in view of the derived quotient interpretation of the kinematical part $V^0\oplus V^1$ of the original $L_\infty$-algebra, we should really be specifying an external action of the Lie algebra $\frk$ on $V^0\oplus V^1\oplus V^2\oplus V^3$ as we are not interested in the (derived) quotient by global symmetries.} Thus the underlying graded vector space is now $V=V^0\oplus V^1\oplus V^2\oplus V^3$ with
\begin{align}
V^0 = \frk \ , \quad V^1=V^2=C^\infty(M,\FR) \qquad \mbox{and} \qquad V^3 = \frk^* \ ,
\end{align}
where $\frk^*$ is the linear dual of $\frk$.
We denote elements in degree 0 by $\xi\in\frk$ and in degree 3 by $\xi^\dual\in\frk^*$. Note that these are not fields, but rather should be regarded as constant maps on $M$. 

Abstractly, the Lie algebra $\frk$ naturally acts on itself by the adjoint action, $\xi_1\cdot \xi_2:=[\xi_1,\xi_2]_\frk$ for all $\xi_1,\xi_2\in\frk$, and on its dual by the coadjoint action $\xi\cdot\xi^\dual$ defined by
\begin{align}
\big\langle \xi',\xi\cdot\xi^\dual\big\rangle = -\big\langle \xi\cdot\xi',\xi^\dual \big\rangle = - \big\langle [\xi,\xi']_\frk,\xi^\dual \big\rangle
\end{align}
for all $\xi,\xi'\in\frk$ and $\xi^\dual\in\frk^*$, where $\langle-,-\rangle:V^0\otimes V^3\to\FR$ is the canonical dual pairing between $\frk$ and $\frk^*$. Concretely, $\frk$ can be represented by vector fields on $M$ acting on fields via the Lie derivative, with $[-,-]_\frk$ the restriction of the Lie bracket of vector fields to the subalgebra $\frk$. These representations obviously all live in the category of $\sU\frk$-modules, and hence are amenable to  twisting by $\CF\in \sU\frk[[\hbar]]\otimes \sU\frk[[\hbar]]$ via the construction of Section~\ref{sec:braidedLinfty}. 

The braided $L_\infty$-algebra $(V[[\hbar]],\{\ell_n^\star\})$ is now extended by including, in addition to the brackets \eqref{eq:BraidedSCbrackets}, the non-zero $2$-brackets
\begin{align}
\begin{split}
\ell_2^\star(\xi_1,\xi_2) = -[\xi_1,\xi_2]_\frk^\star \ , \quad \ell_2^\star(\xi,\phi) = \LL_\xi^\star\, \phi \ , \quad \ell_2^\star(\xi,\varPhi) = \LL_\xi^\star\, \varPhi \qquad \mbox{and} \qquad \ell_2^\star(\xi,\xi^\dual) = \xi\star\xi^\dual \ ,
\end{split}
\end{align}
where $\xi\star\xi^\dual$ is the left braided coadjoint action of $\xi\in\frk$ on $\xi^\vee\in\frk^*[[\hbar]]$.
That these extra brackets are compatible with the braided homotopy Jacobi identities is a consequence of the fact that they are a twist deformation of the classical brackets of~\cite{BVChristian}, from which the statement follows trivially again for degree reasons and since the $\frk$-action commutes with the Laplace--Beltrami operator.\footnote{Note that~\cite{BVChristian} only considers  classical $\phi^4$-theory on Minkowski space, but the generalization is immediate.}
The cyclic structure $\langle-,-\rangle_\star$ is naturally extended to the new homogeneous components by twisting the canonical dual pairing $\langle-,-\rangle:V^0\otimes V^3\to\FR$ in the usual way. 

These brackets encode the global braided transformation rules for the scalar fields:
\begin{align}
\delta_\xi^{\br}\phi=\ell_2^\star(\xi,\phi)=\LL_\xi^\star\,\phi \ .
\end{align}
These close the braided Lie algebra
\begin{align}
\big[\delta_{\xi_1}^{\br},\delta_{\xi_2}^{\br}\big]_\circ^\star = \delta^{\br}_{-\ell_2^\star(\xi_1,\xi_2)} = \delta_{[\xi_1,\xi_2]_\frk^\star}^{\br} \ ,
\end{align}
which follows directly by the closure property of the braided Lie derivative $\LL_\xi^\star$. The braided field equations \eqref{eq:braidedSCfieldeqs} are covariant,
\begin{align}
\delta_\xi^{\br}\frF_\phi^\star = \ell_2^\star(\xi,\frF_\phi^\star) = \LL_\xi^\star\,\frF_\phi^\star \ ,
\end{align}
which follows directly from the braided derivation property of $\LL_\xi^\star$. Finally, the noncommutative action functional \eqref{eq:NCscalaraction} is invariant under these global braided transformations, $\delta_\xi^{\br}S_\star(\phi)=0$, which may also be checked directly. In the classical limit, these are all just the correct global symmetry transformation properties of classical scalar field theory.

\subsubsection*{Abelian twists}

Consider the special case where the Drinfel'd twist $\CF\in \sU\frt[[\hbar]]\otimes \sU\frt[[\hbar]]$ is constructed on the enveloping Hopf algebra of the Cartan subalgebra $\frt\subseteq\frk$, i.e. from a commuting set of Killing vector fields of the spacetime $(M,g)$. Then $\CF$ is an {abelian} twist (see Example~\ref{ex:abeliantwist}). For vector fields $\xi\in\frt$, the braided Lie derivative $\LL_\xi^\star=\LL_\xi$ coincides with the ordinary Lie derivative and the Lie bracket $[-,-]_\frt^\star=0$ is also unchanged by the twist deformation (in fact, all brackets involving elements from $\frt$ are unchanged). That is, the braided symmetry transformations from the Cartan subalgebra $\frt$ act exactly as in the classical theory, and we recover the anticipated statement that the noncommutative scalar field theory is invariant under the classical action of global symmetry transformations from $\frt$. However, this is not generally true for vector fields $\xi\notin\frt$. The novelty of our approach is that, when the Killing Lie algebra $\frk$ is nonabelian, the noncommutative scalar field theory is nevertheless invariant under the braided action of symmetry transformations from all of $\frk$. We consider in detail the traditional example in Section~\ref{sec:braidedLorentz} below. 

\subsection{Example: Braided Lorentz invariance}
\label{sec:braidedLorentz}

Let $M=\FR^{1,d-1}$ with the standard flat Minkowski metric $g=\eta$. The Killing Lie algebra of $(\FR^{1,d-1},\eta)$ is the Poincar\'e algebra
\begin{align}
\frk = \mathfrak{iso}(1,d-1) = \FR^{1,d-1} \rtimes \aso(1,d-1) \ .
\end{align}
The abelian translation algebra $\FR^{1,d-1}$ is generated by $\{\sfP_\mu\}_{\mu=0,1,\dots,d-1}$ and the nonabelian Lorentz algebra $\aso(1,d-1)$ by $\{\sfM_{\mu\nu}\}_{\mu,\nu=0,1,\dots,d-1}$, with $\sfM_{\mu\nu}=-\sfM_{\nu\mu}$ and the semi-direct product Lie brackets
\begin{align}
\begin{split}
[\sfP_\mu,\sfP_\nu]_\frk &= 0 \qquad , \qquad [\sfM_{\mu\nu},\sfP_\lambda]_\frk = \eta_{\mu\lambda}\,\sfP_\nu - \eta_{\nu\lambda}\,\sfP_\mu \ , \\[4pt]
[\sfM_{\mu\nu},\sfM_{\lambda\rho}]_\frk &= \eta_{\mu\lambda}\,\sfM_{\nu\rho} - \eta_{\mu\rho} \, \sfM_{\nu\lambda} - \eta_{\nu\lambda}\,\sfM_{\mu\rho} + \eta_{\nu\rho}\,\sfM_{\mu\lambda} \ ,
\end{split}
\end{align}
where $\eta = \frac12\,\eta_{\mu\nu} \, \dd x^\mu\otimes\dd x^\nu$ in coordinates $x=(x^\mu)\in\FR^{1,d-1}$; as previously we use Einstein summation conventions throughout. 

Denoting the basis vector fields of $\FR^{1,d-1}$ by $\partial_\mu=\frac\partial{\partial x^\mu}$, the action of these generators on fields may be realised by the vector fields implementing translations and Lorentz transformations
\begin{align}
\sfP_\mu = \partial_\mu \qquad \mbox{and} \qquad \sfM_{\mu\nu} = x_\mu\,\partial_\nu - x_\nu\,\partial_\mu \ .
\end{align}
Let $\CF\in\sU\FR^{1,d-1}[[\hbar]]\otimes\sU\FR^{1,d-1}[[\hbar]]$ be the Moyal--Weyl twist (see Example~\ref{ex:MWtwist})
\begin{align}
\CF = \exp\big(-\tfrac{\mathrm{i}\,\hbar}2 \, \theta^{\mu\nu} \, \sfP_\mu \otimes \sfP_\nu\big) \ ,
\end{align}
where $(\theta^{\mu\nu})$ is a constant antisymmetric real $d{\times} d$ matrix. As is well-known (see e.g.~\cite{Szabo:2006wx}), the noncommutative scalar field theory is then translationally invariant, but it breaks Lorentz invariance in the classical sense.

The full relativistic invariance of the noncommutative scalar field theory on $(\FR^{1,d-1},\eta)$ is restored by twisting the action of the Lorentz vector fields $\sfM_{\mu\nu}$ to a braided representation. A quick check reveals that the twist $\CF$ in this example acts trivially on the Lie bracket, $[-,-]_\frk^\star = [-,-]_\frk$, so the underlying braided Lie algebra of Killing vector fields coincides with the classical Poincar\'e algebra. A generic element of the Poincar\'e algebra
\begin{align}
\xi = \varepsilon^\mu\,\sfP_\mu + \tfrac12\,\omega^{\mu\nu}\,\sfM_{\mu\nu}
\end{align}
is specified by a set of real translation parameters $\varepsilon^\mu$ and real Lorentz transformation parameters $\omega^{\mu\nu}=-\omega^{\nu\mu}$. Its braided action on a scalar field $\phi$ is given by
\begin{align}
\delta_\xi^{\br} \phi = \LL_\xi^\star\,\phi = \varepsilon^\mu\,\partial_\mu\phi + \omega^{\mu\nu}\,x_\nu\star\partial_\mu\phi = \LL_\xi\phi - \tfrac{\mathrm{i}\,\hbar}2 \, \omega^{\mu\nu} \, \theta_{\nu\lambda}\, \partial^\lambda \partial_\mu\phi \ .
\end{align}
As expected, only Lorentz transformations of $\phi$ are modified from their classical rules; note that the transformations of the coordinate functions $(x\mapsto x^\mu):\FR^{1,d-1}\to\FR$ themselves are unchanged, as $\LL_\xi^\star x^\mu=\LL_\xi x^\mu$. These braided Poincar\'e transformations close the \emph{classical} Poincar\'e algebra. 

Nevertheless, our braided Lorentz transformations are \emph{not} the same as the `twisted Lorentz transformations' of~\cite{Chaichian:2004za,Wess:2003da}. Those act as the classical Lorentz symmetry on a single field, and do not follow a braided derivation rule on star-products of fields but rather a deformed Leibniz rule defined by the Drinfel'd twist of the coproduct $\Delta_\tCF$ on the enveloping Hopf algebra of $\frk$. It would be interesting to explore whether the standard Wigner classification of particles holds for braided representations of the Lorentz algebra, as it does for the twisted actions of~\cite{Chaichian:2004za}.

\section{Braided $BF$ theory}
\label{sec:braidedBFtheory}

Let us now return to examples of field theories with gauge symmetries. In this section we look at an example of a topological field theory which is a natural extension of Chern--Simons theory, considered as our prototypical guiding example of gauge theories in the (braided) $L_\infty$-algebra formalism in Sections~\ref{sec:Linfty} and~\ref{sec:braidedLinfty}. The story spelled out for Chern--Simons theory is easily adapted for arbitrary $BF$ theories, whose classical $L_\infty$-algebra formulation was worked out in detail by~\cite{ECPLinfty}. These examples serve to nicely illustrate generic features of field theories with higher (or reducible) gauge symmetries in a simple framework, and provide new examples of ensuing novelties with the corresponding higher braided gauge symmetries which also drastically deviate from their classical and star-gauge counterparts.

\subsection{$BF$ theory in the $L_\infty$-algebra formalism}
\label{sec:BFtheory}

Let $M$ be a closed oriented $d$-dimensional manifold. Let $(\frg,[-,-]_\frg)$ be a Lie algebra, and let $ W$ be a $\frg$-module with a nondegenerate pairing $\Tr_{\tW}:
 W\otimes \frg \rightarrow \FR$ which is invariant under the $\frg$-action:
\begin{align}\label{eq:BFpairinginvariance}
	\Tr_{\tW}\big( (X\cdot w)\otimes Y + w\otimes[X, Y]_\frg\big)
	=0 \ ,
\end{align}
for $w\in  W$ and $X,Y\in\frg$. The standard example is when $\frg$ is a quadratic Lie algebra, where $ W=\frg$ is a $\frg$-module under the adjoint action of $\frg$ on itself, as in Section~\ref{sec:whatisgauge}.

The $BF$ action functional is given by
\begin{align}\label{eq:BFaction}
	S(B,A)=\int_{M}\,\Tr_{\tW}( B\wedge F_{\tA}) \ ,
\end{align}
where $F_{\tA}\in\Omega^2(M,\frg)$ is
the curvature of a connection $1$-form $A \in \Omega^{1}(M,\frg)$ and $B\in\Omega^{d-2}(M, W)$. The field equations of $BF$ theory are solved by pairs of a flat $\frg$-connection on $M$ and a covariantly constant
$(d{-}2)$-form valued in the representation $ W$ of $\frg$:
\begin{align}
	\begin{split}
		\frF_{B}& := F_{\tA} = \dd A+ \tfrac{1}{2}\,[A,A]_\frg = 0 \ \in \
		\Omega^{2}(M,\frg) \ , \\[4pt]
		\frF_{\tA}&:= \dd^{\tA}B= \dd B + A\wedge B = 0 \ \in \ \Omega^{d-1}(M, W) \ ,
	\end{split}
\end{align}
where again $A\wedge B$ computes the exterior product of the form
components while pairing the components in $\frg$ and $ W$ via the
$\frg$-action. 

The action functional \eqref{eq:BFaction} is invariant under standard gauge transformations $\lambda \in \Omega^{0}(M,\frg)$ acting as 
\begin{align}
	\delta_{\lambda} (B,A) =\big(-\lambda \cdot B\,,\, \dd \lambda + [A,\lambda]_\frg\big) \ .
\end{align}
Compared to Chern--Simons theory, however, for $d\geq3$ there is an extra `shift' symmetry generated by $(d{-}3)$-forms $\tau \in \Omega^{d-3}(M, W)$ valued in $ W$, which act as
\begin{align} \label{eq:shift}
	\delta_{\tau}(B,A):=  (\dd\tau + A\wedge \tau, 0) \ .
\end{align}
The corresponding Noether identities coincide with the usual Bianchi identities
\begin{align}
	\begin{split}
		\dsf_{(B,A)}(\frF_B,\frF_{\tA}):\!&=\big((-1)^{d-3}\,\dd^{\tA}\frF_{B}\,,\,\dd^{\tA}\frF_{\tA}-\frF_{B} \wedge B\big)\\[4pt] 
		&=\big((-1)^{d-3}\, \dd^{\tA} F_{\tA} \, , \, (\dd^{\tA})^{2} B - F_{\tA}\wedge B\big) = (0,0) \ \in \ \Omega^3(M,\frg) \times \Omega^{d}(M, W) \ .
	\end{split}
\end{align}
Here the first slot is the second Bianchi identity which in this example corresponds to the shift symmetry, in contrast to Chern--Simons theory where it corresponds to the usual gauge transformations generated by $\Omega^{0}(M,\frg)$. The latter symmetries correspond to the first Bianchi identity in $BF$ theory, which is the second slot.

The cyclic $L_{\infty}$-algebra of $BF$ theory in $d$ dimensions is given by the underlying graded vector space $V:= V^{0} \oplus V^{1} \oplus V^{2} \oplus V^3$, 
where
\begin{align}\label{eq:ndBF}
	\begin{split}
		V^{0}=\Omega^{d-3}(M, W) \times \Omega^{0}\big(M,\frg\big) \qquad & , \qquad 
		V^{1}= \Omega^{d-2}(M, W) \times
		\Omega^{1}\big(M,\frg\big) \ ,  \\[4pt]
		V^{2}=\Omega^{2}\big(M,\frg\big) \times
		\Omega^{d-1}(M, W) \qquad & , \qquad
		V^3= \Omega^{3}(M,\frg)\times \Omega^d(M, W)
		\ . 
	\end{split}
\end{align} 
We denote gauge parameters by $(\tau,\lambda)\in V^{0}$, dynamical fields by $(B,A)\in V^{1}$, field equations by $(\CB,\CA)\in V^{2}$, and Noether identities by $(\CT,\varLambda)\in V^{3}$. The non-trivial brackets then comprise the differential
\begin{align}\label{eq:ndBFbrackets1}
	\ell_{1}(\tau,\lambda)=(\dd \tau,\dd\lambda) \  , \quad
	\ell_{1}(B,A)=(\dd A,\dd B) \qquad \mbox{and} \qquad 
	\ell_{1}(\CB,\CA)=(\dd \CB,\dd \CA) \ ,
\end{align}
along with the $2$-brackets
\begin{align}\label{eq:ndBFbrackets2}
\begin{split}
\ell_{2}\big((\tau_{1},\lambda_{1})\,,\,(\tau_{2},\lambda_{2})\big)&=\big(-\lambda_{1} \cdot \tau_{2} +\lambda_{2} \cdot \tau_{1}\,,\,-[\lambda_{1},\lambda_{2}]_\frg\big) \ , \\[4pt]
 \ell_{2}\big((\tau,\lambda)\,,\,(B,A)\big)&=\big(-\lambda \cdot B
	+A\wedge \tau \,,\, -[\lambda,A]_\frg\big) \ , \\[4pt]
\ell_{2}\big((\tau,\lambda) \,,\, (\CB,\CA)\big)&=\big(-
	[\lambda, \CB]_\frg \,,\, -\lambda\cdot \CA + \CB\wedge \tau \big) \ , \\[4pt]
\ell_{2}\big((\tau,\lambda)\,,\,(\CT,\varLambda)\big)&= \big
	(-[\lambda,\CT]_\frg\,,\,-\lambda \cdot \varLambda + (-1)^{d-3} \, \CT \wedge \tau 
	\big) \ , \\[4pt]
\ell_{2}\big((B_{1},A_{1})\,,\,(B_{2},A_{2})\big)&=-\big([A_{1},A_{2}]_\frg\, ,\, A_{1} \wedge B_{2} +A_{2} \wedge B_{1} \big) \ , \\[4pt]
\ell_{2}\big((B,A)\,,\,(\CB,\CA)\big)&=-\big([A,\CB]_\frg\,,\, A\wedge \CA - \CB \wedge B\big) \ .
\end{split}
\end{align}
Thus the dynamics of $BF$ theory in any dimension $d$ is also organised by a differential graded Lie algebra.
The cyclic inner product is given naturally as
\begin{align}\label{eq:BFpairing}
	\begin{split}
		\langle(B,A)\,,\,(\CB,\CA) \rangle &:= \int_{M}\, \Tr_{\tW}( B\wedge\CB) + \Tr_{\tW}(\CA\wedge A ) \ , \\[4pt]
		\langle(\tau,\lambda)\,,\,(\CT,\CP) \rangle &:= \int_{M}\, (-1)^{d-3} \, \Tr_{\tW}(\tau\wedge\CT) + \Tr_{\tW}( \varLambda \wedge \lambda ) \ .
	\end{split}
\end{align}

\subsection{Noncommutative $BF$ theory in the braided $L_\infty$-algebra formalism}
\label{sec:braidedBF}

The differential graded Lie algebra of $d$-dimensional $BF$ theory from Section~\ref{sec:BFtheory} lives in the category of $\sU\frv$-modules, as a consequence of diffeomorphism invariance of the topological field theory. Applying the general formalism of braided field theory from Section~\ref{sec:braidedfieldtheory}, we twist deform it to the 4-term differential graded braided Lie algebra $\big(V[[\hbar]], \ell_{1}^\star, \ell_2^\star \big)$ with $1$-bracket $\ell_1^\star = \ell_1$ and the non-trivial $2$-brackets
\begin{align}\label{eq:braidedBFbrackets}
	\begin{split}
		\ell_{2}^{\star}\big((\tau_{1},\lambda_{1})\,,\,(\tau_{2},\lambda_{2})\big)&=\big(-\lambda_{1} \star \tau_{2} +\sfR_k(\lambda_{2}) \star \sfR^k(\tau_{1})\,,\,-[\lambda_{1},\lambda_{2}]^\star_\frg\big)  \ , 
		\\[4pt]
		\ell_{2}^{\star}\big((\tau,\lambda)\,,\,(B,A)\big)&=\big(-\lambda \star B
		+\sfR_k(A)\wedge_\star \sfR^k(\tau) \,,\, -[\lambda,A]_\frg^\star\big) \
		,  \\[4pt]
		\ell_{2}^{\star}\big((\tau,\lambda) \,,\, (\CB,\CA)\big)&=\big(-
		[\lambda, \CB]^\star_\frg \,,\, -\lambda\star \CA + \sfR_k(\CB)\wedge_\star \sfR^k(\tau) \big) \ ,\\[4pt]
		\ell_{2}^{\star}\big((\tau,\lambda)\,,\,(\CT,\varLambda)\big)&= \big
		(-[\lambda,\CT]^\star_\frg\,,\,-\lambda \star \varLambda + (-1)^{d-3} \, \sfR_k (\CT) \wedge_\star \sfR^k (\tau)
		\big)  \ ,  \\[4pt] 
		\ell_{2}^{\star}\big((B_{1},A_{1})\,,\,(B_{2},A_{2})\big)&=-\big([A_{1},A_{2}]^\star_\frg\, ,\, A_{1} \wedge_\star B_{2} +\sfR_k (A_{2}) \wedge_\star \sfR^k(B_{1}) \big) \ ,  \\[4pt]
		\ell_{2}^{\star}\big((B,A)\,,\,(\CB,\CA)\big)&=-\big([A,\CB]^\star_\frg\,,\, A\wedge_\star \CA - \sfR_k(\CB) \wedge_\star \sfR^k(B)\big) \ .
	\end{split} 
\end{align}
The notation $\lambda  \star  B $ again stands for the left braided representation of $\lambda \in \Omega^{0}(M,\frg)$ on $B\in \Omega^{d-2}(M, W)[[\hbar]]$ as described in Section~\ref{sec:braidedgaugesym}, and similarly for $A \wedge_\star B$.

A pair of gauge parameters $(\tau,\lambda) \in \Omega^{d-3}(M, W) \times \Omega^{0}(M,\frg)$ defines the (left) braided gauge transformations of the fields $(B,A)\in \Omega^{d-2}(M, W) \times \Omega^{1}(M,\frg)$
by 
\begin{align}\label{eq:braidedBFgauge}
\begin{split}
	\delta_{(\tau,\lambda)}^{\br} (B,A):\!&= \ell_{1}^\star(\tau,\lambda) + \ell_{2}^\star\big((\tau,\lambda)\, ,\, (B,A)\big) \\[4pt]
&= \big(\dd \tau +\sfR_k (A)\wedge_\star \sfR^k (\tau)-\lambda\star B \, ,\, \dd \lambda - [\lambda,A]^\star_\frg\big ) \ ,
\end{split}
\end{align}
encoding the expected deformation of the symmetries. Notice that in the braided framework, the shift symmetry is twisted naturally irrespective of the dimension $d$, gauge algebra $\frg$ and representation $ W$. The closure as a braided Lie algebra 
\begin{align}
\big[\delta_{(\tau_1,\lambda_1)}^{\br},\delta_{(\tau_2,\lambda_2)}^{\br}\big]_\circ^\star = \delta_{-\ell_2^\star((\tau_1,\lambda_1),(\tau_2,\lambda_2))}^{\br} = \delta^{\br}_{(\lambda_{1} \star \tau_{2} -\sfR_k(\lambda_{2}) \star \sfR^k(\tau_{1}),[\lambda_{1},\lambda_{2}]^\star_\frg)}
\end{align}
follows naturally as in the commutative case, giving a braided deformation of the semi-direct product
\begin{align}
\Omega^{d-3}(M, W)[[\hbar]]\rtimes_\star \Omega^0(M,\frg)[[\hbar]] \ ,
\end{align}
where we regard the vector space $ W$ (and hence $\Omega^{d-3}(M, W)[[\hbar]]$) as an abelian Lie algebra.

\subsubsection*{Braided $BF$ equations}

Following the prescription from Section~\ref{sec:braidedfieldtheory}, the field equations $\frF_{(B,A)}^\star=0$ valued in $\Omega^2(M,\frg)[[\hbar]]\times \Omega^{d-1}(M, W)[[\hbar]]$ are encoded by 
\begin{align}\label{eq:braidedBFeom}
\begin{split}
	\frF_{(B,A)}^{\star}= (\frF_B^\star , \frF_{\tA}^\star):\!&=\ell_1^\star(B,A) - \tfrac12\,\ell_2^\star\big((B,A)\, ,\, (B,A)\big) \\[4pt]
	&=\big(\dd A + \tfrac{1}{2}\,[A,A]^\star_\frg \, ,\,  \dd B + A\wedge_\star B + \sfR_k(A) \wedge_\star \sfR^k (B) \big) \\[4pt]
	&=\big(F_{\tA}^\br \, , \, \tfrac{1}{2}\,(\dd_{\star\lact}^{\tA} B+ \dd_{\star\ract}^{\tA} B ) \big) \ , 
\end{split}
\end{align}
where in the third line we have identified the left and right braided covariant derivatives from Section~\ref{sec:braidedgaugesym}. The field equations state that the braided $\frg$-connection $A$ is flat while the field $B$ obeys a symmetrized braided covariant constancy condition. Such symmetrized combinations are expected since they preserve reality of the fields (cf.\ Example~\ref{ex:realcovder}), and hence are a generic feature of braided field theories. 

According to the general braided $L_\infty$-algebra formulation, we obtain the braided gauge covariance
\begin{align}\label{eq:braidedBFcov}
	\begin{split}
		\delta_{(\tau,\lambda)}^{\br}\frF_{(B,A)}^\star &= \ell_2^\star\big((\tau,\lambda) \, , \, (\frF_{B}^\star, \frF_{\tA}^\star)\big) =\big(-[\lambda,\frF_{B}^\star]_\frg^\star \, , \, - \lambda \star \frF_{\tA}^\star + \sfR_k (\frF_{B}^\star)\wedge_\star \sfR^k (\tau) \big)\ ,
	\end{split}
\end{align} 
which may also be easily verified explicitly. This extends the statement from Section~\ref{sec:braidedgaugesym} that the braided curvature, left and right covariant derivatives are covariant under gauge transformations from $\Omega^{0}(M,\frg)$.

\subsubsection*{Braided Noether identities}

From the general prescription of Section~\ref{sec:braidedfieldtheory}, the braided Noether identities which are valued in $\Omega^3(M,\frg)[[\hbar]] \times \Omega^{d}(M, W)[[\hbar]]$ are given by
\begin{align}\label{eq:braidedBFnoether}
\begin{split}
	\dsf_{(B,A)}^\star \frF_{(B,A)}^\star :\!&= \ell_1^\star(\frF_{B}^\star, \frF_{\tA}^\star) - \tfrac12\,\Big(\ell_2^\star
	\big((B,A)\,,\, ( \frF_{B}^\star,  \frF_{\tA}^\star)\big) - \ell_2^\star\big(( \frF_{B}^\star,  \frF_{\tA}^\star)\, ,\, (B,A)\big)\Big) \\
	& \quad \,+ \tfrac14\,\ell_2^\star\big(\sfR_k(B,A)\,,\,\ell_2^\star(\sfR^k(B,A)\, ,\, (B,A))\big) \ = \ (0,0) \ ,
\end{split}
\end{align}
where $\sfR_k (B,A) :=  \big( \sfR_k (B), \sfR_k (A )\big)$. These translate to a pair of differential identities among the field equations~\eqref{eq:braidedBFeom}. Expanding and collecting terms, the first slot of \eqref{eq:braidedBFnoether} is the Noether identity corresponding to the braided shift symmetry:
\begin{align}\label{eq:braidedshiftNoether}
	\tfrac12\,\big(\dd_{\star\lact}^{\tA}  \frF_B^\star + \dd_{\star\ract}^{\tA}  \frF_B^\star\big) + \tfrac14\,\big[\sfR_k(A),[\sfR^k(A),A]_\frg^\star\big]_\frg^\star \ = \ 0 \ . 
\end{align}

The second slot is the Noether identity corresponding to the usual braided gauge transformations generated by $\Omega^{0}(M,\frg)$:
\begin{align}\label{eq:braidedBFgNoether}
	\begin{split}
		\tfrac12\,&\big(\dd_{\star\lact}^{\tA}  \frF_{\tA}^\star + \dd_{\star\ract}^{\tA}  \frF_{\tA}^\star - \sfR_k( \frF_{B}^\star)\wedge_\star \sfR^k(B) -  \frF_{B}^\star \wedge_\star B \big) \\ 
		& \hspace{2cm}+\tfrac14\, \Big(-[A,\sfR_k(A)]^\star_\frg \wedge_\star \sfR^k(B) + \sfR_k(A)\wedge_\star \big(\sfR^k(A)\wedge_\star B\big) \\& {\phantom{\hspace{2cm}+\tfrac14\, \big(-[A,\sfR_k(A)]^\star_\frg \wedge_\star \sfR^k(B)+}} + {\sfR_{k}}_{\bar{\textrm{\tiny(1)}}}(A) \wedge_\star \big({\sfR_{k}}_{\bar{\textrm{\tiny(2)}}}(A)\wedge_\star \sfR^k (B)\big) \Big)  \ = \ 0 \ ,
	\end{split}
\end{align}
where we used the Sweedler notation $\Delta_\tCF(\sfR_k) =: {\sfR_{k}}_{\bar{\textrm{\tiny(1)}}}\otimes{\sfR_{k}}_{\bar{\textrm{\tiny(2)}}}$ (with summations understood).
This identity may also be derived 
directly by summing over the braided deformations of the first Bianchi identity in twisted noncommutative differential geometry:
\begin{align}\label{eq:braidedBFfirstbianchi}
	\begin{split}
		(\dd_{\star\lact}^{\tA})^{2} B &= F_{\tA}^\br \wedge_\star B +\tfrac{1}{2}\, A\wedge_\star (A\wedge_\star B) - \tfrac{1}{2} \, \sfR_k (A)\wedge_\star 
		\big(\sfR^k (A) \wedge_\star B\big)\ , \\[4pt]
		(\dd_{\star\ract}^{\tA})^{2} B &= \sfR_k (F_{\tA}^\br) \wedge_\star \sfR^k (B) \\ & \quad \, +\tfrac{1}{2}\, \sfR_k\, \sfR_l (A) \wedge_\star \big(\sfR_m\, \sfR^l (A) \wedge_\star \sfR^m\, \sfR^k (B) \big)  -\tfrac{1}{2} \,  {\sfR_{k}}_{\bar{\textrm{\tiny(1)}}}(A) \wedge_\star \big({\sfR_{k}}_{\bar{\textrm{\tiny(2)}}}(A)\wedge_\star \sfR^k (B)\big) \ , 
	\end{split}
\end{align}
and the mixed derivative expressions. These follow by simply expanding the braided covariant derivatives. The last two terms cancel each other in the classical limit, giving the usual identity $(\dd^{\tA})^{2} B= F\wedge B$. 

\subsubsection*{Braided $BF$ functional}

When $M$ is closed and oriented, the inner product \eqref{eq:BFpairing} twists to a (braided) cyclic pairing given by 
\begin{align}\label{eq:braidedBFpairing}
	\begin{split}
		\langle(B,A)\,,\,(\CB,\CA) \rangle_\star &:= \int_{M}\, \Tr_{\tW}( B\wedge_\star \CB) + \Tr_{\tW} ( \CA \wedge_\star A ) \ , \\[4pt]
		\langle(\tau,\lambda)\,,\,(\CT,\varLambda) \rangle_\star &:= \int_{M}\,(-1)^{d-3} \, \Tr_{\tW}( \tau\wedge_\star \CT ) + \Tr_{\tW} ( \varLambda \wedge_\star \lambda ) \ .
	\end{split}
\end{align}
We shall write out the braided action functional in two forms: First by using a braided cyclic pairing (or forgetting the $\wedge_\star$-cyclicity in the case of compatible twists), and then by employing a compatible twist and its resulting $\wedge_\star$-cyclicity under the integral. This is instructive for considerations of more complicated braided field theories.

For a braided cyclic pairing, the action functional is given by the prescription of Section~\ref{sec:braidedfieldtheory}, which in the example at hand reads
\begin{align}
\begin{split}
		S_\star(B,A) :\!&= \frac12\,\big\langle (B,A)\, ,\, \ell_1^\star(B,A)\big\rangle_\star - \frac16 \, \big\langle (B,A)\,  , \, \ell_2^\star\big((B,A),(B,A)\big)\big\rangle_\star  \\[4pt]
		&= \frac{1}{2}\, \int_M\, \Tr_{\tW} \big( B\wedge_\star \dd A \big)+ \Tr_{\tW} \big( A\wedge_\star \dd B \big)  \\
		& \quad \, +\frac{1}{6}\, \int_M\, \Tr_{\tW} \big( B\wedge_\star [A,A]^\star_\frg \big) +\Tr_{\tW} \Big( A \wedge_\star \big( A\wedge_\star B + \sfR_k(A)\wedge_\star \sfR^k (B)\big) \Big) \\[4pt]
		&= \int_M\, \frac{1}{2}\, \Tr_{\tW} \big( B\wedge_\star \dd A \big)+\frac{1}{2} \, \Tr_{\tW} \big( \dd A\wedge_\star B \big) \\ 
		& \qquad \qquad \,  + \frac{1}{6} \, \Tr_{\tW} \big( B\wedge_\star [A,A]_\frg^\star \big) +\frac{1}{6}\, \Tr_{\tW} \big( \sfR_k (B)\wedge_\star \sfR^k ([A,A]_\frg^\star) \big) \\ 
		&\qquad \qquad \, + \frac{1}{6}\, \Tr_{\tW}\big(\sfR_k(B)\wedge_\star [\sfR^k (A), A]_\frg^\star \big) \ ,
\end{split}
\end{align}
where in the last equality we integrated by parts and used invariance of $\Tr_{\tW}(-\wedge_\star-)$ under the left braided $\frg$-action, which follows from the classical invariance \eqref{eq:BFpairinginvariance} order by order in the deformation parameter $\hbar$. Identifying the braided curvature $2$-form $F_{\tA}^\br=\dd A+\frac12\,[A,A]_\frg^\star$, we arrive at
\begin{align}\label{eq:braidedBFactionnoncycl}
\begin{split}
		S_\star(B,A)& =\int_{M}\, \frac{1}{2}\, \Big( \Tr_{\tW} \big(B\wedge_\star F_{\tA}^\br \big) + \Tr_{\tW} \big(F_{\tA}^\br \wedge_\star B \big) \Big ) 
		\\ & \qquad \qquad \, - \frac{1}{12}\, \Big(\Tr_{\tW} \big( B\wedge_\star [A,A]_\frg^\star \big) + \Tr_{\tW} \big( \sfR_k (B)\wedge_\star \sfR^k ([A,A]_\frg^\star) \big) \Big)  \\
		&\qquad \qquad \, + \frac{1}{6} \, \Tr_{\tW}\big(\sfR_k(B)\wedge_\star [\sfR^k( A), A]_\frg^\star \big) \ .
\end{split}
	\end{align}
	
The first two terms of \eqref{eq:braidedBFactionnoncycl} are ``naive'' deformations which give the classical action functional \eqref{eq:BFaction} in the classical limit. The last three terms cancel each other in the classical limit, but are individually non-covariant. However, one may easily check that together they form a braided covariant combination. The possible non-covariance arises from terms proportional to $\dd \lambda$ in a gauge transformation of $A$, while the covariant terms vanish by invariance of $\Tr_{\tW}$. Indeed one finds
	\begin{align}
\begin{split}
		\delta_{\lambda}^{\br} \,  \Tr_{\tW} \big( B\wedge_\star [A,A]_\frg^\star \big) &= \Tr_{\tW} \big( \sfR_k(B) \wedge_\star [\sfR^k(\dd \lambda),A]_\frg^\star \big) \\ & \quad \, + \Tr_{\tW} \big( \sfR_k(B) \wedge_\star [\sfR_l (A), \sfR^l\, \sfR^k (\dd \lambda)]_\frg^\star \big) \\[4pt]
		&= 2 \, \Tr_{\tW} \big( \sfR_k(B) \wedge_\star [\sfR^k(\dd \lambda),A]^\star_\frg \big) \ , 
\end{split}
	\end{align}	
	where we used the braided symmetry of the bracket in the second equality. Similarly, we find
	\begin{align}
		\delta_{\lambda}^{\br} \,  \Tr_{\tW} \big( \sfR_k (B)\wedge_\star \sfR^k( [A,A]_\frg^\star) \big) = 2 \, \Tr_{\tW} \big(\sfR_k (B) \wedge_\star \sfR^k( [\dd \lambda, A]_\frg^\star) \big) 
\end{align}
and
\begin{align}
\begin{split}
		\delta_{\lambda}^{\br} \,  \Tr_{\tW} \big( \sfR_k (B)\wedge_\star [\sfR^k (A),A]_\frg^\star \big) &= \Tr_{\tW} \big(\sfR_k(B)\wedge_\star [\sfR^k (\dd\lambda),A]_\frg^\star\big) \\
& \quad \, + \Tr_{\tW} \big( \sfR_k(B) \wedge_\star \sfR^k ([\dd \lambda, A]_\frg^\star) \big) \ .
\end{split}
	\end{align}
Adding the three variations with the prefactors in \eqref{eq:braidedBFactionnoncycl} shows that they cancel, and the action functional is braided gauge invariant as expected by the general discussion of Section~\ref{sec:braidedfieldtheory}.
	
Restricting to compatible Drinfel'd twists, the pairing becomes a strictly cyclic structure on the braided $L_\infty$-algebra $\big(V[[\hbar]], \ell_{1}^\star, \ell_2^\star \big)$. In this case, the action functional 
\eqref{eq:braidedBFactionnoncycl} further simplifies to
\begin{align}\label{eq:braidedBFaction}
	S_\star(B,A) = \int_M\, \Tr_{\tW} \big( B\wedge_\star F_{\tA}^\br \big) \ , 	
\end{align}
where the first two terms of \eqref{eq:braidedBFactionnoncycl} combine by $\wedge_\star$-cyclicity, while the last three terms cancel for the same reason together with an $\CR$-matrix identity \eqref{eq:Rmatrixidsw} applied to the final term. Thus the braided noncommutative $BF$ functional is the expected ``naive'' deformation. By varying it with respect to each of fields $B$ and $A$, it is easy to explicitly derive the field equations \eqref{eq:braidedBFeom} using $\wedge_\star$-cyclicity under the integral, invariance of $\Tr_{\tW}$ and $\CR$-matrix identities \eqref{eq:Rmatrixidsw}. Similarly, braided gauge invariance under \eqref{eq:braidedBFgauge} and the corresponding Noether identities \eqref{eq:braidedBFnoether} are verified explicitly. For Hermitian twists and real fields the action functional is real-valued as discussed in Section~\ref{sec:braidedfieldtheory}.

\subsection{Higher braided gauge symmetry}

Another new feature of $BF$ theories, compared to Chern--Simons gauge theory, is that they generally possess additional
higher shift symmetries: The shift
symmetry \eqref{eq:shift} is on-shell reducible in dimensions
$d\geq4$. This means that, strictly speaking, we should also include
in \eqref{eq:ndBF} the negatively-graded vector spaces
$V^{-k}= \Omega^{d-3-k}(M, W)$ for
$k=1,\dots,d-3$, which parameterize `higher gauge transformations',
together with their duals $V^{k+3}=\Omega^{k+3}(M,\frg)$ and the obvious brackets in
\eqref{eq:ndBFbrackets1} and \eqref{eq:ndBFbrackets2}. We spell out this out explicitly for
$BF$ theories in the simplest case of dimension $d=4$.

Any element $\tau \in \Omega^{1}(M, W)$ generates the shift symmetry
\begin{align}
	\delta_{\tau} (B,A)= \big(\dd^{\tA}\tau,0\big) \ \in \ \Omega^{2}(M, W) \times \Omega^1(M,\frg) \ . 
\end{align}
Now consider $\tau':=\tau + \dd^{\tA}\varepsilon\, \in \Omega^{1}(M, W)$ for any $\varepsilon \in \Omega^{0}(M, W)$. This generates the shift symmetry 
\begin{align}
	\begin{split}
		\delta_{\tau'} (B,A) &= \big(\dd^{\tA}\tau',0\big) = \big(\dd^{\tA}\tau + (\dd^{\tA})^{2} \varepsilon,0\big) = \big(\delta_{\tau}B + F_{\tA}\cdot \varepsilon,0\big) = \big(\delta_{\tau}B + \frF_{B} \cdot \varepsilon,0\big) \ ,
	\end{split}
\end{align}
and so the two transformations differ by a term proportional to a
field equation, that is, $\tau$ and $\tau'=: \tau +
\delta_{(\varepsilon,A)}\tau$  generate the same symmetry on-shell. This
leads to a redundancy in the subspace of $V^{0}$ generating the action of shift symmetries on
$V^{1}$; the redundancy lives in the subspace of
covariantly exact $1$-forms valued in $ W$. We may parameterize this by
$V^{-1}:= \Omega^{0}(M, W)$; for $d=4$ this is
enough and no further gauge redundancy in the description exists,
while in higher dimensions $d>4$ the form degrees change and similarly higher-to-higher gauge transformations are required, and so on.

Following~\cite{ECPLinfty}, let us now describe the complete extended $L_\infty$-algebra of $BF$ theory in $d=4$ dimensions. We
extend the cochain complex by introducing 
\begin{align}
	V^{-1}:= \Omega^{0}(M, W) \qquad \text{and} \qquad
	V^{4}:= \Omega^{4}(M,\frg) \ ,
\end{align}
and denote the corresponding elements by $\varepsilon \in V^{-1}$ and $\CE \in V^{4}$. The brackets
\eqref{eq:ndBFbrackets1} and \eqref{eq:ndBFbrackets2} are extended as 
\begin{align}
	\ell_{1} (\varepsilon)=(\dd \varepsilon,0 ) \qquad \text{and} \qquad
	\ell_{1} (\CT,\varLambda)=\dd \CT  \ ,
\end{align}
together with
\begin{align}
	\begin{split}
		\ell_{2} \big(\varepsilon\,,\,(\tau,\lambda)\big)= \lambda
		\cdot
		\varepsilon 
		\quad , \quad
		\ell_{2} \big(\varepsilon\,,\,(B,A) \big)&=-(A\cdot
		\varepsilon,0) \quad , \quad
		\ell_{2} \big(\varepsilon\,,\,(\CB,\CA)\big)=(\CB
		\cdot
		\varepsilon,0)
		\ , \\[4pt]
		\ell_{2} \big(\varepsilon\,,\,
		(\CT,\varLambda)\big)=-(0,\CT\cdot \varepsilon) \quad , \quad
		\ell_{2} (\varepsilon,\CE)&=(0,\CE \cdot \varepsilon) \quad , \quad
		\ell_{2}\big((\tau,\lambda)\,,\,\CE\big)=-[\lambda,\CE]_\frg \ , \\[4pt]
		\ell_{2}\big((B,A)\,,\,(\CT,\varLambda)\big)=-[A,\CT]_\frg \quad & , \quad
		\ell_{2}\big((\CB_{1},\CA_{1})\,,\,(\CB_{2},\CA_{2})\big)=-[\CB_{1},\CB_{2}]_\frg \ .
	\end{split}
\end{align}
The cyclic inner product extends naturally as 
\begin{align}
	\langle \varepsilon, \CE \rangle  = \int_{M}\,
	\Tr_{\tW}(\varepsilon \wedge \CE) \ .
\end{align}

The higher gauge transformation is now encoded in the expected way via~\eqref{eq:highergauge} as
\begin{align}
\delta_{(\varepsilon,(B,A))}(\tau,\lambda) = \ell_{1}(\varepsilon) - \ell_2\big(\varepsilon,(B,A)\big) = \big(\dd^{A} \varepsilon, 0\big) \ .
\end{align}
Its adjoint defines the higher Noether operator
\begin{align}
\dsf_{(B,A)}^{\swone}(\CT,\varLambda) = \ell_1(\CT,\varLambda)- \ell_2\big((B,A),(\CT,\varLambda)\big)= \dd^{A}\CT \ ,
\end{align}
which encodes the level 1 Noether identity via~\eqref{eq:level1Noether} as
\begin{align}
	\dsf_{(B,A)}^{\swone}\circ \dsf_{(B,A)} (\CB,\CA) =-\ell_{2}\big((\frF_{B},\frF_A),(\CB,\CA)\big)= [\frF_{B},\CB]_\frg \ .
\end{align}

The extension of the $L_\infty$-algebra carries through the usual deformation to produce a braided $L_\infty$-algebra without any issues. However, we currently do not have a good physical interpretation of these `higher' braided gauge symmetries. The issue is related to the problem of defining the braided moduli space of classical solutions. We first briefly describe the problem without explicitly writing out the detailed extended braided $L_\infty$-algebra, and then show how it is a general consequence within the braided $L_\infty$-algebra framework.

For $d=4$ the braided shift symmetry generated by $\tau \in \Omega^{1}(M, W)$ acts as $\delta_\tau^{\br} (B,A) = (\dd^{\tA}_{\star\ract} \tau,0) $. For a higher gauge parameter $\varepsilon \in \Omega^{0}(M, W)$, we would expect the deformation of the action on gauge parameters to be $\delta_{(\varepsilon,A)}^{\br} \tau = \dd^{\tA}_{\star\ract} \varepsilon\, \in \Omega^{1}(M, W)[[\hbar]]$. This, however, does not seem to generate an equivalent (in the classical sense) braided gauge transformation via \smash{$\tau' := \tau + \dd^{\tA}_{\star\ract} \varepsilon$}, since
	\begin{align}\label{eq:higherBFbroken}
\begin{split}
		\delta_{\tau'}^{\br} (B,A) &= \big(\dd^{\tA}_{\star\ract} \tau',0\big)  \\[4pt]
		&= \big(\dd^{\tA}_{\star\ract} \tau + (\dd^{\tA}_{\star\ract})^{2} \varepsilon,0\big) \\[4pt]
		&= \big(\delta_{\tau}^{\br} B + \sfR_k (\frF^\star_B) \star \sfR^k (\varepsilon)  - \tfrac{1}{2}\, \sfR_k(A)\wedge_\star \big(\sfR_l(A)\star \sfR^l\, \sfR^k (\varepsilon) \big) \\
		& \hspace{5cm}+\tfrac{1}{2}\, \sfR_k\, \sfR_l (A) \wedge_\star \big(\sfR_m\, \sfR^l (A) \star \sfR^m\, \sfR^k (\varepsilon) \big) \,,\, 0\big) \ . 
\end{split}
	\end{align} 
That is, they do not generate the same symmetry on-shell. The last two polynomial terms spoil this symmetry, which is a consequence of the breaking of the first Bianchi identity \eqref{eq:braidedBFfirstbianchi} used in the third equality. We now explain that this is a generic feature of any higher braided gauge theory, by using its braided $L_\infty$-algebra incarnation.

For this, we extend the braided brackets \eqref{eq:braidedBFbrackets} by twisting the corresponding classical extended brackets, with $\ell^\star_1= \ell_1$ as usual and the braided 2-brackets
 \begin{align}
 	\begin{split}
 		\ell^\star_{2} \big(\varepsilon\,,\,(\tau,\lambda)\big)= \lambda
 		\star
 		\varepsilon 
 		\qquad & , \qquad
 		\ell_{2}^\star \big(\varepsilon\,,\,(B,A) \big)=-\big(\sfR_k (A)\star \sfR^k
 		(\varepsilon),0\big) \ , \\[4pt]
 		\ell^\star_{2} \big(\varepsilon\,,\,(\CB,\CA)\big)=\big(\sfR_k(\CB)
 		\star
 		\sfR^k(\varepsilon),0\big) \qquad & 
 		, \qquad
 		\ell_{2}^\star \big(\varepsilon\,,\,
 		(\CT,\varLambda)\big)=-\big(0,\sfR_k(\CT)\star \sfR^k(\varepsilon)\big) \ , \\[4pt]
 		\ell_{2}^\star (\varepsilon,\CE)=\big(0,\sfR_k(\CE) \star \sfR^k(\varepsilon)\big) \qquad & , \qquad
 		\ell_{2}^\star\big((\tau,\lambda)\,,\,\CE\big)=-[\lambda,\CE]_\frg^\star \ , \\[4pt]
 		\ell_{2}^\star\big((B,A)\,,\,(\CT,\varLambda)\big)=-[A,\CT]_\frg^\star \qquad & , \qquad
 		\ell_{2}^\star\big((\CB_{1},\CA_{1})\,,\,(\CB_{2},\CA_{2})\big)=-[\CB_{1},\CB_{2}]_\frg^\star \ .
 	\end{split}
 \end{align}
The cyclic inner product extends as
 \begin{align}
 	\langle \varepsilon, \CE \rangle_\star  = \int_{M}\,
 	\Tr_{ \tW}(\varepsilon \wedge_\star \CE) \ .
\end{align}

Using these brackets, we see that the (left) higher braided gauge transformation is expressed exactly as in the classical case~\eqref{eq:highergauge}, that is 
\begin{align}
\delta_{(\varepsilon,(B,A))}^{\br} (\tau,\lambda) = \big(\dd^{A}_{\star\ract} \varepsilon,0\big)=\ell^\star_{1} (\varepsilon) - \ell_{2}^\star\big(\varepsilon \, , \, (B,A)\big) \ .
\end{align}
However, using the braided graded Jacobi identities one easily verifies that the classical closure property \eqref{eq:highergaugecov} no longer holds, as
\begin{align}
\begin{split}
\delta_{(\varepsilon,(B, A))}^{\br}\big(\delta_{\tau}^{\br} (B,A)\big)&=\ell_{2}^\star\big(\varepsilon, \frF^\star_{(B,A)}\big) +\ell_{2}^\star\big(\ell_{2}^\star\big(\varepsilon,\sfR_k(B,A)\big)\,,\, \sfR^k(B,A)\big) \\ & \hspace{4cm} -\tfrac{1}{2}\, \ell_2^\star\big(\varepsilon,\ell_{2}^\star\big((B,A)\,,\,(B,A)\big)\big)\ ,
\end{split}
\end{align}
whereby expanding the brackets explicitly recovers the formula~\eqref{eq:higherBFbroken}, as expected. Since this is a feature of any differential graded braided Lie algebra, this will be a generic trait of any higher gauge theory based on a braided $L_\infty$-algebra, where higher brackets will generally contribute higher tensor powers of the fields to this expression.

One may also use cyclicity of the extended pairing to define the adjoint \smash{$\dsf^{\swone\br}_{(B,A)}:V^3[[\hbar]]\to V^4[[\hbar]]$} of the operator $\delta_{(\varepsilon, (B,A))}^{\br}$, exactly as in the classical case~\eqref{eq:level1NoetherOp}. However, the meaning of this operator is currently unclear, due to both the breakdown of the on-shell closure propery of its adjoint operation and also the peculiar definition of the level 0 Noether operator~\eqref{braidedNoether} in the braided field theory. It would be interesting to further explore these issues and develop a proper formulation of higher braided gauge theories. As before, the braided BV formalism of Section~\ref{sec:braidedBV} circumvents these issues, while an appropriate higher category generalization of the description of Section~\ref{sec:braidedMC} ought to make sense of higher braided gauge symmetries in the setting of the braided Maurer--Cartan moduli spaces discussed there.\footnote{This would require the $\infty$-category setting mentioned in Footnote~\ref{fn:inftycategory}.}

\section{Braided noncommutative Yang--Mills theory}
\label{sec:braidedNCYM}

In this final section we look at our first example of a noncommutative field theory with gauge symmetry based on a braided $L_\infty$-algebra which is not a differential graded braided Lie algebra. We twist deform the $L_\infty$-algebra formulation of Yang--Mills theory discussed in Section~\ref{sec:Yang-Millstheory} to uncover a new noncommutative deformation of Yang--Mills theory with novel and surprising features. This is in marked contrast to scalar field theory, whose braided deformation gives the standard noncommutative field theory, and the braided Chern--Simons and $BF$ theories, which are based on differential graded braided Lie algebras, and whose field equations and variational principles are also given by the naive replacements of products of fields with star-products.\footnote{Of course, this does not preclude the different kinematical sectors of these theories due to the braided Lie brackets, as well as the novel general nature of classical solutions in braided field theory and the corresponding modified forms of gauge redundancies through the braided Noether identities.}

\subsection{Noncommutative Yang--Mills theory in the braided $L_\infty$-algebra formalism}
\label{sec:braidedYMtheory}

In this section we consider the braided noncommutative deformation of the Yang--Mills gauge theory described in Section~\ref{sec:Yang-Millstheory}. As in the example of noncommutative scalar field theory, we suppose now that the background spacetime $(M,g)$ has a non-trivial Killing Lie subalgebra of vector fields $\frk \subset \frv=\sfGamma(TM)$. For the same reasons as spelled out in Section~\ref{sec:braidedscalar}, the Yang--Mills cyclic $L_\infty$-algebra $(V, \ell_{1},\ell_2,\ell_3, \langle- , - \rangle)$ lives in the category of $\sU \frk$-modules.

We now apply the general formalism of braided field theory from Section~\ref{sec:braidedfieldtheory} to the more elaborate example of Yang--Mills theory. A choice of a Killing Drinfel'd twist\footnote{{In $d=4$ dimensions, Yang--Mills theory is further invariant under conformal transformations. In this dimension, one may also use Drinfel'd twists generated by the conformal Killing Lie algebra of $(M,g)$. An analogous statement holds in $d=2$ dimensions, wherein Yang--Mills theory is invariant under all area-preserving diffeomorphisms.}} $\CF\in \sU\frk[[\hbar]]\otimes \sU\frk[[\hbar]]$ deforms the underlying $L_{\infty}$-algebra to a 4-term braided $L_{\infty}$-algebra $\big(V[[\hbar]], \ell_{1}^{\star}, \ell_{2}^{ \star},\ell_{3}^{ \star}\big)$ with 1-bracket $\ell_{1}^{\star}= \ell_{1}$, the 2-brackets
\begin{align}\label{eq:braidedYM2brackets}
	\begin{split} 
		\ell_{2}^\star(\lambda_1,\lambda_2)= -[\lambda_1,\lambda_2]_\frg^\star \qquad & , \qquad
		\ell_{2}^{\star} (\lambda, A)=-[\lambda,A]_\frg^\star \ , \\[4pt]
		\ell_{2}^\star(\lambda, \CA)= -[\lambda,\CA]_\frg^\star \qquad , \qquad
		\ell_{2}^\star(\lambda,\varLambda)&=-[\lambda,\varLambda]_\frg^\star \qquad , \qquad
		\ell_{2}^\star(A, \CA)= -[A,\CA]_\frg^\star \ , \\[4pt]
		\ell_{2}^\star(A_1, A_2)= - \dd \ast_\hodge  [A_1,A_2]_\frg^\star -&[A_1,\ast_\hodge  \dd A_2]_\frg^\star + (-1)^{d}\, [\ast_\hodge  \dd A_1, A_2]_\frg^\star \ , 
	\end{split}
\end{align}
and finally the non-zero 3-bracket
\begin{align}
\begin{split}
	\ell_{3}^\star(A_1&,A_2,A_3)  \\[4pt] &= -\big[A_1,\,  \ast_\hodge [A_2,A_3]_\frg^\star\, \big]_\frg^\star-\big[\sfR_k(A_2),\, \ast_\hodge [\sfR^k(A_1),A_3]_\frg^\star\,\big]_\frg^\star+(-1)^{d}\, \big[\ast_\hodge [A_1,A_2]_\frg^\star\, ,A_3\big]_\frg^\star \ ,
\end{split}
\end{align}
where $[-,-]_\frg^\star := [-,-]_\frg\circ \CF^{-1}$.

The (left) braided gauge transformation of a gauge field $A\in \Omega^1(M,\frg)$ by a gauge parameter $\lambda\in\Omega^0(M,\frg)$ follows the standard rule
\begin{align}\label{eq:braidedYMgauge}
	\delta_{\lambda}^{\br}A := \ell_1^\star(\lambda) + \ell_2^\star(\lambda,A) = \dd\lambda - [\lambda,A]_\frg^\star \ ,
\end{align}
closing as a braided Lie algebra:
\begin{align}
\big[\delta_{\lambda_1}^{\br},\delta_{\lambda_2}^{\br}\big]_\circ^\star = \delta^{\br}_{-\ell_2^\star(\lambda_1,\lambda_2)} = \delta_{[\lambda_1,\lambda_2]_\frg^\star}^{\br} \ ,
\end{align}
for all $\lambda_1,\lambda_2\in\Omega^0(M,\frg)$.

\subsubsection*{Braided Yang--Mills equations}

The braided Yang--Mills field equations turn out to have a rather surprising form. Following the prescription from Section~\ref{sec:braidedfieldtheory}, the field equations $\frF^\star_{\tA}=0$ valued in $\Omega^{d-1}(M,\frg)[[\hbar]]$ are given by
\begin{align}\label{eq:braidedYMeom}
	\begin{split}
	\begin{tabular}{|l|}\hline\\
$\displaystyle
		\frF_{\tA}^\star = \tfrac{1}{2}\, \big(\dd^{\tA}_{\star\lact}\ast_\hodge  F_{\tA}^\br + \dd^{\tA}_{\star\ract}\ast_\hodge  F_{\tA}^\br\big)  $ \\[2mm] $\displaystyle \quad \, \hspace{1cm} + \, \tfrac{1}{6} \, \big[\sfR_k(A),\, \ast_\hodge [\sfR^k (A),A]_\frg^\star \,\big]_\frg^\star
		-\tfrac{1}{12} \,  \big[A,\, \ast_\hodge [A,A]_\frg^\star\, \big]_\frg^\star + \tfrac{(-1)^{d}}{12}\, \big[ \ast_\hodge [A,A]_\frg^\star\, ,A\big]_\frg^\star 
		$\\\\\hline\end{tabular}
	\end{split}
\end{align}
where $F_{\tA}^\br=\dd A+\tfrac12\,[A,A]_\frg^\star\in\Omega^2(M,\frg)[[\hbar]]$ is the braided curvature 2-form.
Indeed, some simple algebra gives
\begin{align}\label{eq:braidedYMeomCalc}
\begin{split}
\frF^\star_{\tA} &= \ell_{1}^\star(A) -\tfrac{1}{2}\,  \ell_{2}^\star (A,A) - \tfrac{1}{3!} \, \ell_{3}^\star (A,A,A) \\[4pt]
	&= \dd \ast_\hodge  \dd A + \tfrac{1}{2} \, \dd \ast_\hodge  [A,A]_\frg^\star +\tfrac{1}{2}\,  [A,\ast_\hodge  \dd A]_\frg^\star - \tfrac{(-1)^{d}}{2}\, [\ast_\hodge  \dd A, A]_\frg^\star   \\ 
	& \quad \, +\tfrac{1}{6}\, \big[A,\, \ast_\hodge [A,A]_\frg^\star \,\big]_\frg^\star +\tfrac{1}{6} \, \big[\sfR_k(A),\, \ast_\hodge [\sfR^k (A),A]_\frg^\star \,\big]_\frg^\star -\tfrac{(-1)^{d}}{6}\,\big[ \ast_\hodge [A,A]_\frg^\star\, ,A\big]_\frg^\star \\[4pt]
	&= \dd \ast_\hodge  F_{\tA}^\br +\tfrac{1}{2}\, \big[A, \ast_\hodge  \dd A +\tfrac{1}{2} \ast_\hodge  [A,A]_\frg^\star \, \big]_\frg^\star - \tfrac{(-1)^{d}}{2} \, \big[\ast_\hodge  \dd A + \tfrac{1}{2}\ast_\hodge [A,A]_\frg^\star \, , \, A\big]_\frg^\star  \\ 
	&\quad \, -\tfrac{1}{12}\, \big[A,\, \ast_\hodge [A,A]_\frg^\star\, \big]_\frg^\star + \tfrac{1}{6}\, \big[\sfR_k(A),\, \ast_\hodge [\sfR^k (A),A]_\frg^\star \,\big]_\frg^\star + \tfrac{(-1)^{d}}{12}\, \big[ \ast_\hodge [A,A]_\frg^\star\, ,A\big]_\frg^\star \ ,
	\end{split}
\end{align}
and the result follows by identifying the curvature and left/right covariant derivatives. 

The first term of the field equations \eqref{eq:braidedYMeom} is the expected symmetrized braided covariant derivative of $\ast_\hodge  F_{\tA}^\br$, which reduces to the covariant constancy condition in the classical limit, similarly to the example of braided $BF$ theory from Section~\ref{sec:braidedBF}. However, the appearance of the last three terms is completely unexpected, since in the classical limit $\RR=1\otimes1$ and the terms conspire to cancel. Hence these terms could not have been guessed or even motivated by naively deforming the classical field equations. The braided $L_{\infty}$-algebra perspective necessitates their presence. Similar features occur in four-dimensional braided gravity~\cite{Ciric:2021rhi}, and we expect such genuinely new interaction terms to be a generic feature of gauge theories whose underlying $L_\infty$-algebra is not simply a differential graded Lie algebra.

The field equations \eqref{eq:braidedYMeom} differ drastically from those in the usual approaches to noncommutative Yang--Mills theory based on star-gauge symmetry (cf.\ Section~\ref{sec:usualNCLinfty}), wherein only the unbraided version of the first term appears (with left and right covariant derivatives coinciding). This either imposes restrictions on the gauge algebra $\frg$ or the introduction of (infinitely many) new degrees of freedom. Our braided field theory instead makes sense for any Lie algebra $\frg$ without introducing any spurious degrees of freedom. The braided field equations \eqref{eq:braidedYMeom} are furthermore real when restricted to real gauge fields, as follows from the general discussion of Section~\ref{sec:braidedfieldtheory}.

The general braided $L_{\infty}$-algebra formulation guarantees the gauge covariance of the field equations \eqref{eq:braidedYMeom} under the braided gauge transformations \eqref{eq:braidedYMgauge}. Gauge covariance in the present example reads in the expected way as
\begin{align}
	\delta_{\lambda}^{\br}\, \frF_{\tA}^\star = \ell_{2}^\star (\lambda, \frF_{\tA}^\star) = -[\lambda,\frF_{\tA}^\star]_\frg^\star \ .
\end{align}
Given the explicit form of \eqref{eq:braidedYMeom}, this is a non-trivial statement. The first term is manifestly braided covariant, as it involves covariant derivatives. On the other hand, the last three terms are not separately braided covariant objects. However, it is straightforward to show that together they do form a braided covariant combination, as guaranteed by the braided homotopy relations of Section~\ref{sec:braidedfieldtheory}. We will demonstrate how this works explicitly in Section~\ref{sec:braidedgaugeYM} below.

\subsubsection*{Braided Noether identities}

From the general prescription of Section~\ref{sec:braidedfieldtheory}, the gauge redundancies of braided Yang--Mills theory are encoded by the braided Noether identities
\begin{align}
\begin{split}
	\dsf_{\tA}^\star \frF^\star_{\tA} &:= \ell_1^\star(\frF_{\tA}^\star) - \tfrac12\,\big(\ell_2^\star(A, \frF_{\tA}^\star) - \ell_2^\star(\frF_{\tA}^\star,A)\big) \\ & \quad \, 
	-\tfrac1{12}\,\ell_1^\star\big(\ell_3^\star(A,A,A)\big) + \tfrac14\,\ell_2^\star\big(\sfR_k(A), \, \ell_2^\star(\sfR^k(A),A) \big)  \\
	& \quad \, + \tfrac1{12} \, \Big(\ell_2^\star\big(\ell_3^\star(A,A,A),A\big) - \ell_2^\star\big(A,\ell_3^\star(A,A,A)\big)\Big) \ = \ 0 \ , 
	\end{split}
\end{align}
where we used $\ell^\star_{3}\big(\CA, A_1,A_2\big)=0$. Expanding and collecting terms, this translates to the differential identity 
\begin{align}\label{eq:braidedYMNoether}
\begin{split}
\begin{tabular}{|l|}\hline\\
$\displaystyle
\tfrac{1}{2}\, \big(\dd_{\star\lact}^{\tA} \frF_{\tA}^\star + \dd_{\star\ract}^{\tA} \frF_{\tA}^\star\big) + \dd_{\star\lact}^{\tA} Z_{\tA} + \dd_{\star\ract}^{\tA} Z_{\tA} - 3 \, \dd Z_{\tA} + \tfrac{1}{4}\, \big( \dd \tilde{Z}_{\tA} - [A,\tilde{Z}_{\tA}]_\frg^\star - [\tilde{Z}_{\tA}, A]_\frg^\star \big) $ \\[2mm] $
 \hspace{1.5cm} +\,\tfrac{1}{4}\, \big[ \sfR_k(A)\, ,\, \dd \ast_\hodge  [\sfR^k(A),A]_\frg^\star + [\sfR^k(A),\ast_\hodge  \dd A]_\frg^\star - (-1)^{d}\, [\ast_\hodge  \dd \sfR^k(A), A]_\frg^\star \, \big]_\frg^\star \ = \ 0
 $\\\\\hline\end{tabular}
	\end{split}
\end{align}
where we define $Z_{\tA}$ to be the new covariant combination from the field equations \eqref{eq:braidedYMeom}:
\begin{align}
	\begin{split}
		Z_{\tA}:= \tfrac{1}{6} \, \big[\sfR_k(A),\, \ast_\hodge [\sfR^k (A),A]_\frg^\star \,\big]_\frg^\star
		 -\tfrac{1}{12} \,  \big[A,\, \ast_\hodge [A,A]_\frg^\star\, \big]_\frg^\star + \tfrac{(-1)^{d}}{12}\, \big[ \ast_\hodge [A,A]_\frg^\star\, ,A\big]_\frg^\star \ ,
	\end{split}
\end{align}
and $\tilde{Z}_{\tA}$ to be the first term of this expression:
\begin{align}
	\tilde{Z}_{\tA}:= \big[\sfR_k(A),\, \ast_\hodge [\sfR^k (A),A]_\frg^\star \,\big]_\frg^\star \ .
\end{align}
In the classical limit $Z_{\tA}$ vanishes, while the rest of the inhomogeneous terms cancel upon application of the classical  Jacobi identity for the graded Lie algebra $\big(\Omega^{\bullet}(M,\frg),[-,-]_\frg\big)$. Thus we obtain the classical Noether identity $\dd^{\tA} \frF_{\tA}=0$, as expected.

\subsubsection*{Braided Yang--Mills functional}

By twisting the classical Yang--Mills inner product \eqref{eq:YMpairing}, we obtain the (braided) cyclic structure
\begin{align} \label{eq:braidedYMpairing}
	\langle A ,\CA \rangle_\star := \int_M\,
	\Tr_\frg(A\wedge_\star \CA) \qquad \mbox{and} \qquad
	\langle \lambda,\mit\Lambda \rangle_\star := \int_M\,
	\Tr_\frg(\lambda\wedge_\star \mit\Lambda) \ .
\end{align}
We shall use this to write out and simplify the braided Yang--Mills action functional in two forms, as we did with the braided $BF$ theory of Section~\ref{sec:braidedBF}.
First we consider an arbitrary Drinfel'd twist $\CF$, and hence a braided cyclic pairing (or for a compatible twist we simply forget the strict cyclicity). 
We  define the action functional 
via the prescription of Section~\ref{sec:braidedfieldtheory}:
\begin{align}\label{eq:NCYMaction}
	\begin{split}
		S_\star(A) :\!&= \frac{1}{2}\, \langle A, \ell^\star_{1}(A) \rangle_\star - \frac{1}{3!}\, \langle A, \ell^\star_2 (A,A)\rangle_\star -\frac{1}{4!}\,  \langle A,\ell^\star_3 (A,A,A) \rangle_\star  \\[4pt] & =  \int_M\, \frac{1}{2}\, \Tr_\frg (A \wedge_\star \dd \ast_\hodge  \dd A) \\
		& \qquad \qquad -\frac{1}{6}\, \Tr_\frg \Big( A \wedge_\star\big( -\dd \ast_\hodge  [A,A]_\frg^\star -[A,\ast_\hodge  \dd A]_\frg^\star +(-1)^{d}\, [\ast_\hodge  \dd A, A]_\frg^\star\big) \, \Big) \\ 
		&\qquad \qquad -\frac{1}{24}\, \Tr_\frg \Big(A \wedge_\star \big( -\big[A,\, \ast_\hodge [A,A]_\frg^\star \,\big]_\frg^\star - \big[\sfR_k(A),\, \ast_\hodge [\sfR^k (A),A]_\frg^\star \,\big]_\frg^\star \\ & \hspace{8cm}+(-1)^{d}\,\big[ \ast_\hodge [A,A]_\frg^\star\, ,A\big]_\frg^\star\big)\, \Big) \ .
	\end{split}
\end{align}

We simplify \eqref{eq:NCYMaction}  by making the term $\Tr_\frg(F_{\tA}^\br \wedge_\star \ast_\hodge  F_{\tA}^\br)$ appear, and hence view the remaining terms as a deformation of the ``naive'' noncommutative Yang--Mills action functional. This requires the use of the $\CR$-matrix identities \eqref{eq:Rmatrixidsw}, as well as the invariance of $\Tr_\frg(-\wedge_\star-)$ under the left and right braided adjoint actions, which follows from the classical invariance order by order in the deformation parameter $\hbar$. For instance, we may derive the useful identity
\begin{align}
	\begin{split}
		\int_{M}\, \Tr_\frg \big( A \wedge_\star [A,\ast_\hodge  \dd A]_\frg^\star \, \big) &= \int_{M}\, \Tr_\frg \big([\sfR_k(A),\sfR^k(A)]_\frg^\star\, \wedge_\star \ast_\hodge  \dd A \big) \\[4pt] & = \int_{M}\, \Tr_\frg \big([A,A]_\frg^\star\, \wedge_\star \ast_\hodge  \dd A \big)  \\[4pt]
		&= \int_{M}\, \Tr_{\frg}\big(\ast_\hodge [A,A]_{\frg}^{\star} \,\wedge_\star \dd A\big) \ , 
	\end{split}
\end{align}
where in the first line we use the braided ${\rm ad}(\frg)$-invariance of $\Tr_\frg(-\wedge_\star-)$, followed by braided symmetry of $[-,-]_\frg^\star$ on 1-forms and strict symmetry of the Hodge duality operator $\ast_\hodge $ with respect to the cyclic inner product. In a completely analogous way, we find
\begin{align}
	\int_{M}\, \Tr_\frg \big( A \wedge_\star \,\big[ \ast_\hodge [A,A]_\frg^\star\, ,A\big]_\frg^\star \,\big) = - (-1)^{d} \, \int_{M}\, \Tr_\frg \big([A,\sfR_k(A)]_\frg^\star \, \wedge_\star \, \ast_\hodge  [\sfR^k(A),A]_\frg^\star\, \big) \ .
\end{align}

By treating the rest of the terms in a similar fashion, after a long but straightforward calculation we reach a simplified and more familiar form of the braided noncommutative Yang--Mills functional:
\begin{align}\label{eq:braidedYMactionnoncycl}
\begin{split}
		S_\star(A)&=  \int_{M}\,\frac{1}{2}\, \Tr_\frg \big(F_{\tA}^\br \wedge_\star\,  \ast_\hodge  F_{\tA}^\br\big)+\frac{1}{6} \, \Tr_\frg\big([A,\sfR_k(A)]_\frg^\star\, \wedge_\star \ast_\hodge  \sfR^k (  F_{\tA}^\br) \big) \\
		&\quad \qquad  \, -\frac{1}{12}\, \Big( \Tr_\frg\big(F_{\tA}^\br\wedge_\star\, \ast_\hodge  [A,A]_\frg^\star \, \big) + \Tr_\frg \big(\ast_\hodge [A,A]_\frg^\star\wedge_\star\, F_{\tA}^\br\big)  \Big) \\
		&\quad \qquad  \, + \frac{1}{24}\, \Big( \Tr_\frg \big([A,\sfR_k(A)]_\frg^\star\, \wedge_\star \, \ast_\hodge  [\sfR^k(A),A]_\frg^\star\,  \big) \\
		&\qquad\qquad \qquad \qquad \, - \Tr_\frg \big([A,\sfR_k(A)]_\frg^\star\, \wedge_\star\, \ast_\hodge  \sfR^k ( [A,A]_\frg^\star\, ) \big) \Big ) \ .
		\end{split} 
	\end{align}
The first term is the ``naive'' deformation of the classical Yang--Mills action functional \eqref{eq:YMaction}, similarly to that discussed in Section~\ref{sec:usualNCLinfty}. It is invariant under braided gauge transformations, due to braided ${\rm ad}(\frg)$-invariance of $\Tr_\frg$ and covariance of the braided curvature $F_{\tA}^\br$. The remaining terms are unexpected new interactions that represent the non-trivial deformation we uncover in the braided $L_\infty$-algebra framework. They vanish in the classical limit where $\CR=1\otimes 1$, recovering the classical action functional. This is completely in line with the extra terms we find in the braided field equations \eqref{eq:braidedYMeom}. Given that the (braided) $L_\infty$-algebra formulation of field theories is a broad generalization of Chern--Simons theory, it seems natural that Chern--Simons type terms, involving the gauge field $A$ on its own, turn up in our definition of braided field theories. Although these terms are individually non-covariant, together they form braided gauge covariant combinations. We discuss this point in more detail in Section~\ref{sec:braidedgaugeYM} below.

Restricting to compatible Drinfel'd twists, the pairing becomes a strictly cyclic structure for the braided $L_\infty$-algebra $(V[[\hbar]], \ell_{1}^\star, \ell_2^\star, \ell_{3}^\star )$. In this case, the action functional~\eqref{eq:braidedYMactionnoncycl} simplifies further to
\begin{align}\label{eq:braidedYMaction}
	\begin{split}
	\begin{tabular}{|l|}\hline\\
$\displaystyle
		S_\star(A) = \frac{1}{2}\, \int_{M}\, \Tr_\frg \big(F_{\tA}^\br \wedge_\star\,  \ast_\hodge  F_{\tA}^\br\big)+ \frac{1}{24}\, \int_M\, \Tr_\frg \big([A,\sfR_k(A)]_\frg^\star \, \wedge_\star \, \ast_\hodge  [\sfR^k(A),A]_\frg^\star\,  \big) $ \\[4mm]
		$ \displaystyle \quad \, \hspace{1cm}  - \frac{1}{24}\,\int_M\, \Tr_\frg \big([A,A]_\frg^\star \, \wedge_\star \, \ast_\hodge  [A,A]_\frg^\star\,  \big)
		$\\\\\hline\end{tabular}
	\end{split}
\end{align}
where the first three unexpected terms cancel each other by $\wedge_\star$-cyclicity under the integral, and the last term similarly simplifies.
One may derive  the braided field equations \eqref{eq:braidedYMeom} by varying this explicit action functional via similar manipulations as above. Of course, this is simply a confirmation of the general arguments from Section~\ref{sec:braidedfieldtheory}. For Hermitian twists and real gauge fields, it is real-valued.

\subsection{Braided gauge symmetry}
\label{sec:braidedgaugeYM}

It is instructive to see how gauge invariance of braided Yang--Mills theory is realised explicitly. We start with the field equations \eqref{eq:braidedYMeom}. Consider a gauge transformation of the last term:
\begin{align}
\begin{split}
	\delta_{\lambda}^{\br} \big[ \ast_\hodge [A,A]_\frg^\star\, ,A\big]_\frg^\star &= \big[ \ast_\hodge [\delta_{\lambda}^{\br}A,A]_\frg^\star\, ,A\big]_\frg^\star + \big[ \ast_\hodge [\sfR_k(A),\delta_{\sfR^k(\lambda)}^{\br}A]_\frg^\star\, ,A\big]_\frg^\star \\
	& \quad \, + \big[ \ast_\hodge  \sfR_k([A,A]_\frg^\star)\, ,\delta_{\sfR^k(\lambda)}^{\br}A\big]_\frg^\star  \\[4pt] &= \big[ \ast_\hodge [\dd \lambda,A]_\frg^\star\, ,A\big]_\frg^\star - \big[ \ast_\hodge [[\lambda,A]_\frg^\star,A]_\frg^\star\, ,A\big]_\frg^\star \\&\quad \, +\big[ \ast_\hodge [\sfR_k(A),\sfR^k(\dd\lambda)]_\frg^\star\, ,A\big]_\frg^\star  - \big[ \ast_\hodge [\sfR_k(A), [\sfR^k(\lambda),A]_\frg^\star]_\frg^\star\, ,A\big]_\frg^\star \\&\quad \, + \big[ \ast_\hodge  \sfR_k([A,A]_\frg^\star)\, ,\sfR^k (\dd \lambda) \big]_\frg^\star - \big[ \ast_\hodge \sfR_k([A,A]_\frg^\star)\, ,[\sfR^k(\lambda),A]_\frg^\star\, \big]_\frg^\star \\[4pt]
	&=2\, \big[ \ast_\hodge [\dd \lambda,A]_\frg^\star\, ,A\big]_\frg^\star - (-1)^{d}\,\big[ \dd \lambda ,\, \ast_\hodge  [A,A]_\frg^\star \, \big]_\frg^\star \\ 
	&\quad \, - \big[ \lambda , [ \ast_\hodge [A,A]_\frg^\star\, ,A]_\frg^\star\, \big]_\frg^\star \ ,
	\end{split}
\end{align} 
where in the last equality we used the linearity of the Hodge duality operator and the braided Jacobi identity to combine the terms cubic in $A$ into the sought covariant form. The breaking of covariance arises from the terms involving $\dd \lambda$. 

Similarly, one may check that
\begin{align}
\begin{split}
	\delta_{\lambda}^{\br} \big[A,\, \ast_\hodge [A,A]_\frg^\star\, \big]_\frg^\star
	&=2\, \big[ \sfR_k(A),\, \ast_\hodge [\sfR^k(\dd\lambda),A]_\frg^\star \, \big]_\frg^\star +  \big[ \dd \lambda ,\, \ast_\hodge  [A,A]_\frg^\star \, \big]_\frg^\star \\ 
	&\quad \, - \big[ \lambda , [ A,\, \ast_\hodge [A,A]_\frg^\star\, ]_\frg^\star\, \big]_\frg^\star \ ,
\end{split}
\end{align}
and
\begin{align}
\begin{split}
	\delta_{\lambda}^{\br} [\sfR_k(A),\, \ast_\hodge  [\sfR^k(A),A]_\frg^\star\,]_\frg^\star
	&=  \big[ \dd \lambda ,\, \ast_\hodge  [A,A]_\frg^\star \, \big]_\frg^\star + \big[ \sfR_k(A),\, \ast_\hodge [\sfR^k(\dd\lambda),A]_\frg^\star \, \big]_\frg^\star \\ & \quad \, - (-1)^{d}\, \big[ \ast_\hodge [\dd \lambda,A]_\frg^\star\, ,A\big]_\frg^\star
	-  \big[ \lambda , [ \sfR_k(A),\, \ast_\hodge [\sfR^k(A),A]_\frg^\star\, ]_\frg^\star\, \big]_\frg^\star \ .
	\end{split}
\end{align}
Taking the weighted sum of these three braided variations specified by \eqref{eq:braidedYMeom}, we see that the non-covariant terms cancel. Hence the combination in the field equations \eqref{eq:braidedYMeom} is indeed braided gauge covariant. 

Let us now turn to the general form \eqref{eq:braidedYMactionnoncycl} of the action functional. The second, third and fourth terms form a covariant combination by the same calculation for the braided $BF$ functional from Section~\ref{sec:braidedBF}. The last two terms also give a surprising covariant combination. The potential non-covariance arises from the terms proportional to $\dd \lambda$ in the gauge transformation, and dropping the covariant components which vanish by braided ${\rm ad}(\frg)$-invariance of $\Tr_\frg$, we find
	\begin{align}
	\begin{split}
		&\delta_{\lambda}^{\br} \, \Tr_\frg \big( [A,\sfR_k(A)]_\frg^\star \wedge_\star \, \sfR^k( \ast_\hodge  [A,A]_\frg^\star) \big)  \\[4pt]
		& \qquad \quad =  \Tr_\frg \big( [\dd \lambda,\sfR_k(A)]_\frg^\star \wedge_\star \, \sfR^k( \ast_\hodge  [A,A]_\frg^\star) \big) + \Tr_\frg \big( [\sfR_l(A),\sfR_k(A)]_\frg^\star \wedge_\star \, \sfR^k (\ast_\hodge  [\sfR^l(\dd\lambda),A]_\frg^\star) \big) \\
		& \qquad \quad \quad \,  + \Tr_\frg \big( [\sfR_l(A),\sfR_k(A)]_\frg^\star \wedge_\star \, \sfR^k (\ast_\hodge  [\sfR_m(A),\sfR^m\,\sfR^l(\dd\lambda)]_\frg^\star) \big)  \\ 
		&\qquad \quad \quad \, + \Tr_\frg \big( [\sfR_l(A),\sfR_k\,\sfR^m \,\sfR^l(\dd \lambda)]_\frg^\star \wedge_\star \, \sfR^k\, \sfR_m (\ast_\hodge  [A,A]_\frg^\star) \big)  \\[4pt]
		&\qquad \quad = \Tr_\frg \big(\sfR_k( \ast_\hodge  [A,A]_\frg^\star)  \wedge_\star \, [\sfR^k(\dd \lambda),A]_\frg^\star \, \big) \\ &\qquad \quad \quad \, + \Tr_\frg \big(\sfR_m\, \sfR^k( \ast_\hodge  [\sfR^l(\dd\lambda),A]_\frg^\star) \wedge_\star \, \sfR^m([\sfR_l(A),\sfR_k(A)]_\frg^\star) \big)  \\ 
		&\qquad \quad \quad \,  + \Tr_\frg \big( [\sfR_l(A),\sfR_k(A)]_\frg^\star \wedge_\star \, \sfR^k (\ast_\hodge  [\sfR^l(\dd\lambda), A]_\frg^\star) \big) + \Tr_\frg \big( [\dd \lambda,A]_\frg^\star \wedge_\star \, \ast_\hodge  [A,A]_\frg^\star \, \big)
		\\[4pt] &\qquad \quad = \Tr_\frg \big(\sfR_k( \ast_\hodge  [A,A]_\frg^\star)  \wedge_\star \, [\sfR^k(\dd \lambda),A]_\frg^\star \, \big)+ \Tr_\frg \big( [\dd \lambda,A]_\frg^\star \wedge_\star \, \ast_\hodge  [A,A]_\frg^\star \, \big)  \\
		& \qquad \quad \quad \, + 2\,\Tr_\frg \big(\sfR_l( \ast_\hodge  [\sfR^k(\dd\lambda),A]_\frg^\star) \wedge_\star \, [\sfR^l\,\sfR_k(A),A]_\frg^\star \, \big)  \\[4pt]
		&\qquad \quad = \Tr_\frg \big(\sfR_k (\ast_\hodge  [A,A]_\frg^\star)  \wedge_\star \, [\sfR^k(\dd \lambda),A]_\frg^\star \, \big)+ \Tr_\frg \big( [\dd \lambda,A]_\frg^\star \wedge_\star \, \ast_\hodge  [A,A]_\frg^\star \, \big)  \\
		&\qquad \quad \quad \, + 2\, \Tr_\frg \big( \ast_\hodge  [\dd\lambda,\sfR_k (A)]_\frg^\star \wedge_\star \, [\sfR^k(A),A]_\frg^\star \, \big)  \ ,
		\end{split}
	\end{align}
	where we repeatedly used braided symmetry of $\Tr_\frg(-\wedge_\star-)$ and of $[-,-]_\frg^\star$ on 1-forms along with the $\CR$-matrix identities \eqref{eq:Rmatrixidsw}. In a completely analogous way, one may check that
	\begin{align}
	\begin{split}
		\delta_{\lambda}^{\br} \, \Tr_\frg &\big( [A,\sfR_k(A)]_\frg^\star \wedge_\star \,  \ast_\hodge  [\sfR^k(A),A]_\frg^\star \, \big)  \\[4pt]
		&=2\, \Tr_\frg \big( [\dd\lambda,\sfR_k (A)]_\frg^\star \wedge_\star \,\ast_\hodge   [\sfR^k(A),A]_\frg^\star \, \big)+ \Tr_\frg \big( [\dd \lambda,A]_\frg^\star \wedge_\star \, \ast_\hodge  [A,A]_\frg^\star \, \big) \\
		&\quad \, +    \Tr_\frg \big(\sfR_k ( [A,A]_\frg^\star)  \wedge_\star \, \ast_\hodge [\sfR^k(\dd \lambda),A]_\frg^\star \, \big) \ .
		\end{split}
	\end{align}
	Thus the non-covariant terms cancel by strict symmetry of the Hodge duality operator, ensuring gauge invariance of the action functional as expected from the general theory of Section~\ref{sec:braidedLinfty}.

\subsection{Global braided symmetry}
\label{sec:globalsymYM}

Similarly to the example of noncommutative scalar theory from Section~\ref{sec:globalsymscalar}, braided Yang--Mills theory is invariant under the finite-dimensional braided Lie algebra of twisted Killing symmetries $(\frk,[-,-]_\frk^\star)$, acting via the braided Lie derivative. These are incorporated into the braided $L_\infty$-algebra framework of Section~\ref{sec:braidedYMtheory} by extending the degree~$0$ and~$3$ homogeneous subspaces of the underlying cochain complex of Section~\ref{sec:Yang-Millstheory} to get\footnote{The caveat discussed in Footnote~\ref{fn:globalsym} should also be borne in mind here.}
\begin{align}
V^0 = \frk\times\Omega^0(M,\frg) \ , \quad V^1 = \Omega^1(M,\frg) \ , \quad V^2 = \Omega^{d-1}(M,\frg) \ , \quad V^3 = \frk^*\times\Omega^d(M,\frg) \ ,
\end{align}
with the differential extended as
\begin{align}
\ell_1^\star(\xi,\lambda) = \dd\lambda \ , \quad \ell_1^\star(A) = \dd\ast_\hodge \dd A \qquad \mbox{and} \qquad \ell_1^\star(\CA) = (0,\dd \CA) \ ,
\end{align}
and by modifying the $2$-brackets of \eqref{eq:braidedYM2brackets} as
\begin{align}
\begin{split}
\ell_2^\star\big((\xi_1,\lambda_1)\,,\,(\xi_2,\lambda_2)\big) = \big(-[\xi_1,\xi_2]_\frk^\star & \ , \, -\LL_{\xi_1}^\star\lambda_2 + \LL^\star_{\sfR_k(\xi_2)}\sfR^k(\lambda_1)-[\lambda_1,\lambda_2]_\frg^\star\big) \ , \\[4pt]
\ell_2^\star\big((\xi,\lambda)\,,\,A\big) = \LL_\xi^\star A - [\lambda,A]_\frg^\star \qquad & , \qquad \ell_2^\star\big((\xi,\lambda)\,,\,\CA\big) = \LL_\xi^\star\CA - [\lambda,\CA]_\frg^\star \ , \\[4pt]
\ell_2^\star\big((\xi,\lambda)\,,\,(\xi^\dual,\varLambda)\big) = \big(\xi\star\xi^\dual \,,\,\LL_\xi^\star\varLambda -[&\lambda,\varLambda]_\frg^\star \big) \qquad , \qquad \ell_2^\star(A,\CA) = \big(0,-[A,\CA]_\frg^\star\big) \ ,
\end{split}
\end{align}
while the brackets $\ell_2^\star(A_1,A_2)$ and $\ell_3^\star(A_1,A_2,A_3)$ are unchanged.
The cyclic structure \eqref{eq:braidedYMpairing} is naturally extended as
\begin{align}
\langle A ,\CA \rangle_\star = \int_M\,
	\Tr_\frg(A\wedge_\star \CA) \qquad \mbox{and} \qquad
\langle (\xi,\lambda)\,,\,(\xi^\dual,\varLambda)\rangle_\star = \langle \xi,\xi^\dual \rangle_\star + \int_M \, \Tr_\frg(\lambda\wedge_\star\varLambda) \ .
\end{align}

The extended braided transformations of the gauge fields
\begin{align}
\delta^{\br}_{(\xi,\lambda)}A = \ell_1^\star(\xi,\lambda) + \ell_2^\star\big((\xi,\lambda)\,,\,A\big) = \LL_\xi^\star A + \dd\lambda - [\lambda,A]_\frg^\star
\end{align}
satisfy the gauge algebra
\begin{align}
\big[\delta_{(\xi_1,\lambda_1)}^{\br},\delta_{(\xi_2,\lambda_2)}^{\br}\big]_\circ^\star = \delta_{-\ell_2^\star((\xi_1,\lambda_1),(\xi_2,\lambda_2))}^{\br} = \delta^{\br}_{([\xi_1,\xi_2]_\frk^\star,\LL_{\xi_1}^\star\lambda_2 - \LL^\star_{\sfR_k(\xi_2)}\sfR^k(\lambda_1)+[\lambda_1,\lambda_2]_\frg^\star)} \ ,
\end{align}
with bracket of the braided semi-direct product Lie algebra
\begin{align}
\frk[[\hbar]]\ltimes_\star \Omega^0(M,\frg)[[\hbar]] \ .
\end{align}
The braided noncommutative Yang--Mills equations \eqref{eq:braidedYMeom} are correspondingly covariant,
\begin{align}
\delta^{\br}_{(\xi,\lambda)}\frF_{\tA}^\star = \ell_2^\star\big((\xi,\lambda)\,,\,\frF_{\tA}^\star\big) = \LL_\xi^\star \frF_{\tA}^\star - [\lambda,\frF_{\tA}^\star]_\frg^\star \ ,
\end{align}
and likewise the action functional \eqref{eq:braidedYMactionnoncycl} is invariant, $\delta_{(\xi,\lambda)}^{\br} S_\star(A)=0$.
These extended properties also follow directly from the closure and derivation properties of the braided Lie derivative $\LL_\xi^\star$.

\end{document}